\documentclass[11pt]{article}

\usepackage{amsthm,amsmath,amssymb,bbm,latexsym, amsfonts}
\usepackage{natbib}
\usepackage{multirow}
\usepackage[pdftex]{graphicx}
\usepackage{subfigure}
\usepackage{makecell}
\usepackage{booktabs}
\usepackage{mathtools}
\DeclarePairedDelimiter\ceil{\lceil}{\rceil}

\usepackage{array}
\usepackage{url}
\usepackage{algorithm}
\usepackage{algorithmic}
\usepackage{bm}
\usepackage{wrapfig}
\usepackage{lipsum}
\usepackage{mathrsfs}
\usepackage{titling}
\usepackage{algorithm}

\usepackage[usenames,dvipsnames,svgnames,table]{xcolor}
\usepackage[colorlinks,
linkcolor=red,
anchorcolor=blue,
citecolor=blue
]{hyperref}

\newcommand{\sgn}{\mathop{\mathrm{sign}}}

\usepackage{smile}

\newcommand*{\supp}{\mathrm{supp}}
\newcommand*{\var}{\textnormal{var}}

\newcommand{\e}{\mathbb{E}}
\newcommand{\nn}{\nonumber}

\newcommand \hbP{\widehat{\bTheta}}
\newcommand{\extend}[2]{\mathbf{E}^{#1}_{#2}}

\def\##1\#{\begin{align}#1\end{align}}
\def\$#1\${\begin{align*}#1\end{align*}}

\newcommand {\vecc}{\textnormal {vec}}
\def\T{\intercal} 
\def\sn{\sum_{i=1}^n}
\def\Sb{\mathbf{S}}

\newcommand{\HH}{\cH}
\newcommand{\TT}{\cT}
\newcommand{\F}{\mathrm{F}}
\newcommand{\cov}{\mathrm{cov}}
\newcommand{\FDP}{{\rm FDP}}

\newcommand{\sam}{{\rm sam}}
\newcommand{\ol}{\overline}
\newcommand{\SUM}{\sum_{i=1}^N}
 
\newcommand{\wt}{\widetilde}

\newcommand{\Rom}[1]{\text{\uppercase\expandafter{\romannumeral #1\relax}}}

\usepackage{geometry}
\newcommand{\comment}[1]{\textcolor{red}{#1}}

\usepackage{enumitem}

\begin{document}

\title{ User-Friendly Covariance Estimation   \\ for  Heavy-Tailed Distributions}

\author{Yuan Ke\thanks{Department of Statistics, University of Georgia, Athens, GA 30602, USA. E-mail: \href{mailto:yuan.ke@uga.edu}{\textsf{yuan.ke@uga.edu}}.},~~Stanislav Minsker\thanks{Department of Mathematics, University of Southern California, Los Angeles, CA 90089, USA. E-mail: \href{mailto:minsker@usc.edu}{\textsf{minsker@usc.edu}}. Supported in part by NSF Grant DMS-1712956.},~~Zhao Ren\thanks{Department of Statistics, University of Pittsburgh, Pittsburgh, PA 15260, USA. E-mail: \href{mailto:zren@pitt.edu}{\textsf{zren@pitt.edu}}. Supported in part by NSF Grant DMS-1812030.},~~Qiang Sun\thanks{Department of Statistical Sciences, University of Toronto, Toronto, ON M5S 3G3, Canada. E-mail: \href{mailto:qsun@utstat.toronto.edu}{\textsf{qsun@utstat.toronto.edu}}. Supported in part by a Connaught Award and NSERC Grant RGPIN-2018-06484.}~~and~Wen-Xin Zhou\thanks{Department of Mathematics, University of California, San Diego, La Jolla, CA 92093, USA. E-mail:  \href{mailto:wez243@ucsd.edu}{\textsf{wez243@ucsd.edu}}. Supported in part by NSF Grant DMS-1811376.} }

\date{}
\maketitle

\vspace{-0.5in}
 
\begin{abstract}
We offer a survey of recent results on covariance estimation for heavy-tailed distributions. By unifying ideas scattered in the literature, we propose user-friendly methods that facilitate practical implementation.
Specifically, we introduce element-wise and spectrum-wise truncation operators, as well as their $M$-estimator counterparts, to robustify the sample covariance matrix. Different from the classical notion of robustness that is characterized by the breakdown property, we focus on the tail robustness which is evidenced by the connection between nonasymptotic deviation and confidence level. 
The key observation is that the estimators needs to  adapt to the sample size, dimensionality of the data and the noise level to achieve optimal tradeoff between bias and robustness. Furthermore, to facilitate their practical use, we propose data-driven procedures that automatically calibrate the tuning parameters.  We demonstrate their applications to a series of structured models in high dimensions, including the bandable and low-rank covariance matrices and sparse precision matrices. Numerical studies lend strong support to the proposed methods.
\end{abstract}

\noindent
{\bf Keywords}: Covariance estimation, heavy-tailed data, $M$-estimation, nonasymptotics, tail robustness, truncation.

\section{Introduction}
\label{sec:1}

Covariance matrices are important in multivariate statistics. The estimation of covariance matrices, therefore, serves as a building block for many important statistical methods, including   principal component analysis,   discriminant analysis,  clustering analysis and regression analysis, among many others.
Recently, the problem of estimation of structured large covariance matrices, such as bandable, sparse and low-rank matrices, has attracted ever-growing attention in statistics and machine learning \citep{BL2008a, BL2008b, CRZ2016, FLL2016}. This problem has many applications, ranging from functional magnetic resonance imaging (fMRI), analysis of gene expression arrays to risk management and portfolio allocation. 

Theoretical properties of large covariance estimators discussed in the literature often hinge heavily on the Gaussian or sub-Gaussian\footnote{A random variable $Z$ is said to have sub-Gaussian tails if there exists constants $c_1$ and $c_2$  such that $\PP(|Z-\mathbb{E}Z|>t)\leq c_1\exp(-c_2t^2)$ for all $t \geq 0$.} assumptions \citep{V2012}. See, for example, Theorem~1 of  \cite{BL2008a}.  Such an assumption is typically very restrictive in practice. For example, a recent  fMRI study by \cite{Eklund2016} reported  that most of the common software packages, such as SPM  and FSL,  for fMRI analysis can result in inflated false-positive rates up to 70\% under 5\% nominal levels, and questioned a number of fMRI studies among approximately 40,000 studies according to PubMed. Their results suggested that 
\begin{quote}
{\it The principal cause of the invalid cluster inferences is spatial autocorrelation functions that do not follow the assumed Gaussian shape.}
\end{quote}
\cite{Eklund2016} plotted  the empirical  versus theoretical spatial autocorrelation functions 
for several datasets. The empirical autocorrelation functions have much heavier tails compared to their theoretical counterparts under the commonly used assumption of a Gaussian random field, which causes the failure of fMRI inferences \citep{Eklund2016}.
Similar phenomenon has also been discovered in genomic studies \citep{LHGY2003, PH2005} and in quantitative finance \citep{C2001}.  It is therefore imperative to develop robust inferential procedures that are less sensitive to the distributional assumptions.

We are interested in constructing estimators that admit tight nonasymptotic deviation guarantees under weak moment assumptions.  Heavy-tailed distribution is a viable model for data contaminated by outliers that are typically encountered in applications. Due to heavy tailedness, the probability that some observations are sampled far away from the ``true'' parameter of the population is non-negligible. We refer to these outlying data points as stochastic outliers. 
A procedure that is robust against such outliers,  evidenced by its better finite-sample performance than a non-robust method, is called a {\it tail-robust} procedure.
In this paper, by unifying ideas scattered in the literature, we provide a unified framework for constructing user-friendly tail-robust covariance estimators. Specifically, we propose element-wise and spectrum-wise truncation operators, as well as their $M$-estimator counterparts, with adaptively chosen robustification parameters. Theoretically, we establish nonasymptotic deviation bounds and demonstrate that the robustification parameters should adapt to the sample size, dimensionality and noise level for optimal tradeoff between  bias and robustness. 
Our goal is to obtain estimators that are computationally efficient and  easily implementable in practice. To this end, we propose data-driven schemes to calibrate the tuning parameters, making our proposal user-friendly.  Finally, we discuss applications to several structured models in high dimensions, including bandable matrices, low-rank covariance matrices as well as sparse precision matrices. 
In the supplementary material, we further consider  robust covariance estimation and inference under factor models, which might be of independent interest.

 Our definition of robustness is different from the conventional perspective under Huber's $\epsilon$-contamination model \citep{H1964}, where the focus has been on developing robust procedures with a high breakdown point. The breakdown point \citep{hampel1971general}  of an estimator is defined (informally) as  the largest proportion of outliers in the data for which the estimator remains stable. 
Since the seminal work of \cite{tukey1975mathematics}, a number of depth-based robust procedures have been developed; see, for example, the papers by \cite{liu1990notion}, \cite{zuo2000general}, \cite{mizera2002depth} and \cite{ salibian2002bootrapping}, among others. Another line of work focuses on robust and resistant $M$-estimators, including  the least median of squares and least trimmed squares \citep{rousseeuw1984least}, the S-estimator  \citep{rousseeuw1984robust} and the MM-estimator \citep{yohai1987high}.   We refer to \cite{portnoy2000robust} for  a literature review on classical robust statistics, and to \cite{chen2018robust} for recent developments on nonasymptotic analysis under contamination models.

The rest of the paper is organized as follows. We start with a motivating example in Section \ref{sec:2}, which reveals the downsides of the sample covariance matrix. In Section \ref{sec:general}, we introduce two types of generic robust covariance estimators and establish their deviation bounds under different norms of interest. 
The finite-sample performance of the proposed estimators, both element-wise and spectrum-wise, depends on a careful tuning of the robustification parameter, which should adapt to the noise level for bias-robustness tradeoff. 
We also discuss the median-of-means estimator, which is virtually tuning-free at the cost of slightly stronger assumptions.
For practical implementation, in Section \ref{sec:3.5} we propose a data-driven scheme to choose the key tuning parameters. Section \ref{sec:4} presents various applications to estimating structured covariance and precision matrices. Numerical studies are provided in Section \ref{sec:5}. We close this paper with a discussion in Section \ref{sec:6}.

\subsection{Overview of the previous work}
\label{sec.literature}

In the past several decades, there has been a surge of work on robust covariance estimation in the absence of normality. Examples include  the Minimum Covariance Determinant (MCD) estimator, the Minimum Volume Ellipsoid (MVE) estimator, Maronna's \citep{M1976}  and Tyler's \citep{T1987} $M$-estimators of multivariate scatter matrices. We refer to \cite{HRvA2008} for a comprehensive review.
  Asymptotic properties of these methods have been established for the family of elliptically symmetric distributions;  see, for example, \cite{D1992}, \cite{BDJ1993} and \cite{ZCS2016}, among others. However,  the aforementioned estimators either rely on  parametric assumptions, or impose a shape constraint on the sampling distribution. 
Under a general setting where neither of these assumptions are made, robust covariance estimation remains a challenging problem.

The work of \cite{C2012} triggered a growing interest in developing tail-robust estimators,  
which are characterized by tight nonasymptotic deviation analysis, rather than mean squared errors.
The current state-of-the-art methods for covariance estimation with heavy-tailed data include those of \cite{C2016}, \cite{M2016}, \cite{MW2018}, \cite{ABFL2018}, \cite{MZ2018}. From a spectrum-wise perspective, \cite{C2016} constructed a robust estimator of the Gram and covariance matrices of a random vector $\bX \in \RR^d$ via estimating the quadratic forms $\EE \langle \bu , \bX\rangle^2$ uniformly over the unit sphere in $\RR^d$, and proved error bounds under the operator norm. More recently, \cite{MZ2018} proposed a different robust covariance estimator that admits tight deviation bounds under the finite kurtosis condition. Both constructions, however,  involve brute-force search over  every direction in a $d$-dimensional $\varepsilon$-net, and thus are computationally intractable.  From an element-wise perspective, \cite{ABFL2018} combined robust estimates of the first and second moments to obtain variance estimators. In practice, three potential drawbacks
of this approach are: (i) the accumulated error consists of those from estimating the first and second moments, which may be significant; (ii) the diagonal variance estimators are not necessarily positive and therefore additional adjustments are required; and (iii) using the cross-validation to calibrate a total number of $O(d^2)$ tuning parameters is computationally expensive.

Building on the ideas of \cite{M2016} and \cite{ABFL2018},  we  propose  user-friendly tail-robust covariance estimators that enjoy desirable finite-sample deviation bounds under weak moment conditions. The constructed estimators only involve simple truncation techniques and are  computationally friendly. 
Through a novel data-driven tuning scheme, we are able to efficiently compute these robust estimators for large-scale problems in practice. These two points distinguish our work from the literature on the topic.  
The proposed robust procedures  serve as building blocks for estimating large structured covariance and precision matrices, and we illustrate their broad applicability in a series of problems.

\subsection{Notation}
\label{sec.notation}

We  adopt the following notation throughout the paper. 
For any $0\leq r ,s\leq \infty $ and a $d\times d$ matrix $\Ab= (A_{k\ell})_{1\leq k,\ell\leq d}$, we define the max norm $\| \Ab \|_{\max} = \max_{1\leq k,\ell\leq d} |A_{k \ell}|$, the Frobenius norm $\|\Ab \|_{\text F}=(\sum_{1\leq k,\ell\leq d}A_{k\ell}^2)^{1/2}$ and the operator norm 
\[
\| \Ab \|_{r,s} = \sup_{ \bu = (u_1 ,\ldots, u_d)^\T :  \| \bu \|_r = 1} \| \Ab \bu \|_s ,
\] 
where $ \| \bu \|_r^r = \sum_{k=1}^d |u_k |^r$ for $r\in (0,\infty)$, $\| \bu \|_0=\sum_{k=1}^d I(|u_k| \neq 0)$ and $\| \bu \|_{\infty} =\max_{1\leq k\leq d} |u_k|$.
In particular, it holds $ \| \Ab \|_{1,1} =  \max_{1\leq \ell \leq d} \sum_{k=1}^d |A_{k\ell}|$ and $\| \Ab \|_{\infty , \infty} =  \max_{1\leq k \leq d} \sum_{\ell=1}^d |A_{k\ell}|$. Moreover, we write $\| \Ab \|_2 := \| \Ab \|_{2,2}$ for the spectral norm and use ${\rm r}(\Ab) ={\rm tr}(\Ab) /\| \Ab \|_2 $ to denote the effective rank of a nonnegative definite matrix $\Ab$, where ${\rm tr}(\Ab)= \sum_{k=1}^d A_{kk}$ is the trace of $\Ab$.
When $\Ab$ is symmetric, it is well known that $\| \Ab \|_2 = \max_{1\leq k\leq d} | \lambda_k(\Ab)|$ where $\lambda_1(\Ab)\geq \lambda_2(\Ab)\ldots\geq\lambda_d(\Ab)$ are the eigenvalues of $\Ab$. 
For any matrix $\Ab \in \RR^{d\times d}$ and an index set $J \subseteq  \{ 1,\ldots, d\}^2$, we use $J^{\rm c}$ to denote the complement of $J$, and $\Ab_J$ to denote the submatrix of $\Ab$ with entries indexed by $J$. For a real-valued random variable $X$, let ${\rm kurt}(X)$ be the kurtosis of $X$, defined as ${\rm kurt}(X) = \EE (X -\mu)^4 / \sigma^4$, where $\mu = \EE X $ and $\sigma^2 = \var(X)$.

\section{Motivating example: a challenge of heavy-tailedness} \label{sec:2}

Suppose that we observe a sample of independent and identically distributed (i.i.d.) copies $\bX_1, \ldots, \bX_n$ of a random vector $\bX=(X_1,\ldots, X_d)^\T \in \RR^d$ with mean  $\bmu$ and covariance matrix $\bSigma = (\sigma_{k\ell})_{1\leq k,\ell \leq d }$.  To assess the difficulty of mean and covariance estimation for heavy-tailed distributions, we first provide a lower bound for the deviation of the empirical mean under the $\ell_\infty$-norm in $\RR^d$.
	
	\begin{proposition} \label{thm:element-wise:lower}
		For any $\sigma>0$ and $0<\delta<(2e)^{-1}$,
		there exists a distribution in $\RR^d$ with mean $\bmu$ and covariance matrix $\sigma^2 \Ib_d$ such that the empirical mean $\bar{\bX}=(1/n)\sum_{i=1}^n \bX_i$ of i.i.d. observations $\bX_1, \ldots, \bX_n$ from  this distribution satisfies, with probability at least $\delta$, 
		\#
		\| \bar{\bX} - \bmu \|_{\infty}   \geq  \sigma  \sqrt{\frac{d}{n \delta}}\bigg(1-\frac{2e\delta}{n}\bigg)^{(n-1)/2}. \label{poly.upper_lower}
		\#
	\end{proposition}

The above proposition  is a multivariate extension of Proposition~6.2 of \cite{C2012}. It provides a lower bound under the $\ell_\infty$-norm for estimating a mean vector via the empirical mean.  On the other hand, combining the union bound with Chebyshev's inequality, we  obtain  that with probability at least $1-\delta$,
	\#
	\| \bar{\bX} - \bmu \|_{\infty}   \leq    \sigma \sqrt{\frac{d }{n \delta}}.  \nn
	\#
Together, this upper bound and inequality \eqref{poly.upper_lower} show that the worst case deviations of the empirical means grow polynomially in $1/\delta$  under the $\ell_\infty$-norm in the presence of heavy-tailed distributions. As we will see later, a more robust estimator can achieve an exponential-type deviation bound under weak moment conditions.

To demonstrate the practical implications of Proposition \ref{thm:element-wise:lower}, we perform a toy numerical study on covariance matrix estimation. Let $\bX_1, \ldots, \bX_n$ be i.i.d. copies of $\bX= (X_1,\ldots, X_d)\in \RR^d$, where $X_k$'s are independent and have centered Gamma$(3,1)$ distribution so that $\bmu=\bf{0}$ and $\bSigma = 3\Ib_d$. We compare the performance of three methods: the sample covariance matrix, the element-wise truncated covariance matrix and the spectrum-wise  truncated covariance matrix. The latter two are tail-robust covariance estimators that will be introduced in Sections~\ref{sec2.1} and \ref{sec2.2} respectively. We report the estimation errors under the max norm. We take $n=200$ and let $d$ increase from $50$ to $500$ with a step size of $50$. The results are based on $50$ simulations. Figure \ref{Fig_1} displays the average  (line) and the spread (dots) of estimation errors for each method as the dimension increases. We see that the sample covariance estimator has not only the largest average error but also the highest variability in all the settings. This example demonstrates that the sample covariance matrix suffers from poor finite-sample performance when data are heavy-tailed.
\begin{figure}[t]
 \centering
 \includegraphics[width=5 in]{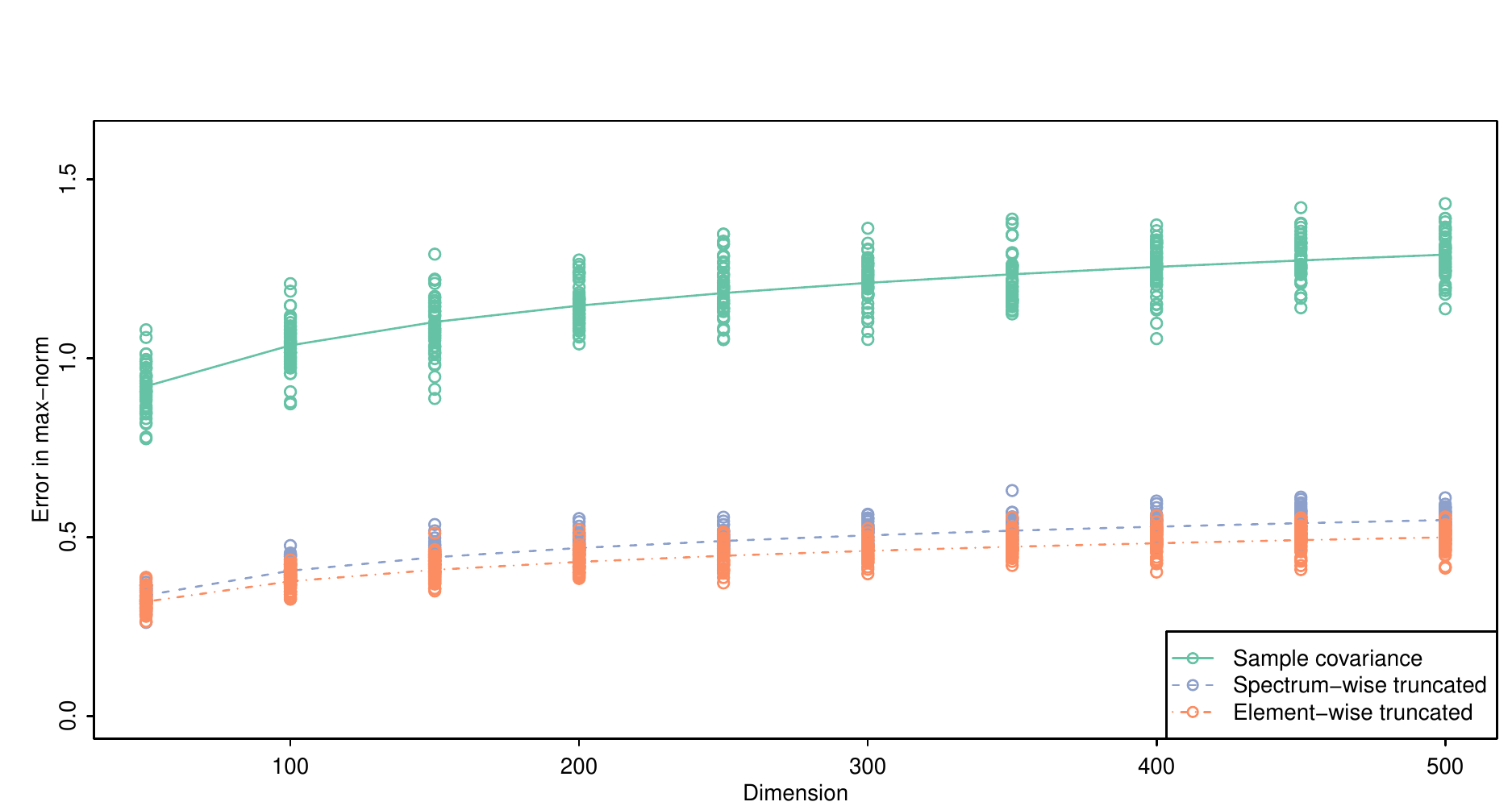}
  \caption{Plots of estimation error under max norm versus dimension.}
 \label{Fig_1}
\end{figure}

\section{Tail-robust covariance estimation}
\label{sec:general}

\subsection{Element-wise truncated estimator}
\label{sec2.1}

We consider the same setting as in the previous section. 
For mean estimation, the suboptimality of deviations of $\bar \bX = (\bar X_1,\ldots, \bar X_d)^\T$ under $\ell_\infty$-norm is due to the fact that the tail probability of $|\bar X_k -  \mu_k |$ decays only polynomially. A simple yet natural idea for improvement is to truncate the data to eliminate outliers introduced by heavy-tailed noises,  so that  each entry of the resulting estimator has sub-Gaussian tails.
To implement this idea, we introduce the following truncation operator, which is closely related to the Huber loss.

\begin{definition}[Truncation operator] \label{truncated.def}
Let $\psi_\tau(\cdot)$ be a truncation operator given by
\#
	\psi_\tau(u) =  ( | u |\wedge \tau  ) \sgn (u), \ \ u\in \RR,\label{psi.def}
\#
where the truncation parameter $\tau >0$ is also referred to as the {\it robustification parameter} that trades off bias for robustness.
\end{definition}

As an illustration, we assume that $\bmu={\bf 0}$ whence $\bSigma= \EE(\bX \bX^\intercal )$.
We  apply the  truncation operator above to each entry of   $\bX_i\bX_i^\T$, and then take the average to obtain
\$
	 \hat{\sigma}^{\TT}_{0,k\ell} = \frac{1}{n} \sn \psi_{\tau_{k\ell}} (X_{ik} X_{i\ell}  ), \ \ 1\leq k,\ell \leq d, 
\$
where $\tau_{k\ell}>0$ are robustification parameters.
When the mean vector $\bmu$ is unspecified, a straightforward approach is to first estimate the mean vector using existing robust methods  \citep{M2015, LM2017}, and then to employ $\hat{\sigma}^{\TT}_{0,k\ell}$ as robust estimates of the second moments. 
Estimating the first and second moments separately will unavoidably introduce additional tuning parameters, which increases both  statistical variability and computational complexity.
In what follows, we propose to use the pairwise difference approach, which is free of mean estimation.  To the best of our knowledge, the difference-based techniques can be traced back to \cite{rice1984bandwidth} and \cite{hall1990asymptotically} in the context of bandwidth selection and variance estimation for nonparametric regression.

 Let $N := n(n-1)/2$ and define the paired data
\#
	\{ \bY_1, \bY_2, \ldots, \bY_N \} = \{ \bX_1 - \bX_2, \bX_1 - \bX_3, \ldots, \bX_{n-1} - \bX_n \}, \label{paired.data}
\#
which are identically distributed from a random vector $\bY$ with mean $ \textbf{0}$ and covariance matrix $\cov(\bY) = 2 \bSigma$.
It is easy to check that the sample covariance matrix,  $\hat{\bSigma}^{\sam}  = (1/n) \sn (\bX_i - \bar{\bX})(\bX_i - \bar{\bX})^\T$ with $\bar{\bX} = (1/n)\sn \bX_i$, can be expressed as a U-statistic 
\$
\hat{\bSigma}^{\sam} =  \frac{1}{N} \SUM \bY_i \bY_i^\intercal/2. 
\$

Following the argument from the last section,  we apply the truncation operator $\psi_\tau$ to $\bY_{i}\bY_i^\T/2$ entry-wise, and then take the average to obtain
\$
	 \hat{\sigma}_{1,k\ell}^{\TT} = \frac{1}{N} \SUM \psi_{\tau_{k\ell}} (Y_{ik} Y_{i\ell} /2 ),  \ \  1\leq k , \ell \leq d. 
\$
Concatenating these estimators, we define the element-wise truncated  covariance matrix estimator via
\#
	\hat{\bSigma}_{1}^{\TT}  = \hat \bSigma_1^{\TT}(\bGamma) = (  \hat{\sigma}^{\TT}_{1,k\ell} )_{1\leq k,\ell \leq d}  , \label{element-wise.censored.est}
\#
where $\bGamma  = (\tau_{k\ell})_{1\leq k,\ell \leq d }$ is a symmetric matrix of parameters.
 $\hat{\bSigma}_{1}^{\TT}$ can be viewed as a truncated version of the  sample covariance matrix $\widehat\bSigma^\sam$. 
{We assume that $n\geq 2, ~d\geq 1$ and define $m= \left \lfloor{  n/2  }\right \rfloor$, the largest integer not exceeding $n/2$.}
Moreover, let $\Vb = (v_{k\ell})_{1\leq k,\ell \leq d} $ be a symmetric $d\times d$ matrix such that 
\[
	v_{k\ell}^2 =  \EE (Y_{1k} Y_{1\ell} /2)^2  =  \EE \{ (X_{1k} - X_{2k})(X_{1\ell} - X_{2\ell}) \}^2/4.
\]

\begin{theorem} \label{thm:element-wise}
For any $0<\delta <1$, the estimator $\hat \bSigma_1^{\TT} = \hat \bSigma^{\TT}_1( \bGamma)$ defined in \eqref{element-wise.censored.est}  with 
\#
	 \bGamma  = \sqrt{{m}/({2\log d + \log \delta^{-1}  })} \,\Vb   \label{tuning1}
\# 
satisfies
\#
	\PP  \Bigg(  \| \hat{\bSigma}^{\TT}_1 - \bSigma \|_{\max}   \geq  2 \|\Vb \|_{\max}  \sqrt{\frac{2\log d + \log \delta^{-1} }{m }   }   \Bigg)  \leq 2 \delta . \label{sigma1.bound}
\#
\end{theorem}

Theorem~\ref{thm:element-wise} indicates that, with properly calibrated parameter matrix $\bGamma$, the resulting covariance matrix estimator achieves  element-wise tail  robustness against heavy-tailed distributions: provided the fourth moments are bounded, each entry of $\widehat \bSigma_1^{\TT}$ concentrates tightly around its mean so that the maximum error scales as $\sqrt{ \log(d)/n} + \sqrt{\log(\delta^{-1})/n}$. Element-wise, we are able to accurately estimate $\bSigma$ at high confidence levels under the constraint that $\log(d)/n$ is small.
Implicitly, the dimension $d=d(n)$ is regarded as a function of $n$, and we shall use array asymptotics ``$n,d \to \infty$" to characterize large sample behaviors. 
The finite sample performance, on the other hand, is characterized via nonasymptotic probabilistic bounds with explicit dependence on $n$ and $d$.

\begin{remark}   \label{remark1}
It is worth mentioning that the estimator given in \eqref{element-wise.censored.est} and \eqref{tuning1} is not a genuine sub-Gaussian estimator, in a sense that it depends on the confidence level $1-\delta$ at which one aims to control the error. More precisely, following the terminology used by \cite{DLLO2016}, it is called a $\delta$-dependent sub-Gaussian estimator (under the max norm). Estimators of a similar type include those of \cite{C2012}, \cite{M2015}, \cite{BJL2015}, \cite{HS2016}, \cite{M2016} and \cite{ABFL2018}, among  others. For univariate mean estimation, \cite{DLLO2016} proposed multiple-$\delta$ mean estimators that satisfy exponential-type concentration bounds uniformly over $\delta \in [\delta_{\min}, 1)$. The idea is to combine a sequence of $\delta$-dependent estimators in a way very similar to Lepski's  method \citep{Lep1990}.
\end{remark}

 \begin{remark}  \label{remark2}
Since the element-wise truncated estimator is obtained by treating each covariance $\sigma_{k\ell}$ separately as a univariate parameter, the problem is equivalent to estimation of a large vector given by the concatenation of the columns of $\bSigma$.
These type of results are particularly useful for proving upper bounds for sparse covariance and precision estimators  in high dimensions; see Section \ref{sec:4}. Integrated with $\ell_\infty$-type perturbation bounds, it can also be applied to principle component analysis and factor analysis for heavy-tailed data \citep{FSZZ2018}. However, when dealing with large covariance matrices with bandable or low-rank structure, controlling the estimation error under spectral norm is arguably more relevant. A natural idea is then to truncate the spectrum of the sample covariance matrix instead of its entries, which leads to the spectrum-wise truncated estimator defined in the following section.
\end{remark}

\subsection{Spectrum-wise  truncated estimator}
\label{sec2.2}

In this section, we propose and study a covariance estimator that is tail-robust in the spectral norm.
To this end,  we directly apply the truncation operator to matrices in their spectrum domain. 
 We need the following standard definition of a matrix functional.

\begin{definition}[Matrix functional] \label{def.matrix.functional}
Given a real-valued function $f$ defined on $\RR$ and a symmetric $\Ab\in \RR^{K\times K}$ with eigenvalue decomposition $\Ab=\Ub\bLambda \Ub^\T$ such that $\bLambda=\textnormal{diag}(\lambda_1,\ldots, \lambda_K)$, $f(\Ab)$ is defined as $f(\Ab)=\Ub f(\bLambda)\Ub^\T$, where $f(\bLambda)=\textnormal{diag}  (f(\lambda_1), \ldots,f(\lambda_K)  )$.
\end{definition}

Following the same rational as in the previous section, we propose a spectrum-wise truncated  covariance estimator based on the pairwise difference approach:
\#
  \hat\bSigma_{2}^{\TT} = \hat \bSigma^{\TT}_{2 }(\tau)  =\frac{ 1}{N} \SUM \psi_\tau( \bY_i \bY_i^\T/2 ) , \label{U.cov.estimator}
\#
where $\bY_i$ are given in \eqref{paired.data}. Note that $ \bY_i \bY_i^\T/2$ is a rank-one matrix with eigenvalue $\|   \bY_i \|_2^2/2$ and the corresponding eigenvector $ \bY_i/\| \bY_i \|_2$. By Definition~\ref{def.matrix.functional}, $ \widehat\bSigma_2^{\TT}$ can be rewritten as 
\$
& \frac{1}{N} \SUM  \psi_\tau\bigg( \frac{ 1}{2} \|   \bY_i \|_2^2\bigg) \frac{\bY_i \bY_i^\T}{\|   \bY_i \|_2^2 } \\
&=\frac{ 1}{ \binom{n}{2}}\sum_{1\leq  i <  j\leq  n}  \psi_\tau\bigg(\frac{1}{2} \| \bX_i -\bX_j \|_2^2 \bigg) \frac{(\bX_i-\bX_j)(\bX_i-\bX_j)^\T}{ \|  \bX_i-\bX_j  \|_2^2}.
\$
This alternative expression renders the computation almost effortless. The following theorem provides an exponential-type concentration inequality for $\widehat\bSigma_{2}^{\TT}$ under operator norm, which is a useful complement to the Remark 8 of \cite{M2016}. Similarly to Theorem \ref{thm:element-wise}, our next result shows that  
$\widehat\bSigma_2^{\TT}$ achieves exponential-type  concentration in the operator norm for heavy-tailed data with finite operator-wise fourth moment, meaning that
\#
v^2 =  \frac{1}{4}  \|  \EE \{ (\bX_1 - \bX_2 )(\bX_1 - \bX_2)^\intercal \}^2  \|_2  \label{v2.def}
\#
is finite. 
\begin{theorem}\label{thm:u-type}
For any $0< \delta <1$, the estimator $\hat \bSigma_2^{\TT} = \hat \bSigma_{2}^{\TT}(\tau)$ with
\#
	\tau =  v \sqrt{\frac{m}{\log(2d) + \log \delta^{-1}}} \label{tau.cond}
\# 
satisfies, with probability at least $1- \delta$,
\#
 \| \hat \bSigma_2^{\TT} -\bSigma \|_2 \leq  2 v \sqrt{\frac{\log(2d) + \log \delta^{-1}}{m}}   . \label{spectral.concentration}
\#
\end{theorem}

To better understand this result, note that $v^2$ can be written as
\$
 \frac{1}{2} \|    \EE \{ (\bX-\bmu)(\bX-\bmu)^\intercal \}^2  + {\rm tr}(\bSigma) \bSigma + 2 \bSigma^2  \|_2 ,
\$
which is well-defined if the fourth moments $\EE(X_k^4)$ are finite.
Let 
\[
K = \sup_{\bu\in \RR^d} {\rm kurt}(\bu^\T \bX)
\] 
be the maximal kurtosis of the one-dimensional projections of $\bX$. Then
\$
 v^2 \leq   \| \bSigma \|_2 \{  (K+1) {\rm tr}(\bSigma)/2 +  \|\bSigma \|_2 \}. 
\$
The following result is a direct consequence of Theorem \ref{thm:u-type}: $\widehat\bSigma_2^{\TT}$ admits exponential-type concentration for data with finite kurtoses.

\begin{corollary}  \label{coro:spectrum}
Assume that $K= \sup_{\bu\in \RR^d} {\rm kurt}(\bu^\T \bX)$ is finite. Then, for any $0< \delta <1$, the estimator $\hat \bSigma_2^{\TT} = \hat \bSigma_{2}^{\TT}(\tau)$ defined in Theorem~\ref{thm:u-type} satisfies
\#
 \| \hat{\bSigma}_{2}^{\TT} - \bSigma \|_2 \lesssim  K^{1/2} \| \bSigma \|_2  \sqrt{ \frac{ {\rm r}(\bSigma) ( \log d +\log \delta^{-1} ) }{n}}  \label{spec.bound}
\#
with probability at least $1-\delta$.
\end{corollary}

\begin{remark} \label{rem: MZ estimator}
An estimator proposed by  \cite{MZ2018} achieves a more refined deviation bound, namely, with $\| \bSigma \|_2 \sqrt{ {\rm r}(\bSigma) ( \log d +\log \delta^{-1})}$ in (\ref{spec.bound}) improved to $\| \bSigma \|_2 \sqrt{{\rm r}(\bSigma)  \log {\rm r}(\bSigma) } + \| \bSigma \|_2 \sqrt{\log \delta^{-1}}$; in particular, the deviations are controlled by the spectral norm $\|\bSigma\|_2$ instead of the (possibly much larger) trace ${\rm{tr}}(\bSigma)$. Estimators admitting such deviations guarantees are often called ``sub-Gaussian'' as they achieve performance similar to the sample covariance obtained from data with multivariate normal distributions. Unfortunately, the aforementioned estimator is computationally intractable. 
The question of computational tractability was subsequently resolved by \cite{H2018}  and \cite{CFB2019}.  The former showed that a polynomial-time algorithm  achieves statistically optimal rate under the $\ell_2$-norm, and the latter proposed an estimator that has a significantly faster runtime and has sub-Gaussian error bounds; in particular, these results apply to covariance estimation in Frobenius norm. 
Yet it remains an open problem to design a polynomial-time algorithm capable of efficiently computing the estimator proposed by \cite{MZ2018} that achieves near-optimal deviation in the spectral norm. 
\end{remark}

\subsection{An $M$-estimation viewpoint}

In this section, we discuss alternative tail-robust covariance estimators from an $M$-estimation perspective, and study both the  element-wise and spectrum-wise  truncated  estimators. The connection with truncated covariance estimators  is discussed at the end of this section. To proceed, we revisit the definition of Huber loss. 

\begin{definition}[Huber loss] \label{huber.def}
 The Huber loss $\ell_\tau(\cdot)$ \citep{H1964} is defined as
\begin{align}
	\ell_\tau(u) =
	\left\{\begin{array}{ll}
	 u^2/2  ,    & \mbox{if } | u | \leq  \tau ,  \\
	\tau | u | -  \tau^2 /2 ,   &  \mbox{if }  | u | > \tau ,
	\end{array}  \right.
\end{align}
where $\tau >0$ is a robustification parameter similar to that in Definition~\ref{truncated.def}.
\end{definition}

Compared with the squared error loss, large values of $u$ are down-weighted in the Huber loss, yielding robustness. Generally speaking, minimizing Huber's loss produces a biased estimator of the mean, and parameter $\tau$ can be chosen to control the bias.
In other words, $\tau$ quantifies the tradeoff between bias and robustness. 
As observed by \cite{sun2017adaptive}, in order to achieve an optimal tradeoff, $\tau$ should adapt  to the sample size, dimension and the noise level of the problem.

Starting with the  element-wise method, we define the entry-wise estimators
\# \label{huber.est}
	\hat{\sigma}^{\HH} _{1,k \ell} = \argmin_{\theta \in \RR}  \SUM \ell_{\tau_{k \ell}}  (  Y_{ik}Y_{i\ell} /2   -\theta  ), \ \ 1\leq k,\ell \leq d, 
\#
where $\tau_{k\ell}$ are robustification parameters satisfying $\tau_{k\ell} = \tau_{\ell k}$. 
When $k=\ell$, even though the minimization is over $\RR$, it turns out the solution $\hat{\sigma}^{\HH}_{1,kk}$ is still positive almost surely and therefore provides a reasonable estimator of $\sigma^{\HH}_{1,kk}$. To see this, for each $1\leq k\leq d$, define $   \theta_{0k} =  \min_{1\leq i\leq N} Y_{ik}^2/2$ and note that for any  $\tau>0$ and $\theta \leq    \theta_{0k}$,
\[
 \SUM  \ell_\tau (  Y_{ik}^2 /2 -\theta  ) \geq \SUM  \ell_\tau (    Y_{ik}^2 /2   -   \theta_{0k}  ).
\]
It implies that $\hat{\sigma}^{\HH}_{1,kk} \geq   \theta_{0k}$, which is strictly positive as long as there are no tied observations. Again, concatenating these marginal estimators, we obtain a Huber-type $M$-estimator 
\#
	\hat \bSigma_{1}^{\HH} = \hat{\bSigma}_1^{\HH}( \bGamma) = (\hat{ \sigma}^{\HH}_{1,k \ell})_{1\leq k,\ell \leq d} , \label{Huber1}
\# 
where $\bGamma  = (\tau_{k\ell})_{1\leq k,\ell \leq d }$. The following main result of this section indicates that $\hat{\bSigma}^{\HH}_1 $ achieves tight concentration under the max norm for data with finite fourth moments.

\begin{theorem} \label{thm:huber-type}
Let $\Vb = ( v_{k\ell})_{1\leq k,\ell \leq d}$ be a symmetric matrix with entries 
\# \label{v.cond}
	 v_{k\ell}^2 =    \var(   (X_{1k} - X_{2k})(X_{1\ell} - X_{2\ell })/2 ).
\#
For any $0< \delta <1$, the covariance estimator $\hat \bSigma^{\HH}_1 $ given in \eqref{Huber1}  with 
\#
	 \bGamma  = \sqrt{\frac{m}{2\log d + \log \delta^{-1}}} \,\Vb \label{def:taukl}
\# 
satisfies
\#
	\PP  \Bigg(  \| \hat{\bSigma}^{\HH}_1  - \bSigma \|_{\max}   \geq  4  \| \Vb \|_{\max}  \sqrt{\frac{2\log d + \log \delta^{-1}}{m }   }   \Bigg)  \leq 2 \delta   \label{max.concentration}
\#
as long as $m \geq 8 \log  (d^2 \delta^{-1} )$.
\end{theorem}

The $M$-estimator counterpart of the spectrum-truncated  covariance estimator was first proposed by \cite{M2016} using a different robust loss function, and extended by \cite{MW2018} to more general framework of $U$-statistics. In line with the previous element-wise $M$-estimator, we restrict our attention to the Huber loss and consider 
\#
	\hat \bSigma^{\HH}_2 \in \argmin_{\Mb \in \RR^{d\times d}: \Mb= \Mb^\intercal }   \tr\bigg\{   \frac{1}{N} \SUM  \ell_\tau( \bY_i \bY_i^\intercal/2 - \Mb ) \bigg\} ,  \label{Huber2}
\#
which is a natural robust variant of the sample covariance matrix
$$
\hat \bSigma^{\sam}  = \argmin_{\Mb \in \RR^{d\times d}: \Mb= \Mb^\intercal }    \tr\bigg\{    \frac{1}{N}   \SUM  ( \bY_i \bY_i^\intercal/2 - \Mb )^2 \bigg\}.
$$
Define the $d\times d$ matrix $\Sb_0 = \EE \{ (\bX_1 -\bX_2)(\bX_1 - \bX_2)^\intercal  /2 - \bSigma \}^2$ that satisfies 
$$
	\Sb_0  = \frac{\EE \{(\bX-\bmu)(\bX-\bmu)^\T\}^2 + \tr(\bSigma)  \bSigma   }{2}  .
$$
The following result is modified  from  Corollary 4.1 of \cite{MW2018}. 
\begin{theorem}  \label{thm:wei-minsker}
Assume that there exists some $K>0$ such that $\sup_{\bu\in \RR^d} {\rm kurt}(\bu^\T \bX) \leq K$. Then for any $0<\delta<1$ and $v  \geq   \| \Sb_0 \|_2^{1/2}$, the $M$-estimator $\hat \bSigma^{\HH}_2$ with $\tau = v\sqrt{m/(2\log d+ 2\log \delta^{-1})}$ satisfies
\#
	 \| \hat \bSigma^{\HH}_2 - \bSigma \|_2 \leq C_1 v \sqrt{\frac{\log d + \log \delta^{-1}}{m}} 
\#
with probability at least $1- 5 \delta$ as long as $n \geq C_2 K \cdot {\rm r}(\bSigma) ( \log d +  \log \delta^{-1})$, where $C_1, C_2>0$ are absolute constants.
\end{theorem}

To solve the convex optimization problem \eqref{Huber2}, \cite{MW2018} propose the following gradient descent algorithm: starting with an initial estimator $ \hat \bSigma^{ (0)}$, at iteration $t=1,2,\ldots$, compute
\#
 \hat \bSigma^{(t)} =   \hat \bSigma^{(t-1)} -  \frac{1}{N} \sum_{i=1}^N  \psi_\tau\Big( \bY_i \bY_i^\intercal/2 -   \hat \bSigma^{(t-1)}  \Big) ,  \nn
\#
where $\psi_\tau$ is given in \eqref{psi.def}. 
From this point of view,  the truncated estimator $ \hat\bSigma_{2}^{\TT} $ given in \eqref{U.cov.estimator} can be viewed as the first step of the gradient descent iteration for solving optimization problem \eqref{Huber2} initiated at $\hat \bSigma^{ (0)} = \textbf{0}$. This procedure enjoys a nice contraction property, as demonstrated by Lemma~3.2 of \cite{MW2018}.
However, since the difference matrix $ \bY_i \bY_i^\intercal/2 -   \hat \bSigma^{(t-1)} $ for each $t$ is no longer rank-one, we need to perform  a singular value decomposition to compute the matrix $\psi_\tau ( \bY_i \bY_i^\intercal/2 -   \hat \bSigma^{(t-1)}  )$ for $i=1,\ldots, N$.

We end this section with a discussion of the similarities and differences  between M-estimators and estimators defined via truncation. Both types of estimators achieve tail robustness through a bias-robustness tradeoff, either element-wise or spectrum-wise. 
However (informally speaking), $M$-estimators truncate symmetrically around the true expectation as shown in (\ref{huber.est}) and (\ref{Huber2}), while the truncation-based estimators truncate around zero as  in   (\ref{element-wise.censored.est}) and (\ref{U.cov.estimator}). Due to smaller bias, $M$-estimators are expected to outperform the simple truncation estimators. However, since  the optimal choice of the robustification parameter is often much larger than the population moments in magnitude, either  element-wise or spectrum-wise, the difference between truncation estimators and $M$-estimators becomes insignificant when the sample size $n$ is large. Therefore, we advocate using the simple truncated estimator primarily due to its simplicity and computational efficiency.

\subsection{Median-of-means estimator}
\label{sec2.1.1}

Truncation-based approaches described in the previous sections require knowledge of robustification parameters $\tau_{k\ell}$. Adaptation and tuning of these parameters will be discussed in Section \ref{sec:3.5} below. 
Here, we suggest another method that does not need any tuning but requires stronger assumptions, namely,  existence of moments of order six.
This method is based on the median-of-means (MOM) technique \citep{Nemirovski1983Problem-complex00,DLLO2016,MS2017}. 
To this end, assume that the index set $\{1,\ldots,n\}$ is partitioned into $k$ disjoint groups $G_1,\ldots,G_k$ (partitioning scheme is assumed to be independent of $\bX_1,\ldots,\bX_n$) such that the cardinalities $|G_j|$ satisfy 
$\left| |G_j| - \frac{n}{k} \right| \leq 1$ for $j=1,\ldots,k$. 
For each $j=1,\ldots,k$, let $\bar \bX_{G_j} =  (1/|G_j|)\sum_{i\in G_j} \bX_i$ and 
\[
\hat{\bSigma}^{(j)} = \frac{1}{|G_j|}\sum_{i\in G_j}(\bX_i - \bar \bX_{G_j})(\bX_i - \bar \bX_{G_j})^\T
\]
be the sample covariance evaluated over the data in group $j$. 
Then, for all $1\leq \ell ,m\leq d$, the MOM estimator of $\sigma_{\ell m}$ is defined via 
\[
\hat\sigma_{\ell m}^{\mathrm{MOM}} = \mathrm{median} \big\{  \hat{\sigma}^{(1)}_{\ell m},\ldots,\hat{\sigma}^{(k)}_{\ell m} \big\},
\]
where $\hat{\sigma}^{(j)}_{\ell m}$ is the entry in the $\ell$-th row and $m$-th column of $\hat{\bSigma}^{(j)}$. This leads to
\[
\hat{\bSigma}^{\mathrm{MOM}} = \left( \hat\sigma_{\ell m}^{\mathrm{MOM}} \right)_{1\leq \ell ,m\leq d}.
\] 
Let $\Delta^2_{\ell m} = \mathrm{Var} ( (X_\ell - \mathbb E X_\ell)(X_m - \mathbb E X_m))$ for $1\leq \ell, m\leq d$. The following result provides a deviation bound for the MOM estimator $\hat{\bSigma}^{\mathrm{MOM}}$ under the max norm.

\begin{theorem} 
\label{thm:mom}
Assume that $\min_{\ell ,m}\Delta^2_{\ell m} \geq c_\ell > 0$ and $\max_{1\leq k\leq d }\mathbb E | X_k - \mathbb E X_k |^6  \leq c_u < \infty$. 
Then, there exists $C_0>0$ depending only on $(c_\ell, c_u)$ such that
\#
\PP   \Bigg(   \| \hat{\bSigma}^{\mathrm{MOM}} - \bSigma \|_{\max}   \geq  3\max_{\ell ,m} \Delta_{\ell m}\left\{ \sqrt{\frac{ \log (d+1) + \log \delta^{-1}}{ n}} + C_0\frac{k}{n}\right\}   \Bigg) \leq 2 \delta
\label{MOM.bound}
\#
for all $\delta$ satisfying $\sqrt{\{ \log (d+1) + \log \delta^{-1} \}/k} + C_0{\sqrt{k/n}} \leq 0.33$.
\end{theorem}

\begin{remark}
\
\begin{enumerate}
\item The only user-defined parameter in the definition of $\hat{\bSigma}^{\mathrm{MOM}}$ is the number of subgroups $k$. 
The bound above shows that, provided $k\ll \sqrt{n}$ (say, one could set $k = \frac{\sqrt{n}}{\log n})$, the term $C_0\frac{k}{n}$ in \eqref{MOM.bound} is of smaller order, and we obtain an estimator that admits tight deviation bounds for a wide range of $\delta$. In this sense, estimator $\hat{\bSigma}^{\mathrm{MOM}}$ is essentially a multiple-$\delta$ estimator \citep{DLLO2016}; see Remark~\ref{remark1}.

\item Application of the MOM construction to large covariance estimation problems has been explored by \cite{ABFL2018}. 
However, the results obtained therein are insufficient to conclude that MOM estimators are truly ``tuning-free''.
Under a bounded fourth moment assumption,  \cite{ABFL2018} derived a deviation bound (under max norm) for the element-wise median-of-means estimator with the number of partitions depending on a prespecified confidence level parameter. See Proposition~5 therein.

\end{enumerate}
\end{remark}

\section{Automatic tuning of robustification parameters} 
\label{sec:3.5}

For all the proposed tail-robust estimators besides the median-of-means, the robustification parameter needs to adapt properly to the sample size, dimensionality and noise level in order to achieve optimal tradeoff between bias and robustness in finite samples. An intuitive yet computationally expensive idea is to use cross-validation. 
Another approach is based on Lepski's method \citep{Lepskii1997}; this approach yields adaptive estimators with provable guarantees \citep{M2016,MW2018}, however, it is also not computationally efficient. 
In this section, we propose tuning-free approaches for constructing both  truncated and $M$-estimators that have low computational costs. Our nonasymptotic analysis provides useful guidance on the choice of key tuning parameters.

\subsection{Adaptive truncated estimator}\label{tuning:1}
In this section we introduce a data-driven procedure that automatically tunes the robustification parameters in the  element-wise truncated covariance estimator.
Practically, each individual parameter can be tuned by cross-validation from a grid search. This, however, will incur extensive computational cost even when the dimension $d$ is moderately large. 
Instead, we propose a data-driven procedure that automatically calibrates the $d(d+1)/2$ parameters. This procedure is motivated by the theoretical properties established in Theorem \ref{thm:element-wise}. 
To avoid notational clutter, we fix $1\leq k\leq \ell \leq d $ and define $\{ Z_{  1} \ldots,  Z_{ N} \}  = \{  Y_{1k}Y_{1\ell}/2    , \ldots, Y_{Nk}Y_{N\ell}/2 \} $ such that  $\sigma_{k\ell} = \EE Z_{  1}$. Then $\hat\sigma^{\TT}_{1, k\ell}$ can be written as $(1/N)\SUM \psi_{\tau_{k\ell}} (Z_{ i})$.  In view of \eqref{tuning1}, an ``ideal" choice of $\tau_{k\ell}$ is
\#\label{eq:scaling}
	\tau_{k \ell} = v_{k\ell} \sqrt{\frac{m }{2\log d + t}} ~\mbox{ with }~ v^2_{k\ell} = \EE Z_{  1}^2 ,
\#
where $t = \log \delta^{-1} \geq 1$ is  prespecified to control the confidence level and will be discussed later. 
A naive estimator of $v_{k\ell}^2$ is the empirical second moment $(1/N) \SUM Z_{  i}^2$, which tends to overestimate the true value when data have high kurtosis. Intuitively, a well-chosen $\tau_{k\ell}$ makes $(1/N) \SUM \psi_{\tau_{k\ell}} (Z_{  i})$ a good estimator of $\EE Z_{ 1}$, and meanwhile, we expect the empirical truncated second moment $(1/N)\SUM \psi^2_{\tau_{k\ell}} (Z_{ i}) = (1/N) \SUM ( Z_{  i}^2   \wedge \tau_{k\ell}^2 )$ to be a reasonable estimate of $\EE Z_{1}^2$ as well.
Plugging this empirical truncated second moment into \eqref{eq:scaling} yields
 \begin{align}
   \frac{ 1}{N}\SUM \frac{  ( Z_{ i}^2   \wedge \tau^2  ) }{\tau^2} =  \frac{2\log d + t}{m}   , \ \ \tau > 0.  \label{empirical.censored.moment.equation}
\end{align}
We then solve the above equation to obtain $\widehat \tau_{k\ell}$, a data-driven choice of $\tau_{k\ell}$. 
By Proposition~3 in \cite{WZZZ2018},  equation \eqref{empirical.censored.moment.equation} has a unique solution as long as $2\log d +t  < (m/N) \SUM I\{   Z_{  i}  \!\neq\! 0 \}$. We characterize the theoretical properties of this tuning method in a companion paper \citep{WZZZ2018}.

Regarding the choice of $t = \log \delta^{-1}$: on the one hand, as it controls the confidence level according to \eqref{sigma1.bound}, we should let $t=t_n$ be sufficiently large so that the estimator is concentrated around the true value with high probability.  On the other hand, $t$ also appears in the deviation bound that corresponds to the width of the confidence interval, so it should not grow too fast as a function of $n$. In practice, we recommend using $t=\log n$ (or equivalently, $\delta = n^{-1}$), a typical slowly varying function of $n$. 

To implement the spectrum-wise truncated covariance estimator in practice, note that there is only  one tuning parameter   whose theoretically optimal scale is
$$
	\frac{1}{2} \|  \EE \{ (\bX_1 - \bX_2 )(\bX_1 - \bX_2)^\intercal \}^2  \|_2^{1/2} \sqrt{\frac{m}{\log (2d) + t}}.
$$ 
Motivated by the data-driven tuning scheme for the  element-wise estimator, we choose $\tau$ by (approximately) solving the equation
\begin{align}
  \bigg\| \frac{1}{ \tau^2 N } \SUM  \bigg(    \frac{ \| \bY_i \|_2^2}{2}  \bigwedge  \tau    \bigg)^2  \frac{  \bY_i  \bY_i^\T }{ \|  \bY_i \|_2^2} \bigg\|_2 = \frac{\log(2d) + t  }{ m } , \nn
\end{align}
where as before we take $t=\log n$.

\subsection{Adaptive Huber-type $M$-estimator}\label{sec:adHuber}
To construct a data-driven approach for automatically tuning the adaptive Huber estimator, we follow the same rationale from the previous subsection.  Since the optimal $\tau_{k\ell}$ now depends on $\var(Z_{1})$ instead of the second moment $\EE Z_{ 1}^2 $, it is therefore conservative to directly apply the above data-driven method in this case. Instead, we propose to estimate $\tau_{k\ell}$ and  $\sigma_{k\ell}$ simultaneously by solving the following system of equations
\begin{subequations}\label{tau.M-view}
\begin{align}
	f_1(\theta, \tau )&  =   \frac{ 1}{N}\SUM \frac{ \{ ( Z_{ i}  -\theta )^2 \wedge \tau^2  \}  }{\tau^2} - \frac{2\log d + t}{n} =0 , \\
	f_2(\theta, \tau ) & =   \SUM \psi_\tau( Z_{  i} - \theta ) = 0 , 
\end{align}
\end{subequations}
for $\theta \in \RR$ and $\tau>0$. Via a similar argument, it can be shown that  the equation $f_1(\theta, \cdot) = 0$ has a unique solution as long as $2\log d +t  < (n/N) \SUM I\{   Z_{ i}  \neq \theta \}$; for any $\tau>0$, the equation $f_2(\cdot,  \tau)=0$ also has a unique solution. Starting with an initial estimate $\theta^{(0)} = (1/N)\SUM Z_{  i} $, which is the sample variance estimator of $\sigma_{k\ell}$, we iteratively solve $f_1(\theta^{(s-1)}, \tau^{(s)})=0$ and $f_2(\theta^{(s)}, \tau^{(s)})=0$ for $s=1, 2,\ldots$ until convergence.
The resultant estimator, denoted by $\hat{\sigma}^{\HH} _{3,k \ell} $ with slight abuse of notation, is then referred to as the adaptive Huber estimator of $\sigma_{k\ell}$.
We then obtain the data-adaptive Huber covariance matrix estimator as $
\hat\bSigma^{\HH}_3 = (\hat{\sigma}^{\HH} _{3,k \ell} )_{1\leq k,\ell \leq d}$. Algorithm~\ref{algm_2} presents the summary of this data-driven approach.

\begin{algorithm}
\caption{Data-Adaptive Huber Covariance Matrix Estimation}\label{algm_2}

\noindent
{\bf Input}\  Data  vectors $\bX_i \in  \RR^{d}$ ($i=1,\ldots, n$), tolerance level $\epsilon$ and maximum iteration $S_{\max}$.

\noindent
{\bf Output}\  Data-adaptive Huber covariance matrix estimator $\widehat \bSigma^{\HH}_3 = (\hat{\sigma}_{3, k \ell}^{\HH} )_{1\leq k,\ell \leq d}$. 

\begin{enumerate}
\item[1:] Compute pairwise differences  $\bY_1 = \bX_1 - \bX_2 , \bY_2 = \bX_1 - \bX_3,  \ldots, \,\bY_N = \bX_{n-1} - \bX_n$, where $N= n(n-1)/2$.

\item[2:] {~~~~}{\textbf{for}}$1\leq k \leq \ell \leq d$   {\textbf{do}}

\item[3:] {~~~~ ~~~~} $\theta^{(0)} = (2N)^{-1}\SUM Y_{ik}Y_{i\ell} $.

\item[4:] {~~~~ ~~~~}  {\textbf{for}} $s= 1, \ldots S_{\max}$   {\textbf{do}}

\item[5:]  {~~~~ ~~~~ ~~~~ }$\tau^{(s)} \leftarrow $  solution of  $f_1(\theta^{(s-1)},\cdot)=0$.

\item[6:] {~~~~ ~~~~ ~~~~} $\theta^{(s)} \leftarrow $ solution of $f_2(\cdot , \tau^{(s)})=0$.

\item[7:] {~~~~ ~~~~ ~~~~}  {\textbf{if}} $|\theta^{(s)}-\theta^{(s-1)}| < \epsilon$ \ {\textbf{break}} 

\item[8:] {~~~~ ~~~~}{\textbf{stop}} $\hat{\sigma}_{3, \ell k }^{\HH} =\hat{\sigma}_{3, k \ell}^{\HH} =\theta^{(S_{\max})}$.

\item[9:] {~~~~}{\textbf{stop}}

\item[10:] {\textbf{return}} $\widehat\bSigma^{\HH}_3 = (\hat{\sigma}_{3, k \ell}^{\HH} )_{1\leq k,\ell \leq d}$.
\end{enumerate}
\end{algorithm}

\section{Applications to structured matrix estimation}
\label{sec:4}
The robustness properties of the element-wise and spectrum-wise truncation estimators are demonstrated in Theorems \ref{thm:element-wise} and \ref{thm:u-type}. In particular, the exponential-type concentration bounds are essential for establishing reasonable estimators for high-dimensional structured covariance and precision matrices. In this section, we apply the proposed generic robust methods to the estimation of bandable and low-rank covariance matrices as well as sparse precision matrix.

\subsection{Bandable covariance matrix estimation} \label{sec:bandable}
Motivated by applications to climate studies and spectroscopy in which the index set of variables $\bX = (X_1,\ldots, X_d)^\T$ admits a natural order, one can expect that a large ``distance'' $|k -\ell |$ implies near-independence. We characterize this feature by the following class of bandable covariance matrices considered by \cite{BL2008a} and by \cite{CZZ2010}:
\#
\mathcal{F}_{\alpha } ( M_{0},M ) = \bigg\{   \bSigma = & (\sigma_{k\ell})_{1\leq k,\ell \leq d}    \in \RR^{d\times d} :  \lambda _1 ( \bSigma  ) \leq M_{0} , \nn \\
& \max_{1\leq \ell \leq d}\sum_{k:  \vert
k - \ell  \vert > m }    \vert \sigma _{k \ell } \vert \leq  \frac{M}{m^\alpha} \text{ for all } m  \bigg\}.
\label{bandableness class}
\#
Here $M_0 , M$ are regarded as universal constants and the parameter $\alpha $ specifies the decay rate of $\sigma _{k \ell}$ to zero as $\ell\rightarrow \infty $ for each row. 

When $\bX$ follows sub-Gaussian distribution, \cite{CZZ2010} proposed a minimax-optimal estimator over $\mathcal{F}_{\alpha }\left( M_{0},M\right)$ under the spectral norm. Specifically, they proposed a tapering estimator  
$\hat{\bSigma}^{\rm{tap}}_{m}=(\hat{\sigma}_{k \ell }\cdot \mathbb{\omega }_{\left\vert k - \ell \right\vert })$, where the positive integer $m \leq d$ specifies the bandwidth, $\mathbb{\omega }_q=1$, $2-2q/m$, $0$, when $q\leq m/2$, $m/2<q\leq m$, $q>m$, respectively.  $\hat{\bSigma}^{\rm{sam}} = (\hat \sigma_{k\ell})_{1\leq k,\ell \leq d}$ denotes the sample covariance. With the optimal choice of bandwidth $m \asymp \min \{n^{1/(2\alpha +1)}, d\}$, \cite{CZZ2010} showed that $\hat{\bSigma}^{\rm{tap}}_{m}$ achieves the minimax rate of convergence $\{  \sqrt{\log (d)/n} + n^{- \alpha/(2\alpha +1)} \} \wedge \sqrt{d/n}$ under the spectral norm.

To obtain a root-$n$ consistent covariance estimator, we expect the coordinates  of $\bX$ to have at least  finite fourth moments. Under this  condition, it is unclear whether the optimal rate can be achieved over $\mathcal{F}_{\alpha }\left( M_{0},M\right)$ without imposing additional distributional assumptions, such as the elliptical symmetry \citep{MZ2014,chen2018robust}.
Estimators that naively use the sample covariance will inherit its sensitivity to outliers. 
Recall the definition of $ \hat\bSigma_{2}^{\TT} $ in \eqref{U.cov.estimator}; a simple idea is to replace the sample covariance by a spectrum-wise truncated estimator $ \hat\bSigma_{2}^{\TT} $ in the first step, to which the tapering procedure can be applied. However, such an estimator is not optimal: indeed, the analysis of a tapering estimator requires each small  principal submatrix of the initial estimator to be highly concentrated around the population object. 
Suppose that we truncate the $\ell_2$-norm of the entire  vector $\bY_i$ at a level $\tau$ scaling with ${\rm tr}(\bSigma)$. For each subset $J \subseteq \{1 , \ldots, d\}$, let $\bY_{iJ}$ be the subvector of $\bY_i$ indexed by $J$. Then the corresponding principal submatrix 
$$
\frac{1}{N} \SUM  \psi_\tau\bigg( \frac{ 1}{2} \|   \bY_i \|_2^2\bigg) \frac{\bY_{iJ} \bY_{i J}^\T}{\|   \bY_i \|_2^2 } 
$$
is not an ideal robust estimator of $\bSigma_{JJ}$ because the ``optimal'' $\tau$ in this case should scale with ${\rm tr}(\bSigma_{JJ})$ rather than ${\rm tr}(\bSigma)$. This explains why directly applying the tapering procedure to $\hat\bSigma_{2}^{\TT} $ is not ideal.

In what follows, we propose an optimal robust covariance estimator based on the spectrum-wise truncation technique introduced in Section \ref{sec2.2}. First we introduce some notation. Let $\bZ_{i}^{(p,q)} = (Y_{i , p},Y_{i, p+1},\ldots,Y_{i , p+q-1})^\T$ be a subvector of $\bY_i$ given in \eqref{paired.data}. 
Accordingly,  define the truncated  estimator of the principal submatrix of $\bSigma$ as
\begin{equation}
\hat\bSigma_{2}^{(p,q), \TT} = \hat \bSigma^{(p,q), \TT}_{2 }(\tau)  =\frac{ 1}{N} \SUM \psi_\tau( \bZ^{(p,q)}_i \bZ^{(p,q)\T}_i/2 ), \label{principal.estimator}
\end{equation}
where $\tau$ is as in (\ref{tau.cond}) with $d$ replaced by $q$ and $v = \| \EE  \{\bZ_{1}^{(p,q)} \bZ_{1}^{(p,q)  \T }\}^2 \|_2/4 $.
Moreover, we define an operator that embeds a small matrix into a large zero matrix: for a $q \times q$ matrix $\Ab = (a_{k \ell })_{1\leq k , \ell \leq q}$, define the $d \times d$ matrix $\extend{d}{p}(\Ab) = (b_{k \ell  })_{1\leq k, \ell \leq d}$, where $p$ indicates the location and
\begin{equation*}
b_{k \ell } = 
\begin{cases}
a_{k-p+1, \ell -p+1} & \text{ if } p \leq  k , \ell  \leq  p + q -1, \\
0 & ~\mbox{otherwise} .
\end{cases}
\end{equation*}
Our final robust covariance estimator is then defined as 
\begin{equation} \label{eq: bandable estimator}
\hat\bSigma_{q} =  \sum_{j=-1}^{\ceil{(d-1)/q}} \extend{d}{j q +1}(\hat\bSigma_{2}^{(j q+1,2q) ,\TT}) - \sum_{j=0}^{\ceil{(d-1)/q}} \extend{d}{j q+1}(\hat\bSigma_{2}^{(j q +1, q), \TT}).
\end{equation} 

The idea behind the construction above is that a bandable covariance matrix in $\mathcal{F}_{\alpha }\left( M_{0},M\right)$ can be approximately decomposed into several  principal submatrices of size $2q$ and $q$, as shown in Figure \ref{Fig_bandable}. 
Using spectrum-wise truncated estimators $\hat\bSigma_{2}^{(p,q),\TT}$ and $\hat\bSigma_{2}^{(p,2q),\TT}$ to estimate the corresponding principal submatrices in this decomposition leads to the proposed estimator $\hat\bSigma_{q}$. 
\begin{figure}[t]
	\centering
	\includegraphics[width=5 in]{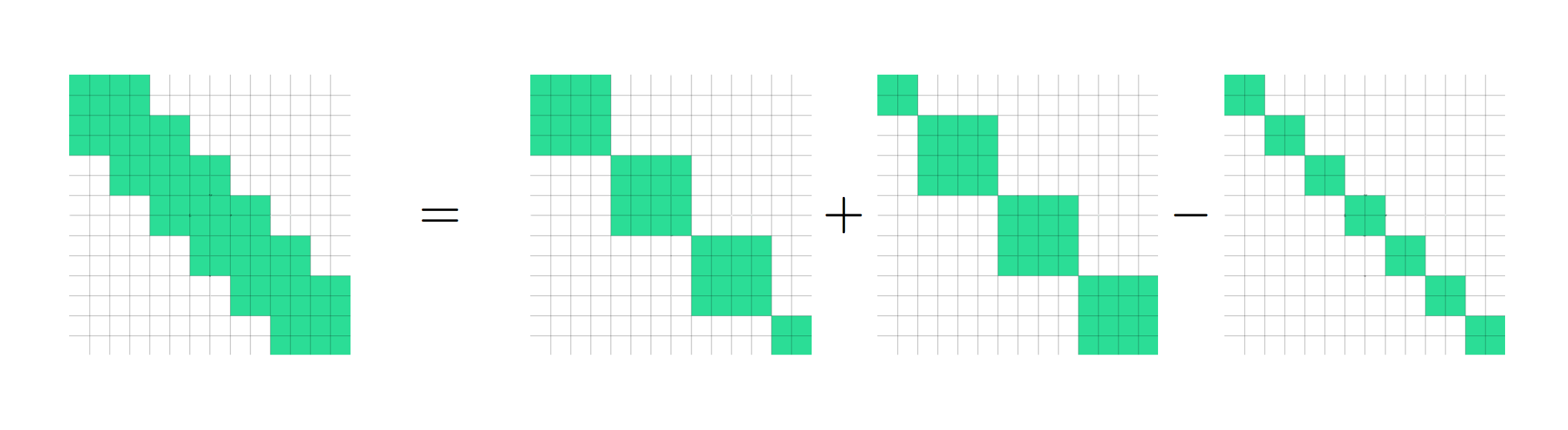}
	\caption{Motivation of our estimator of bandable covariance matrices}
	\label{Fig_bandable}
\end{figure}

This construction is different from the literature where the banding or tapering procedure is directly applied to an initial estimator, say the sample covariance matrix \citep{BL2008a, CZZ2010}. It is worth mentioning that a similar robust estimator can be constructed following the idea of \cite{CZZ2010}, which differs from our proposal. 
Computationally, our estimator evaluates as many as $O(d/q)$  matrices of size $q \times q$ (or $2q \times 2q$), while the method developed by  \cite{CZZ2010} computes as many as $O(d)$ such matrices.

The following result shows that the  estimator defined in \eqref{eq: bandable estimator} achieves near-optimal rate of convergence under the spectral norm as long as $\bX$ has uniformly bounded fourth moments. 
The proof is deferred to the supplementary material. 

\begin{theorem}\label{thm:band}
Assume that $\bSigma \in \mathcal{F}_{\alpha }\left( M_{0},M\right)$ and $\sup_{\bu\in \mathbb{S}^{d-1}} {\rm kurt}(\bu^\T \bX)\leq M_1$ for some constant $M_1>0$. For any $c_0>0$,  take $\delta=( n^{c_0} d)^{-1}$ in the definition of $\tau$ for constructing principal submatrix estimators $\hat\bSigma_{2}^{(p,q), \TT}$ in (\ref{principal.estimator}). 
Then, with a bandwidth $q \asymp  \{ n/ \log (nd)\}^{1/(2\alpha + 1)} \wedge  d$, the estimator $\hat\bSigma_{q}$ defined in (\ref{eq: bandable estimator}) is such that with probability at least $1-2 n^{-c_0}$, 
	\$
	\|\hat\bSigma_q -\bSigma\|_{2}
	 \leq 	C   \min \Bigg\{ \bigg(\frac{\log (n d )}{n} \bigg)^{  \alpha/( 2\alpha+1 ) }, \,  \sqrt{  \frac{d \cdot \log(nd)}{n}} \,   \Bigg\} ,
	\$
	where $C >0$ is a constant depending only on $M, M_0, M_1, c_0$.
\end{theorem}
According to the minimax lower bounds established by \cite{CZZ2010}, up to a logarithmic term our robust estimator achieves the optimal rate of convergence that is enjoyed by the tapering estimator when the data are sub-Gaussian. The logarithmic term is not easy to remove: for instance, one cannot improve the rate of convergence by replacing each principal submatrix estimator in (\ref{principal.estimator}) with the theoretically more refined but computationally intractable estimator proposed by \cite{MZ2018} (see Remark \ref{rem: MZ estimator}). 

\begin{remark}
To achieve the near-optimal convergence rate shown in Theorem \ref{thm:band},  the ideal choice of the bandwidth $q$ depends on the knowledge of $\alpha$. A fully data-driven and adaptive estimator can be obtained by selecting the optimal bandwidth via Lepski's method \citep{Lepskii1997}. We refer to \cite{LR2018} for constructing a similar adaptive estimator for a precision matrix with bandable Cholesky factor.  
\end{remark}

\subsection{Low-rank covariance matrix estimation}

In this section, we consider a structured model where $\bSigma = {\rm cov}(\bX)$ is approximately low-rank.
Using the trace-norm as a convex relaxation of the rank, our estimator is the solution to the following trace-norm penalized optimization program:
\#
\hat\bSigma_{2, \gamma}^{\TT} \in   \argmin_{\Ab \in \cS_d }  \bigg\{ \frac{1}{2} \| \Ab - \hat\bSigma_{2}^{\TT} \|_{\F}^2 + \gamma \| \Ab \|_{\tr}    \bigg\} ,  \label{lowrank.U}
\#
where $\cS_d$ denotes the set of $d\times d$ positive semi-definite matrices, $\gamma>0$ is a regularization parameter and $\hat\bSigma_{2}^{\TT}$, defined in \eqref{U.cov.estimator}, serves as a pilot estimator. This trace-penalized method was first proposed by \cite{L2014} with the initial estimator taken to be the sample covariance matrix, and later studied by \cite{M2016} using a different initial estimator. In fact, given the initial estimator $\hat\bSigma_{2}^{\TT}$, the estimator given in \eqref{lowrank.U} has the following closed-form expression \citep{L2014}:
\#
\hat\bSigma_{2 , \gamma}^{\TT} = \sum_{k=1}^d \max\{  \lambda_k( \hat\bSigma_{2}^{\TT} ) - \gamma, 0 \} \bv_k( \hat\bSigma_{2}^{\TT} ) \bv_k(\hat\bSigma_{2}^{\TT} )^\T , \label{lowrank.U2}
\#
where $\lambda_1(\hat\bSigma_{2}^{\TT}) \geq \cdots \geq \lambda_d(\hat\bSigma_{2}^{\TT})$ are the eigenvalues of $\hat\bSigma_{2}^{\TT}$ in an non-increasing order and $\bv_1(\hat\bSigma_{2}^{\TT}),\ldots, \bv_d(\hat\bSigma_{2}^{\TT})$ are the associated orthonormal eigenvectors. The following theorem provides a deviation bound for $\hat\bSigma_{2 , \gamma}^{\TT} $ under the Frobenius norm. The proof follows directly from Theorem \ref{thm:u-type} and Theorem 1 of \cite{L2014}, and therefore is omitted.

\begin{theorem} \label{thm:lowrank.U-type}
	For any $t>0$ and $v>0$ satisfying \eqref{v2.def}, let 
	$$
	\tau =  v \sqrt{ \frac{m}{\log(2d)+t} } ~\mbox{ and }~ \gamma \geq 2 v \sqrt{\frac{\log(2d) + t}{m}} . 
	$$
	Then  with probability at least $1 - e^{-t}$, the trace-penalized estimator $\hat\bSigma_{2,\gamma}^{\TT}$ satisfies
	\#
	\| \hat\bSigma_{2,\gamma}^{\TT}  - \bSigma \|_{\F}^2 \leq 
	\inf_{\Ab \in \cS_d }  \Big[ \| \bSigma - \Ab \|_{\F}^2 + \min\{ 4 \gamma \| \Ab \|_{\tr} , 3\gamma^2 {\rm rank}(\Ab) \} \Big]  \nn \\
	\mbox{ and }~  \| \hat\bSigma_{2,\gamma}^{\TT}  - \bSigma \|_2 \leq 2 \gamma. \nn
	\#
	In particular, if ${\rm rank}(\bSigma) \leq  r_0$, then with probability at least $1- e^{-t}$,
	\#
	   \| \hat\bSigma_{2,\gamma}^{\TT}  - \bSigma \|_{\F}^2  \leq   \min\{ 4 \| \bSigma \|_2 \gamma   , 3  \gamma^2 \} r_0  .
	\#
\end{theorem}

\subsection{Sparse precision matrix estimation}

Our third example is related to sparse precision matrix estimation in high dimensions. Recently, \cite{ABFL2018} showed that minimax optimality is achievable within a larger class of distributions if the sample covariance matrix is replaced by a robust pilot estimator, and also provided a unified theory for covariance and precision matrix estimation based on general pilot estimators. Specifically, \cite{ABFL2018} robustifed the CLIME estimator \citep{CLL2011} using three different pilot estimators: adaptive Huber, median-of-means and rank-based estimators.
Based on the element-wise  truncation procedure and the difference of trace (D-trace) loss proposed by \cite{zhang2014sparse}, we  further consider a robust method for estimating the precision matrix $\bTheta^* = \bSigma^{-1}$ under sparsity, which represents a useful complement to the methods developed by  \cite{ABFL2018}.

The advantage of using the D-trace loss is that it automatically gives rise to a symmetric solution. Specifically, using the element-wise truncated estimator $\hat{\bSigma}_{1}^{\TT}  = \hat \bSigma_1^{\TT}(\bGamma) $ in \eqref{element-wise.censored.est} as an initial estimate of $\bSigma$, we propose to solve 
\#\label{eq:mp}
\widehat\bTheta  \in \argmin_{\bTheta \in \RR^{d\times d}} \bigg\{  \underbrace{  \frac{1}{2} \big\langle\bTheta^2, \hat{\bSigma}_{1}^{\TT}   \big\rangle-\tr(\bTheta) }_{\cL(\bTheta )}+ \lambda \|\bTheta \|_{\ell_1} \bigg\}.
\#
where $ \|\bTheta \|_{\ell_1} =  \sum_{ k\neq \ell} |\Theta_{k \ell }|$ for $\bTheta= (\Theta_{k\ell})_{1\leq k,\ell \leq d}$. For simplicity, we write $\cL(\bTheta)= \langle\bTheta^2, \hat{\bSigma}_{1}^{\TT} \rangle-\tr(\bTheta)$.  \citet{zhang2014sparse} imposed a positive definiteness constraint on $\bTheta$, and proposed an alternating direction method of multipliers (ADMM) algorithm to solve the constrained D-trace loss minimization. However, with the positive definiteness  constraint, the ADMM algorithm at each iteration computes the singular value decomposition of a $d\times d$ matrix, and therefore is computationally intensive for large-scale data.  In \eqref{eq:mp}, we impose no constraint on $\bTheta$ primarily for computational simplicity. 

 Before presenting the main theorem, we need to introduce an assumption on the restricted eigenvalue of the Hessian matrix of $\cL$. First, note that the Hessian of $\cL$ can be written as 
\$
\Hb_{\bGamma}=  \frac{1}{2} ( \Ib\otimes \widehat\bSigma_1^\TT +\widehat\bSigma_1^\TT\otimes \Ib ) ,
\$
where $\bGamma$ is the tuning parameter matrix in \eqref{element-wise.censored.est}.
For matrices $\Ab, \Bb\in \RR^{d^2 \times d^2}$, we define $\langle\Ab, \Ab\rangle_\Bb=\vecc (\Ab)^\T\,\Bb\,\vecc(\Ab)$, where $\vecc(\Ab)$ designates the $d^2$-dimensional vector concatenating the columns of $\Ab$.  Let $\cS=\supp(\bTheta^* ) \subseteq \{ 1,\ldots, d\}^2$, the support set of $\bTheta^*$. 
\begin{definition} (Restricted Eigenvalue for Matrices) \label{def:restricted eigenvalue}
For any $\xi >0$ and $m\geq 1$, we define the maximal and minimal restricted eigenvalues of the Hessian matrix $\Hb_{\bGamma}$  as
\$
\kappa_{-} (\bGamma, \xi,m) &=  \underset{\Wb}{\inf} \bigg\{
\frac{\langle\Wb, \Wb\rangle_{\Hb_{\bGamma}}}{\|\Wb\|_{\F}^2} :\Wb \in \RR^{d\times d}, \Wb \ne 0,   \exists J \mbox{ such that } \cS \subseteq J, \\
& \quad\qquad\qquad\qquad \qquad |J|\le m,\, \|\Wb_{J^{\rm c}}\|_{\ell_1} \le \xi  \| \Wb_{J}  \|_{\ell_1}
\bigg\};\\
\kappa_{+} (\bGamma, \xi,m) &=  \underset{\Wb}{\sup} \bigg\{
\frac{\langle\Wb, \Wb\rangle_{\Hb_{\bGamma}}}{\|\Wb\|_{\F}^2} :\Wb \in \RR^{d\times d}, \Wb \ne 0, \exists J \mbox{ such that } \cS \subseteq J, \\
& \quad\qquad\qquad\qquad \qquad  |J|\le m, \, \|\Wb_{J^{\rm c}}\|_{\ell_1} \le \xi   \| \Wb_{J}  \|_{\ell_1}\bigg\}.
\$
\end{definition}

\begin{cond}(Restricted Eigenvalue Condition)  \label{REC}
We say restricted eigenvalue condition with $(\bGamma, 3,k)$ holds  if $0<\kappa_-=\kappa_-(\bGamma, 3,k)\leq\kappa_+(\bGamma, 3,k)= \kappa_+<\infty$. 
\end{cond}

Condition \ref{REC} is a form of the localized restricted eigenvalue condition \citep{fan2015tac}. Moreover, we assume that the true precision matrix $\bTheta^*$ lies in the following class of matrices:
\$
\cU(s,M)&=\bigg\{\bOmega\in \RR^{d\times d}: \bOmega=\bOmega^\T,  \bOmega\succ 0, \|\bOmega\|_1\leq M, \, \sum_{k, \ell}  I(\Omega_{k \ell }\neq0)\leq s \bigg\}.
\$
A similar class of precision matrices has been studied in the literature; see, for example, \cite{zhang2014sparse}, \cite{CRZ2016} and \cite{ sun2018graphical}.  Recall the definition of $\Vb$ in Theorem~\ref{thm:element-wise}. 
We are ready to present the main result, with the proof deferred to the supplementary material.

\begin{theorem}\label{thm:spre}
Assume that $\bTheta^* =\bSigma^{-1} \in \cU(s,M)$.
Let $\bGamma \in \RR^{d\times d}$ be as in Theorem \ref{thm:element-wise} and let $\lambda$ satisfy
\$
\lambda =  4 C \|\Vb \|_{\max}  \sqrt{\frac{2\log d + \log \delta^{-1} }{\lfloor n/2\rfloor }} ~\mbox{ for some } C \geq  M.
\$
Assume  Condition~\ref{REC} is fulfilled with $k=s$ and $\bGamma$ specified above. Then with probability at least $1-2\delta$,  we have 
\$
	\|\hbP-\bTheta^* \|_{\text F}
	&	\leq  6\, C \kappa_-^{-1} \|\Vb \|_{\max} \,s^{1/2} \sqrt{\frac{2 \log d + \log \delta^{-1} }{\lfloor n/2\rfloor }}.
\$
\end{theorem}

\begin{remark}
The nonasymptotic probabilistic bound in Theorem~\ref{thm:spre} is established under the assumption that  Condition~\ref{REC} holds.
It can be shown that Condition~\ref{REC} is satisfied with high probability as long as the coordinates of $\bX$ have bounded fourth moments. The proof is based on  an argument  similar to that in the proof of Lemma 4 in the work of \cite{sun2017adaptive}, and thus is omitted here. 
\end{remark}

\section{Numerical study}\label{sec:5}

In this section, we assess the numerical performance of proposed tail-robust covariance estimators. We consider the  {\it element-wise truncated covariance estimator} $\hat{\bSigma}_{1}^{\TT}$ defined in \eqref{element-wise.censored.est}, the {\it spectrum-wise  truncated  covariance estimator} $\hat\bSigma_{2}^{\TT}$ defined in \eqref{U.cov.estimator}, the {\it Huber-type $M$-estimator} $\hat\bSigma_{1}^{\HH}$ given in \eqref{Huber1} and the {\it adaptive Huber $M$-estimator} $\hat\bSigma^{\HH}_3$ in Section \ref{sec:adHuber}.

Throughout this section, we let $\{\tau_{k\ell}\}_{1\leq k,\ell \leq d}=\tau$ for $\hat\bSigma_{1}^{\HH}$.  To compute $\hat\bSigma_{2}^{\TT}$ and $\hat\bSigma_{1}^{\HH}$,  the robustification parameter $\tau$ is selected by five-fold cross-validation. The robustification parameters $\{\tau_{k\ell}\}_{1\leq k,\ell \leq d}$ for $\hat{\bSigma}_{1}^{\TT}$ are  tuned by solving the equation \eqref{empirical.censored.moment.equation}, and thus is an adaptive elementwise-truncated estimator.
To implement the {\it adaptive Huber $M$-estimator} $\hat\bSigma^{\HH}_3$, we calibrate $\{\tau_{k\ell}\}_{1\leq k,\ell \leq d}$ and estimate $\{\sigma_{k\ell}\}_{1\leq k,\ell \leq d}$ simultaneously by solving the equation system \eqref{tau.M-view} as described in Algorithm \ref{algm_2}.

We first generate a data matrix $\Yb\in \mathbb{R}^{n\times d}$ with rows being i.i.d. vectors from a  distribution with mean ${\textbf 0}$ and covariance matrix $\Ib_d$.  We then rescaled the data and set $\Xb =\Yb \bSigma^{1/2}$ as the final data matrix, where $\bSigma \in \mathbb{R}^{d\times d}$ is a structured covariance matrix.  
We consider four distribution models outlined below:
\begin{itemize}
\item[(1)](Normal model). The rows of $\Yb$ are i.i.d. generated from the standard normal distribution.  
\item[(2)] (Student's $t$ model). $\Yb=\Zb/\sqrt{3}$, where the entries of $\Zb$ are i.i.d. with Student's distribution with 3 degrees of freedom. 
\item[(3)] (Pareto model). $\Yb=4\Zb/3$, where the entries of $\Zb$ are i.i.d. with Pareto distribution with shape parameter 3 and scale parameter 1.  
\item[(4)] (Log-normal model). $\Yb =\exp\{0.5+ \Zb\}/(e^3-e^2)$, where the entries of $\Zb$ are i.i.d. with standard normal distribution.
\end{itemize}

The covariance matrix $\bSigma$ has one of the following three structures:
\begin{itemize}
\item[(a)](Diagonal structure). $\bSigma=\Ib_d$;
\item[(b)](Equal correlation structure). $\sigma_{k \ell }=1$ for $k=\ell$ and $\sigma_{k \ell }=0.5$ when $k \neq \ell$;
\item[(c)](Power decay structure). $\sigma_{k \ell}=1$ for $k=\ell$ and $\sigma_{ k  \ell } = 0.5^{|k-\ell|}$ when $k  \neq  \ell$.
\end{itemize}

In each setting, we choose $(n, d)$ as $(50, 100)$, $(50, 200)$ and $(100, 200)$, and simulate 200 replications for each scenario. The performance is assessed by the relative mean error (RME) under spectral, max or Frobenius norm:
\begin{equation*}
\text{RME}= \frac{\sum_{i=1}^{200}\| \widehat{\bSigma}_i - \bSigma \|_{2,\max, \F}}{\sum_{i=1}^{200}\| \widetilde{\bSigma}_i - \bSigma \|_{2,\max, \F}} ,
\end{equation*}
where $\widehat{\bSigma}_i$ is the estimate of $\bSigma$ in the $i$th simulation using one of the four robust methods and $\widetilde{\bSigma}_i$ denotes the sample covariance estimate that serves as a benchmark. 
The smaller the RME is, the more improvement the robust method achieves.

Tables \ref{Tab_sim_1}--\ref{Tab_sim_3} summarize the simulation results, which indicate that  all the robust estimators outperform the sample covariance matrix by a visible margin when data are generated from a heavy-tailed or an  asymmetric distribution. On the other hand, the proposed estimators perform almost as good as the sample covariance matrix when the data follows a normal distribution, indicating high efficiency in this case. The performance of the four robust estimators are comparable in all scenarios: the {\it spectrum-wise  truncated  covariance estimator} $\hat\bSigma_{2}^{\TT}$ has the smallest RME under spectral norm, while the other three estimators perform better under max and Frobenius norms. This outcome is inline with our intuition discussed in Section \ref{sec:general}.
Furthermore, the computationally efficient {\it adaptive Huber $M$-estimator} $\hat\bSigma^{\HH}_3$ performs comparably as the {\it Huber-type $M$-estimator} $\hat\bSigma_{1}^{\HH}$ where the robustification parameters are chosen by cross-validation.

\begin{table}[htbp]
\small
\begin{center}
\caption{RME under diagonal structure.}
\label{Tab_sim_1}
\begin{tabular}{c|ccc|ccc|ccc|ccc}
\hline\hline
\multicolumn{13}{c}{$n=50, \ p=100$}
\\
\multicolumn{1}{c}{} & \multicolumn{3}{c}{Normal} & \multicolumn{3}{c}{$t_3$} & \multicolumn{3}{c}{Pareto} & \multicolumn{3}{c}{Log-normal}
\\\hline
 
 & 2 & max & F  & 2 & max & F  & 2 & max & F  & 2 & max & F 
\\
$\hat\bSigma_{1}^{\HH}$  & 0.97 & 0.95 & 0.98 & 0.37 & 0.39 & 0.65 
		   & 0.27 & 0.21 & 0.47 & 0.27 & 0.21 & 0.51 
\\
$\hat \bSigma^{\HH}_3$ & 0.97 & 0.90 & 0.96 & 0.37 & 0.36 & 0.59 
		   & 0.29 & 0.24 & 0.45 & 0.24 & 0.19 & 0.49 
\\
$\hat{\bSigma}_{1}^{\TT}$  & 0.97 & 0.91 & 0.96 & 0.40 & 0.38 & 0.62 
           & 0.27 & 0.23 & 0.42 & 0.25 & 0.18 & 0.50 
\\
$\hat\bSigma_{2}^{\TT}$    & 0.96 & 0.99 & 0.98 & 0.34 & 0.41 & 0.67 
            & 0.26 & 0.25 & 0.44 & 0.25 & 0.26 & 0.56 
\\\hline
\multicolumn{13}{c}{$n=50, \ p=200$}
\\
\multicolumn{1}{c}{} & \multicolumn{3}{c}{Normal} & \multicolumn{3}{c}{$t_3$} & \multicolumn{3}{c}{Pareto} & \multicolumn{3}{c}{Log-normal}
\\\hline
 
  & 2 & max & F & 2 & max & F  & 2 & max & F & 2 & max & F 
\\
$\hat\bSigma_{1}^{\HH}$ & 0.98 & 0.95 & 0.98 & 0.32 & 0.29 & 0.60 
		   & 0.29 & 0.23 & 0.41 & 0.24 & 0.20 & 0.43 
\\
$\hat \bSigma^{\HH}_3$ & 0.98 & 0.96 & 0.97 & 0.31 & 0.26 & 0.54 
		   & 0.27 & 0.20 & 0.42 & 0.24 & 0.19 & 0.38 
\\
$\hat{\bSigma}_{1}^{\TT}$  & 0.97 & 0.95 & 0.96 & 0.33 & 0.29 & 0.63 
           & 0.26 & 0.19 & 0.39 & 0.23 & 0.18 & 0.42 
\\
$\hat\bSigma_{2}^{\TT}$    & 0.95 & 0.98 & 0.95 & 0.31 & 0.33 & 0.65 
            & 0.24 & 0.26 & 0.48 & 0.22 & 0.23 & 0.48 
\\\hline
\multicolumn{13}{c}{$n=100, \ p=200$}
\\
\multicolumn{1}{c}{} & \multicolumn{3}{c}{Normal} & \multicolumn{3}{c}{$t_3$} & \multicolumn{3}{c}{Pareto} & \multicolumn{3}{c}{Log-normal}
\\\hline

  & 2 & max & F & 2 & max & F  & 2 & max & F & 2 & max & F 
\\
$\hat\bSigma_{1}^{\HH}$ & 0.99 & 0.98 & 0.99 & 0.40 & 0.47 & 0.58 
		   & 0.46 & 0.49 & 0.51 & 0.32 & 0.20 & 0.47
\\
$\hat \bSigma^{\HH}_3$ & 0.95 & 0.99 & 0.98 & 0.39 & 0.47 & 0.59 
		   & 0.45 & 0.45 & 0.48 & 0.28 & 0.21 & 0.49
\\
$\hat{\bSigma}_{1}^{\TT}$  & 0.97 & 0.94 & 0.97 & 0.38 & 0.45 & 0.57 
           & 0.46 & 0.49 & 0.51 & 0.26 & 0.27 & 0.47
\\
$\hat\bSigma_{2}^{\TT}$    & 0.94 & 1.01 & 0.95 & 0.33 & 0.51 & 0.64 
            & 0.42 & 0.53 & 0.61 & 0.28 & 0.27 & 0.58
\\\hline
\end{tabular}
\end{center}
\begin{singlespace}
\begin{center}
\begin{minipage}{30pc}
{\it 
Mean relative errors of the the four robust estimators $\hat\bSigma_{1}^{\HH}$, $\hat \bSigma^{\HH}_3$, $\hat{\bSigma}_{1}^{\TT}$ and $\hat\bSigma_{2}^{\TT}$ over 200 replications when the true covariance matrix has a diagonal structure. 2, max and F denote the spectral, max and Frobenius norms, respectively.  }
\end{minipage}
\end{center}
\end{singlespace}
\end{table}


\begin{table}[htbp]
\small
\begin{center}
\caption{RME under equal correlation structure. }
\label{Tab_sim_2}
\begin{tabular}{c|ccc|ccc|ccc|ccc}
\hline\hline
\multicolumn{13}{c}{$n=50, \ p=100$}
\\
\multicolumn{1}{c}{} & \multicolumn{3}{c}{Normal} & \multicolumn{3}{c}{$t_3$} & \multicolumn{3}{c}{Pareto} & \multicolumn{3}{c}{Log-normal}
\\\hline
 
 & 2 & max & F & 2 & max & F  & 2 & max & F  & 2 & max & F 
\\
$\hat\bSigma_{1}^{\HH}$ & 0.97 & 0.94 & 0.97 & 0.68 & 0.12 & 0.68 
		   & 0.68 & 0.23 & 0.59 & 0.58 & 0.27 & 0.46 
\\
$\hat \bSigma^{\HH}_3$ & 0.96 & 0.95 & 0.96 & 0.69 & 0.15 & 0.64 
		   & 0.62 & 0.21 & 0.59 & 0.52 & 0.27 & 0.44 
\\
$\hat{\bSigma}_{1}^{\TT}$  & 0.97 & 0.96 & 0.97 & 0.67 & 0.14 & 0.67 
           & 0.64 & 0.22 & 0.57 & 0.59 & 0.28 & 0.47 
\\
$\hat\bSigma_{2}^{\TT}$    & 0.95 & 0.99 & 1.02 & 0.56 & 0.26 & 0.71 
            & 0.62 & 0.27 & 0.60 & 0.50 & 0.33 & 0.51
\\\hline
\multicolumn{13}{c}{$n=50, \ p=200$}
\\
\multicolumn{1}{c}{} & \multicolumn{3}{c}{Normal} & \multicolumn{3}{c}{$t_3$} & \multicolumn{3}{c}{Pareto} & \multicolumn{3}{c}{Log-normal}
\\\hline
 
  & 2 & max & F & 2 & max & F  & 2 & max & F & 2 & max & F 
\\
$\hat\bSigma_{1}^{\HH}$ & 0.97 & 0.94 & 0.98 & 0.77 & 0.21 & 0.76 
		   & 0.67 & 0.34 & 0.50 & 0.69 & 0.23 & 0.67 
\\
$\hat \bSigma^{\HH}_3$ & 1.00 & 0.97 & 0.98 & 0.77 & 0.22 & 0.73 
		   & 0.63 & 0.31 & 0.50 & 0.70 & 0.23 & 0.68 
\\
$\hat{\bSigma}_{1}^{\TT}$  & 0.99 & 0.97 & 0.96 & 0.78 & 0.24 & 0.71 
           & 0.63 & 0.33 & 0.46 & 0.70 & 0.23 & 0.68 
\\
$\hat\bSigma_{2}^{\TT}$    & 0.95 & 0.98 & 1.00 & 0.74 & 0.35 & 0.80 
            & 0.61 & 0.34 & 0.51 & 0.66 & 0.31 & 0.72 
\\\hline
\multicolumn{13}{c}{$n=100, \ p=200$}
\\
\multicolumn{1}{c}{} & \multicolumn{3}{c}{Normal} & \multicolumn{3}{c}{$t_3$} & \multicolumn{3}{c}{Pareto} & \multicolumn{3}{c}{Log-normal}
\\\hline

  & 2 & max & F & 2 & max & F  & 2 & max & F & 2 & max & F 
\\
$\hat\bSigma_{1}^{\HH}$ & 1.00 & 0.96 & 0.99 & 0.79 & 0.23 & 0.78 
		   & 0.63 & 0.46 & 0.57  & 0.53 & 0.21 & 0.47
\\
$\hat \bSigma^{\HH}_3$ & 0.98 & 0.98 & 0.97 & 0.79 & 0.24 & 0.79 
		   & 0.69 & 0.48 & 0.58  & 0.57 & 0.22	& 0.48
\\
$\hat{\bSigma}_{1}^{\TT}$  & 1.00 & 1.00 & 0.99 & 0.78 & 0.21 & 0.77 
           & 0.65 & 0.45 & 0.57 & 0.55	& 0.23	& 0.50 
\\
$\hat\bSigma_{2}^{\TT}$    & 0.97 & 1.02 & 1.03 & 0.73 & 0.32 & 0.83 
            & 0.62 & 0.54 & 0.61 &  0.50 & 0.29 & 0.55 
\\\hline
\end{tabular}
\end{center}
\end{table}


\begin{table}[htbp]
\small
\begin{center}
\caption{RME under power decay structure.}
\label{Tab_sim_3}
\begin{tabular}{c|ccc|ccc|ccc|ccc}
\hline\hline
\multicolumn{13}{c}{$n=50, \ p=100$}
\\
\multicolumn{1}{c}{} & \multicolumn{3}{c}{Normal} & \multicolumn{3}{c}{$t_3$} & \multicolumn{3}{c}{Pareto} & \multicolumn{3}{c}{Log-normal}
\\\hline
 
 & 2 & max & F & 2 & max & F  & 2 & max & F  & 2 & max & F 
\\
$\hat\bSigma_{1}^{\HH}$ & 0.98 & 0.95 & 0.98 & 0.58 & 0.30 & 0.71 
		   & 0.48 & 0.29 & 0.57 & 0.69 & 0.39 & 0.79 
\\
$\hat \bSigma^{\HH}_3$ & 0.95 & 0.95 & 0.93 & 0.58 & 0.28 & 0.72 
		   & 0.48 & 0.26 & 0.58 & 0.70 & 0.39 & 0.78 
\\
$\hat{\bSigma}_{1}^{\TT}$  & 0.97 & 0.98 & 0.96 & 0.59 & 0.30 & 0.71 
           & 0.49 & 0.26 & 0.57 & 0.72 & 0.39 & 0.77 
\\
$\hat\bSigma_{2}^{\TT}$    & 0.98 & 0.98 & 0.99 & 0.52 & 0.33 & 0.77 
            & 0.47 & 0.31 & 0.60 & 0.66 & 0.45 & 0.81 
\\\hline
\multicolumn{13}{c}{$n=50, \ p=200$}
\\
\multicolumn{1}{c}{} & \multicolumn{3}{c}{Normal} & \multicolumn{3}{c}{$t_3$} & \multicolumn{3}{c}{Pareto} & \multicolumn{3}{c}{Log-normal}
\\\hline
 
  & 2 & max & F & 2 & max & F  & 2 & max & F & 2 & max & F 
\\
$\hat\bSigma_{1}^{\HH}$ & 0.98 & 0.95 & 0.97 & 0.58 & 0.30 & 0.71 
		   & 0.48 & 0.29 & 0.57 & 0.69 & 0.39 & 0.79 
\\
$\hat \bSigma^{\HH}_3$ & 0.96 & 0.93 & 0.95 & 0.56 & 0.29 & 0.66 
		   & 0.49 & 0.26 & 0.55 & 0.72 & 0.38 & 0.77 
\\
$\hat{\bSigma}_{1}^{\TT}$  & 0.98 & 0.97 & 0.97 & 0.59 & 0.27 & 0.71 
           & 0.48 & 0.26 & 0.58 & 0.70 & 0.36 & 0.80 
\\
$\hat\bSigma_{2}^{\TT}$    & 0.98 & 0.98 & 1.01 & 0.54 & 0.24 & 0.76 
            & 0.41 & 0.31 & 0.60 & 0.68 & 0.42 & 0.82 
\\\hline
\multicolumn{13}{c}{$n=100, \ p=200$}
\\
\multicolumn{1}{c}{} & \multicolumn{3}{c}{Normal} & \multicolumn{3}{c}{$t_3$} & \multicolumn{3}{c}{Pareto} & \multicolumn{3}{c}{Log-normal}
\\\hline

  & 2 & max & F & 2 & max & F  & 2 & max & F & 2 & max & F 
\\
$\hat\bSigma_{1}^{\HH}$ & 0.99 & 0.98 & 1.00 & 0.45 & 0.25 & 0.66 
		   & 0.42 & 0.31 & 0.54 & 0.48	& 0.35 & 0.62

\\
$\hat \bSigma^{\HH}_3$ & 0.98 & 0.98 & 0.99 & 0.47 & 0.26 & 0.68 
		   & 0.41 & 0.30 & 0.53 & 0.47 & 0.34	& 0.61

\\
$\hat{\bSigma}_{1}^{\TT}$  & 1.00 & 0.99 & 1.00 & 0.50 & 0.30 & 0.68 
           & 0.41 & 0.34 & 0.56 & 0.49	& 0.38	& 0.64

\\
$\hat\bSigma_{2}^{\TT}$   & 0.99 & 1.04 & 1.01 & 0.41 & 0.31 & 0.70 
            & 0.40 & 0.39 & 0.59 & 0.43 & 0.43 & 0.69

\\\hline
\end{tabular}
\end{center}
\end{table}

\section{Discussion}\label{sec:6}

In this paper, we surveyed and unified selected recent results on covariance estimation for heavy-tailed distributions. More specifically, we proposed  element-wise and spectrum-wise truncation techniques to robustify the sample covariance matrix.  The robustness, referred to as the {\it tail robustness}, is demonstrated  by finite-sample deviation analysis  in the presence of heavy-tailed data: the proposed estimators  achieve exponential-type deviation bounds under mild moment conditions. 
We emphasize that the tail robustness is different from the classical notion of robustness that is often characterized by the breakdown point \citep{hampel1971general}. 
Nevertheless, it does not provide any information on the convergence properties of an estimator, such as consistency and efficiency. Tail robustness is a concept that combines robustness, consistency, and finite-sample error bounds.

We discussed three types of procedures in Section~\ref{sec:general}: truncation-based  methods,  their $M$-estimation counterparts and the median-of-means method. Truncated estimators have closed-form expressions and therefore are easy to implement in practice. The corresponding $M$-estimators achieve comparable sub-Gaussian-type error bounds, which are of the order $ \sqrt{\log(d/\delta)/n}$ under the max norm and of order $\sqrt{  {\rm r}(\bSigma)\log(d/\delta)/n} $ under the spectral norm, but with sharper moment-dependent constants. Computationally, $M$-estimators can be efficiently evaluated via gradient descent method or iteratively reweighted least squares algorithm.
Both truncated and $M$-estimators involve robustification parameters that need to be calibrated to fit the noise level of the problem. Adaptation and tuning of these parameters are discussed in Section~\ref{sec:3.5}. The MOM estimator proposed in Section~\ref{sec2.1.1} is tuning-free because the number of blocks depends neither on noise level nor on confidence level.
Following the terminology proposed by \cite{DLLO2016}, truncation-based estimators are $\delta$-dependent estimators as they depend on the confidence level $1-\delta$ at which one aims to control, while the MOM estimator achieves sub-Gaussian error bounds simultaneously at all confidence levels in a certain range but requires slightly stronger assumptions, namely, the existence of sixth moments instead of fourth.

Three examples discussed in Section~\ref{sec:4} illustrate that both element-wise and spectrum-wise truncated covariance estimators can serve as building blocks for a variety of estimation problems in high dimensions. A natural question is  whether one can construct a single robust estimator that achieves exponentially fast concentration both element-wise and spectrum-wise, that is, satisfies the results in Theorems \ref{thm:element-wise} and \ref{thm:u-type} simultaneously. Here we discuss a theoretical solution to this question. In fact, one can arbitrarily pick one element, denoted as $\hat \bSigma^{\TT}$, from the collection of matrices 
\#
	\mathcal{H}=\Bigg\{ \Sb \in \RR^{d\times d} : ~&\Sb = \Sb^\T,  ~\| \hat \bSigma_2^{\TT} -\Sb \|_2 \leq  2 v \sqrt{\frac{\log(2d) + \log \delta^{-1}}{m}}  \nn \\
	& \mbox{and }~  \| \hat{\bSigma}^{\TT}_1 - \Sb  \|_{\max}   \leq  2 \|\Vb \|_{\max}  \sqrt{\frac{2\log d + \log \delta^{-1} }{m }   }\Bigg\}.   \nn
\#
Due to Theorems \ref{thm:element-wise} and \ref{thm:u-type}, with probability at least $1-3\delta$, the set $\mathcal{H}$ is non-empty since  it contains the true covariance matrix $\bSigma$. Therefore, it follows from the the triangle inequality that the inequalities
\#
	\| \hat \bSigma^{\TT} -\bSigma \|_2 \leq  4 v \sqrt{\frac{\log(2d) + \log \delta^{-1}}{m}}, \nn  \\
\mbox{and } ~ \| \hat \bSigma^{\TT}  - \bSigma \|_{\max}   \leq  4 \|\Vb \|_{\max}  \sqrt{\frac{2\log d + \log \delta^{-1} }{m }   } . \nn
\#
hold simultaneously with probability at least $1-3\delta$.

\newpage
\section*{Supplementary Material} 
\appendix

In Sections~\ref{proof:lower bound}--\ref{proof:sec:4}, we provide proofs of all the theoretical results in the main text. In addition, we investigate robust covariance estimation and inference under factor models in Section~\ref{sec:factor.model}, which might be of independent interest.

\section{Proof of Proposition~\ref{thm:element-wise:lower}}
\label{proof:lower bound}

Without loss of generality we assume $\bmu={\bf 0}$. We construct a random vector $\bX \in \RR^d$ that follows the distribution below:
$$
\PP\Big\{ \bX=(0,\ldots,0, \underbrace{n\eta}_{j\text{th}},0,\ldots,0)^\intercal\Big\}=\PP\Big\{ \bX=(0,\ldots,0, \underbrace{-n\eta}_{j\text{th}},0,\ldots,0)^\intercal\Big\}=\frac{\sigma^2}{2n^2\eta^2}
$$
for each $j=1,\ldots,d$, and 
$$\PP( \bX={\bf 0} )=1-\frac{d\sigma^2}{n^2\eta^2}.$$
Here we assume  $\eta^2>d\sigma^2/n^2$ so  that $\PP ( \bX={\bf 0} )>0$. In other words, the number of non-zero elements of $\bX$ is at most $1$. It is easy to see that the mean and covariance matrix of $\bX$ are ${\bf 0}$ and $\sigma^2\Ib_d$, respectively.

Consider the empirical mean $\bar{\bX}= (1/n)\sum_{i=1}^n\bX_i $, where $\bX_1,\ldots,\bX_n$ are i.i.d. from $\bX$. It follows that 
\#
	\PP (\|\bar\bX\|_{\infty}\geq \eta )&\geq \PP\Big(\text{exactly one of the } n \text{ samples is not equal to } {\bf 0}\Big) \nn \\
	&= \frac{d\sigma^2}{n\eta^2}\bigg(1-\frac{d\sigma^2}{n^2\eta^2}\bigg)^{ n-1  }. \nn
\#
Therefore, as long as $\delta<(2e)^{-1}$, the following bound
$$ \|\bar\bX\|_{\infty}\geq \sigma \sqrt{\frac{d }{n \delta}}\bigg(1-\frac{2e\delta}{n}\bigg)^{(n-1)/2}$$
holds with probability at least $\delta$, as claimed.
\qed

\section{Proofs for Section~\ref{sec:general}} 
\label{sec:A1}

 \subsection{Proof of Theorem~\ref{thm:element-wise}}
 
 For each $1\leq k\leq \ell\leq d$, note that $\hat{\sigma}^{\TT}_{1,k\ell}$ is a $U$-statistic with a bounded kernel of order two, say  $\hat{\sigma}^{\TT}_{1,k\ell} = {n \choose 2}^{-1} \sum_{1\leq i<j\leq n} h_{k\ell}(\bX_i, \bX_j)$. According to \cite{H1963}, $\hat{\sigma}_{1,k\ell}^{\TT}$ can be represented as an average of (dependent) sums of independent random variables. Specifically, define
$$
	W(\bx_1,\ldots,\bx_n) = \frac{h_{k\ell}(\bx_1,\bx_2) + h_{k\ell}(\bx_3, \bx_4) + \cdots + h_{k\ell}(\bx_{2m-1}, \bx_{2m})}{m}  
$$
for $\bx_1,\ldots, \bx_n \in \RR^d$.
Let $\sum_{\cP}$ denote the summation over all $n!$ permutations $(i_1,\ldots, i_n)$ of $[n]:= \{ 1, \ldots, n \}$ and $\sum_{\cC}$ denote the summation over all ${n \choose 2}$ pairs $(i_1, i_2)$ ($i_1<i_2$) from $[n]$. Then we have $m \sum_{\cP} W(\bx_1,\ldots,\bx_n)  = m 2! (n-2)! \sum_{\cC} h_{k\ell}(\bx_{i_1}, \bx_{i_2})$ and hence
\#
	\hat{\sigma}_{1,k\ell}^{\TT} = \frac{1}{n!}  \sum_{\cP} W(\bX_{i_1},\ldots,\bX_{i_n}) . \label{U.representation}
\#

Write $\tau = \tau_{k\ell}$ and $v= v_{k\ell}$ for simplicity. For any $y>0$, by Markov's inequality,  \eqref{U.representation}, convexity and independence, we derive that
\#
	&	\PP (  \hat{\sigma}_{1,k\ell}^{\TT}  - \sigma_{k\ell} \geq y  ) \leq e^{-  (m/\tau) (y+ \sigma_{k\ell})    } \EE e^{  (m/\tau) \hat{\sigma}_{1,k\ell}^{\TT}    } \nn \\
& \leq e^{-  (m/\tau)  (y+ \sigma_{k\ell})    } \frac{1}{n!}\sum_{\cP} \EE e^{ (1/\tau) \sum_{j=1}^m h_{k\ell} (\bX_{i_{2j-1}}, \bX_{i_{2j}}) } \nn \\
& = e^{-   (m/\tau)  (y+ \sigma_{k\ell})    } \frac{1}{n!}\sum_{\cP} \prod_{j=1}^m  \EE e^{  (1/\tau)   h_{k\ell} (\bX_{i_{2j-1}}, \bX_{i_{2j}}) }. \nn
\#
Note that $ h_{k\ell} (\bX_{i_{2j-1}}, \bX_{i_{2j}}) = \psi_\tau ( Y_{\pi k} Y_{ \pi \ell}/2 ) = \tau \psi_1( Y_{\pi k} Y_{ \pi \ell}/(2\tau))$ for some $1\leq \pi \leq N$.  In addition, it is easy to verify the inequality that
\# 
-\log (1-x+x^2)\leq \psi_1(x)\leq  \log (1+x+x^2)  ~\mbox{ for all } x\in \RR. \label{basic.ineq}
\#
Therefore, we have
\#
&	\EE e^{   (1/\tau)    h_{k\ell} (\bX_{i_{2j-1}}, \bX_{i_{2j}}) }   \leq \EE \{ 1 + Y_{\pi k} Y_{ \pi \ell}/(2\tau) + (Y_{\pi k} Y_{ \pi \ell})^2/(2\tau)^2 \} \nn \\
	&  = 1 +  \sigma_{k\ell} / \tau  +(1/\tau)^2\EE( Y_{\pi k} Y_{ \pi \ell}/2)^2  \leq e^{  \sigma_{k\ell} / \tau   + (v / \tau)^2 } . \nn
\#
Combining the above calculations gives
\#
\PP (  \hat{\sigma}_{1,k\ell}^{\TT}  - \sigma_{k\ell} \geq y  ) \leq e^{-  (m/\tau)  y +  m (v/\tau)^2} =  e^{-m y^2/(4v^2)} , \nn
\#
where the equality holds by taking $\tau = 2v^2/y$. Similarly, it can be shown that $\PP (  \hat{\sigma}_{1,k\ell}^{\TT}   - \sigma_{k\ell} \leq - y  ) \leq  e^{-m y^2/(4v^2)} $. Consequently, for $\delta \in (0,1)$, taking $y= 2v \sqrt{ ( 2\log d +  \log  \delta^{-1} ) /m}$, or equivalently, $ \tau = v\sqrt{m/ (2\log d+  \log \delta^{-1} )}$, we arrive at
\#
	 \PP\Bigg(  | \hat{\sigma}_{1,k\ell}^{\TT}   - \sigma_{k\ell}  | \geq 2v \sqrt{\frac{2 \log d +\log  \delta^{-1}}{m}}   \Bigg) \leq  \frac{2\delta }{d^2}.   \nn
\#
From the union bound it follows that
\$
	\PP\Bigg( \| \hat{\bSigma}_1^{\TT}  - \bSigma \|_{\max} > 2 \max_{1\leq k,\ell \leq d} v_{k\ell}\sqrt{\frac{ 2 \log d + \log  \delta^{-1}}{n}}  \Bigg) \leq (1+d^{-1}) \delta.
\$
This proves \eqref{sigma1.bound}.  \qed

\subsection{Proof of Theorem~\ref{thm:u-type}}

To begin with, note that $\hat \bSigma_2^{\TT}$ can be written as a $U$-statistic of order 2. Define the index set $\cI = \{ (i , j): 1\leq i <  j \leq n \}$ with cardinality $\binom{n}{2}$. Let $h(\bX_i, \bX_j)= (\bX_i-\bX_j)(\bX_i-\bX_j)^\T/2 $ and $\Zb_{i, j}= \tau^{-1} \psi_\tau(h(\bX_i, \bX_j))= \psi_{1}(\tau^{-1}h(\bX_i, \bX_j))$, such that
\$
\wt \bSigma  :=  \frac{1}{\tau} \hat \bSigma_2^{\TT} = \frac{1}{ \binom{n}{2}}\sum_{(i, j) \in \cI } \Zb_{i, j}.
\$

We now rewrite the $U$-statistic $\wt \bSigma$ as a convex combination of sums of independent random matrices. As in the proof of Theorem~\ref{thm:element-wise}, we define
\$
\Wb_{(1, \ldots, n)} = m^{-1} ( \Zb_{1,2}+ \Zb_{3,4}+\ldots+ \Zb_{2m-1, 2m} ).
\$
For every permutation $\pi=(i_1, \ldots, i_n)$, we adopt the notation $\Wb_\pi=\Wb_{(i_1, \ldots, i_n)}$ such that $\wt\bSigma^\tau = (n !)^{-1}\sum_{\pi\in\cP} \Wb_\pi$. Using the convexity of the mappings $\Ab \mapsto \lambda_{\max}(\Ab)$ and $x\mapsto e^x$, we obtain that
\$
\exp \{ \lambda_{\max} (\wt \bSigma -  \bSigma^\tau ) \} \leq  \frac{1}{n!} \sum_{\pi\in\cP}\exp \{ \lambda_{\max}  ( \Wb_\pi-\bSigma^\tau ) \},
\$
where $\bSigma^\tau : =\tau^{-1}\bSigma$. Combined with Markov's inequality and the inequality $e^{\lambda_{\max}(\Ab)}\leq \tr e^\Ab$, this further implies
\$
	& \PP \{ \sqrt{m} \,\lambda_{\max} (  \hat \bSigma_2^{\TT}  -\bSigma )\geq  y  \}  =\PP\Big\{ e^{ \lambda_{\max} (m \wt\bSigma- m\bSigma^\tau ) } \geq  e^{  y \sqrt{m} / \tau  }  \Big\} \\
&\leq  e^{-  y \sqrt{m} / \tau } \frac{1}{n!} \sum_{\pi\in \cP}\EE\exp  \{ \lambda_{\max} (m \Wb_\pi- m \bSigma^\tau ) \} \\
&\leq  e^{-  y \sqrt{m} / \tau  } \frac{1}{n!} \sum_{\pi\in \cP}\EE \tr\exp  (m \Wb_\pi-m \bSigma^\tau ).
\$
For every $\pi=( i_1, \ldots, i_n) \in \cP$, define $\Zb_{\pi , j}= \Zb_{i_{2j-1}, i_{2j}}$ and $\Hb_{\pi , j}=h(\bX_{i_{2j-1}}, \bX_{i_{2j}})$, such that $\Zb_{\pi ,1}, \ldots, \Zb_{\pi, m}$ are independent and $\EE \Hb_{\pi , j}=\bSigma$. Then $\Wb_\pi$ can be written as $\Wb_\pi=  m^{-1} ( \Zb_{\pi , 1}+\ldots+ \Zb_{\pi , m} )$. Recall that $\psi_\tau(x)=\tau\psi_1(x/\tau)$. In view of \eqref{basic.ineq}, we have the matrix inequality
\$
-\log( \Ib  -  \tau^{-1}\Hb_{\pi  , j} +  \tau^{-2 } \Hb^2_{\pi  , j}  ) \preceq   \Zb_{\pi , j}\preceq \log( \Ib +  \tau^{-1} \Hb_{\pi,  j}  +  \tau^{-2} \Hb_{\pi , j}^2  ).
\$
Then we can bound $\EE\exp \tr (m\Wb_\pi- m \bSigma^\tau )$ by
\#\label{thm:utype:eq:1}
&\EE_{[m-1]}\EE_m\tr\exp \bigg( \sum_{j=1}^{m-1}{ \Zb_{\pi , j}}- m\bSigma^\tau+{ \Zb_{\pi , m}}\bigg)\nn\\
&\hspace{1cm}\leqslant \EE_{[m-1]}\EE_m\tr\exp \bigg\{ \sum_{j=1}^{m-1}{ \Zb_{\pi , j}}-m \bSigma^\tau+\log( \Ib + \tau^{-1} \Hb_{\pi  ,m} + \tau^{-2}\Hb_{\pi ,  m}^2 )\bigg\},
\#
where the expectation $\EE_m$ is taken with respect to $\{\bX_{i_{2m-1}}, \bX_{i_{2m}}\}$ and the expectation $\EE_{[m-1]}$ is taken with respect to $\{\bX_{i_{1}},..., \bX_{i_{2m-2}}\}$. To bound the right-hand side of \eqref{thm:utype:eq:1}, we follow a similar argument as in \cite{M2016}. By Lieb's concavity theorem (see, e.g. Fact 2.5 in \cite{M2016}) and Jensen's inequality, we arrive at
\$
& \EE\tr\exp ( m \Wb_\pi-m \bSigma^\tau ) \nn \\
&\leq  \EE\tr\exp \bigg\{\sum_{j=1}^{m-1}{\Zb_{\pi , j}}-m \bSigma^\tau+\log(\Ib +  \tau^{-1} \EE \Hb_{\pi ,m} + \tau^{-2} \EE \Hb_{\pi ,m}^2 )\bigg\}\\
&\leq \tr \exp\bigg\{\sum_{j=1}^m \log (\Ib+ \tau^{-1} \EE \Hb_{\pi , j}  + \tau^{-2} \EE \Hb_{\pi ,j}^2 ) - m\bSigma^\tau\bigg\}\\
&\leq  \tr \exp\bigg( \frac{1}{\tau^2} \sum_{j=1}^m\EE \Hb_{\pi ,j}^2 \bigg) \\
& \leq d\exp( m \tau^{-2} \|  \EE \Hb_{\pi , 1}^2 \|_2  )\\
&= d\exp( m \tau^{-2} v^2 ),
\$
where we used the bound $ \tr e^{\Ab} \leq d e^{\| \Ab \|}$ in the last inequality and the definition $v^2$ from \eqref{v2.def} in the last equality.

Letting $\tau = 2v^2\sqrt{m}/y$, we get
\$
& \PP \{  \sqrt{m} \, \lambda_{\max} (  \hat \bSigma_2^{\TT}  -\bSigma )\geq y \}   \leq  d \exp\bigg(-\frac{ y\sqrt{m}}{\tau } + \frac{m v^2}{\tau^2} \bigg) \leq d e^{-y^2/(4v^2)}.
\$
Similarly, it can be shown that
\$
\PP \{ \sqrt{m} \, \lambda_{\min} (  \hat \bSigma_2^{\TT} -\bSigma )\leq  - y \} \leq d  e^{-y^2/(4v^2)}.
\$

Finally, taking $y = 2 v\sqrt{\log(2d)+ \log \delta^{-1}}$ in the last two displays proves \eqref{spectral.concentration}. \qed

\subsection{Proof of Theorem~\ref{thm:huber-type}}

Let $v_{\max} = \max_{1\leq k,\ell \leq d} v_{k\ell}$. By the union bound, for any $y>0$ it holds
\#
	& \PP( \| \hat{\bSigma}_1^{\HH} - \bSigma \|_{\max} \geq v_{\max} \, y  ) \nn \\
& \leq \sum_{1\leq k\leq \ell \leq d}  \PP( | \hat{\sigma}^{\HH}_{1, k\ell} - \sigma_{k \ell} | \geq v_{k\ell} \,y )   \leq \frac{d(d+1)}{2} \max_{1\leq k\leq \ell \leq d } \PP( | \hat{\sigma}_{1,k\ell}^{\HH} - \sigma_{k \ell} | \geq v_{k\ell} \,y ).\label{union.bound}
\#
In the rest of the proof, we fix $(k,\ell) \in [d] \times [d]$ and write $\tau = \tau_{k\ell}$ and $v = v_{k\ell}$ for simplicity. Moreover, define the index set $\cI= \{ (i,j) : 1\leq i <  j \leq n\}$, the collection $\{Z_{i,j} = (X_{ik}- X_{jk})(X_{i\ell} - X_{j \ell} )/2 : (i,j) \in \cI \}$ of random variables indexed by $\cI$ and the loss function $\cL(\theta) = \sum_{(i,j) \in \cI} \ell_{\tau}(Z_{i,j} - \theta)$. With this notation, we have $$\hat \sigma_{1,k\ell}^{\HH}  = \hat{\theta} := \argmin_{\theta \in \RR} \cL(\theta).$$ Without loss of generality, we assume $\bmu = (\mu_1,\ldots, \mu_d)^\T = \textbf{0}$; otherwise, we can simply replace $X_{ik}$ by $X_{ik}-\mu_k$ for all $i\in [n]$ and $k\in [d]$.

Note that $\hat{\theta}$ is the unique solution of the equation
$$
	\Psi(\theta):= \frac{1}{{n \choose 2}}\sum_{(i,j) \in \cI} \psi_\tau(Z_{i,j} - \theta) = 0 ,
$$
where $\psi_\tau(\cdot)$ is defined in \eqref{psi.def}. Similarly to the proof of Theorem~\ref{thm:element-wise}, we define
\$
 w_{(1, \ldots, n) }(\theta)= \frac{1}{m \tau} \{ \psi_\tau( Z_{1,2} -\theta) + \psi_\tau( Z_{3,4}-\theta ) + \ldots + \psi_\tau( Z_{2m-1, 2m} -\theta)  \} .
\$
Denote by $\cP$ the class of all $n!$ permutations on $[n]$ and let $\pi=(i_1, \ldots, i_n)$ be a permutation, i.e., $\pi(j)=i_j$ for $j=1,\ldots, n$. Put $w_\pi (\theta)=w_{(i_1,\ldots, i_n)}(\theta)$ for $\pi \in \cP$, such that $\tau^{-1}m\Psi(\theta) = (n !)^{-1}\sum_{\pi\in\cP} m w_\pi(\theta)$. By convexity, we have
\#
	\EE \{ e^{  \tau^{-1}m\Psi(\theta)} \}  \leq \frac{1}{n !} \sum_{\pi \in \cP}   \EE \{ e^{ m w_\pi(\theta) } \}.\nn
\#
Recall that $\EE Z_{i,j} = \sigma_{ k \ell}$ for any $(i,j) \in \cI$.  By \eqref{v.cond}, 
$$
	v^2 = \var(Z_{1,2})  
	 = \frac{1}{2} \{  \EE ( ( X_{k} - \mu_k)^2 ( X_\ell - \mu_\ell)^2 ) + \sigma_{kk} \sigma_{\ell \ell}   \}.
$$ 
For $\pi = (1,\ldots, n)$, by \eqref{basic.ineq} and the fact that $\tau^{-1} \psi_\tau(x) =  \psi_1(x/\tau)$, we have
\#
	& \EE \{ e^{ m w_\pi(\theta) } \} = \prod_{j=1}^m \EE \exp\{ \psi_1( (Z_{2j-1,2j} - \theta)/\tau )  \} \nn \\
	& \leq  \prod_{j=1}^m  \EE   \{ 1 + \tau^{-1} (Z_{2j-1,2j} - \theta)  + \tau^{-2} (Z_{2j-1,2j} - \theta)^2 \} \nn \\
	& \leq  \prod_{j=1}^m  [ 1 + \tau^{-1} ( \sigma_{ k \ell } - \theta)  + \tau^{-2} \{ v^2 + ( \sigma_{k\ell}  - \theta )^2 \} ] \nn \\
	& \leq \exp[ m \tau^{-1} ( \sigma_{ k \ell } - \theta)  + m \tau^{-2} \{ v^2 + ( \sigma_{k\ell}  - \theta )^2 \} ]. \label{exp.moment.bd1}
\#
Similarly, it can be shown that 
\#
	\EE \{ - e^{ m w_\pi(\theta) } \} \leq \exp[ - m \tau^{-1} ( \sigma_{ k \ell } - \theta)  + m \tau^{-2} \{ v^2 + ( \sigma_{k\ell}  - \theta )^2 \} ]. \label{exp.moment.bd2}
\#
Inequalities \eqref{exp.moment.bd1} and \eqref{exp.moment.bd2} hold for every permutation $\pi \in \cP$. For $\eta \in (0,1)$, define
\#
	B_+(\theta) =   \sigma_{ k \ell } - \theta   +   \frac{ v^2 + ( \sigma_{k\ell}  - \theta )^2 }{\tau} + \frac{\tau \log \eta^{-1} }{m} , \nn \\
	B_-(\theta) =   \sigma_{ k \ell } - \theta  -  \frac{ v^2 + ( \sigma_{k\ell}  - \theta )^2 }{\tau} - \frac{\tau \log \eta^{-1} }{m}. \nn
\#
Together, \eqref{exp.moment.bd1}, \eqref{exp.moment.bd2} and Markov's inequality imply
\#
 \PP\{ \Psi(\theta)  > B_+(\theta)\} \leq  e^{-  \tau^{-1} m B_+(\theta)} \EE\{ e^{ \tau^{-1} m\Psi(\theta)} \} \leq \eta , \nn \\
\mbox{and }~ \PP\{ \Psi(\theta)  < B_-(\theta)\} \leq  e^{-  \tau^{-1} m B_-(\theta)} \EE\{ - e^{ \tau^{-1} m\Psi(\theta)} \} \leq \eta . \nn
\#
Recall that $\Psi(\hat{\theta})=0$. Let $\theta_+$ be the smallest solution of the quadratic equation $B_+(\theta_+)=0$, and $\theta_-$ be the largest solution of the equation $B_-(\theta_-)=0$. Noting that $\Psi(\cdot)$ is decreasing, it follows from the last display that
\#
	\PP ( \theta_- \leq  \hat{\theta} \leq \theta_+ )   \geq 1 - 2 \eta . \nn
\#
Similarly to the proof of Proposition~2.4 in \cite{C2012}, it can be shown that with $\tau = v \sqrt{m/\log  \eta^{-1}  }$,
\#
	\theta_+ \leq \sigma_{k\ell} + 2\bigg(  \frac{v^2}{\tau} + \frac{\tau\log  \eta^{-1}  }{  m} \bigg)  ~\mbox{ and }~ \theta_- \geq \sigma_{k\ell}- 2\bigg(  \frac{v^2}{\tau} + \frac{\tau\log  \eta^{-1}  }{  m} \bigg)  \nn
\#
as long as $m \geq 8 \log  \eta^{-1} $. Consequently, we obtain that with probability at least $1-2\eta $, $ |  \hat{\sigma}^{\HH}_{1,k \ell} - \sigma_{k\ell} | \leq 4 v \sqrt{ ( \log \eta^{-1} ) /m}$. 

Taking $y= 4\sqrt{ (\log  \eta^{-1} ) / m}$ in \eqref{union.bound} yields $\| \hat{\bSigma}_1^{\HH} - \bSigma \|_{\max} \leq 4 v_{\max}  \sqrt{ ( \log  \eta^{-1} )  / m}$ with probability at least $1-  d(d+1) \eta $. Finally, taking $\delta = d^2 \eta $ proves \eqref{max.concentration}. \qed

\subsection{Proof of Theorem~\ref{thm:mom}}
 
We will use Theorem~1 and Lemma~1 in \cite{MS2017} that connect the performance of $\hat\sigma_{\ell m}^{\mathrm{MOM}}$ to the rate of convergence of $\hat\sigma_{\ell m}^{(1)}$ to the normal law. 
It is well known that, whenever the 4th moments of the entries of $\bX$ are finite, $\sqrt{|G_1|} \frac{\hat{\sigma}^{(1)}_{\ell m}  - \sigma_{\ell m}}{\Delta_{\ell m}}$ converges in distribution to the standard normal distribution.  
The rate of this convergence can be obtained via an analogue of the Berry-Esseen theorem for the sample covariance. Specifically, for any $1\leq \ell ,m\leq d$, we seek an upper bound on
\[
\sup_{t\in \mathbb R}\left| \PP  \Bigg( \sqrt{|G_1|} \frac{\hat{\sigma}^{(1)}_{\ell m}  - \sigma_{\ell m}}{\Delta_{\ell m}}\leq t\Bigg) - 
\PP (Z \leq t )
\right|,
\]
where $Z\sim \mathcal N(0,1)$. To this end, we will  use Theorem 2.9 in \cite{pinelis2016optimal}. 
Using the notation therein, we take $\bV = (X_\ell - \mathbb E X_\ell ,X_m - \mathbb E X_m,X_\ell X_m - \mathbb E (X_\ell X_m))^\intercal$, $f(x_1,x_2,x_3) = x_3 - x_1\cdot x_2$, and deduce that 
\#
\label{BE.bound}
\sup_{t\in \mathbb R}\left| \PP  \Bigg( \sqrt{|G_1|} \frac{\hat{\sigma}^{(1)}_{\ell m}  - \sigma_{\ell m}}{\Delta_{\ell m}}\leq t\Bigg) - 
\PP (Z \leq t ) 
\right| \leq \frac{\mathfrak{C}_{\ell m}}{\sqrt{|G_1|}},
\#
where $\mathfrak{C}_{\ell m}>0$ is a constant depending on $\Delta_{\ell m}$ and $\mathbb E  | (X_\ell - \mathbb E X_\ell ) (X_m - \mathbb E X_m) |^3$. 
Together with Theorem 1 and Lemma 1 of \cite{MS2017}, \eqref{BE.bound} implies that 
\[
\left| \hat\sigma_{\ell m}^{\mathrm{MOM}} - \sigma_{\ell m} \right| \leq 3 \Delta_{\ell m}\, {\sqrt{\frac{k}{n}}}\left( \sqrt{\frac{s}{k}} + \mathfrak{C}_{\ell m}{\sqrt{\frac{k}{n}}} \right)
\]
with probability at least $1 - 4 e^{-2s}$ for all $s>0$ satisfying
\#
\label{eq.s}
\sqrt{\frac{s}{k}} + \mathfrak{C}_{\ell m}{\sqrt{\frac{k}{n}}} \leq 0.33.
\# 
Taking the union bound over all $\ell,m$, we obtain that with probability at least $1 - 2d(d+1) e^{-2s}$,
\[
\| \hat{\bSigma}^{\mathrm{MOM}} - \bSigma \|_{\max} \leq 3 \max_{\ell ,m}\Delta_{ \ell m}\left( \sqrt{\frac{s}{n}} + \max_{\ell ,m}\mathfrak{C}_{\ell m} \,\frac{k}{n}\right)
\]
for all $s>0$ satisfying \eqref{eq.s}. 
The latter is equivalent to the statement of the theorem.  \qed

\subsection{Proof of Corollary \ref{coro:spectrum}}

From the proof of Theorem~\ref{thm:u-type},  we find that 
$$
 \|  \EE \{ (\bX_1 - \bX_2 )(\bX_1 - \bX_2)^\intercal \}^2 \|_2  = 2 \|    \EE \{ (\bX-\bmu)(\bX-\bmu)^\intercal \}^2  + {\rm Tr}(\bSigma) \bSigma + 2 \bSigma^2  \|_2 .
$$
Under the bounded kurtosis condition that $K = \sup_{\bu\in \mathbb{S}^{d-1}} {\rm kurt}(\bu^\T \bX)<\infty$, it follows from Lemma~4.1 in \cite{MW2018} that
$$
 \|    \EE \{ (\bX-\bmu)(\bX-\bmu)^\intercal \}^2\|_2  \leq K {\rm Tr}(\bSigma) \| \bSigma \|_2.
$$
Together, the last two displays imply
$$
	\|  \EE \{ (\bX_1 - \bX_2 )(\bX_1 - \bX_2)^\intercal \}^2 \|_2  \leq  2 \| \bSigma \|_2 \{  (K+1) {\rm Tr}(\bSigma) + 2 \|\bSigma \|_2 \}.
$$
 Taking $v= \|  \EE \{ (\bX_1 - \bX_2 )(\bX_1 - \bX_2)^\intercal \}^2 \|_2^{1/2} /2$ that scales with $ {\rm Tr}(\bSigma)^{1/2} \| \bSigma \|_2^{1/2}={\rm r}(\bSigma)^{1/2} \| \bSigma \|_2$, the  resulting estimator satisfies
\#
 \| \hat{\bSigma}^{\TT}_{2} - \bSigma \|_2 \lesssim  K^{1/2} \| \bSigma \|_2  \sqrt{ \frac{ {\rm r}(\bSigma) ( \log d + t ) }{n}} \label{spec.bound}
\#
with probability at least $1- e^{-t}$.  \qed

\section{Proofs for Section~\ref{sec:4}}
\label{proof:sec:4}
\subsection{Proof of Theorem~\ref{thm:band}}
Define each principal submatrix of $\bSigma$ as $\bSigma^{(p,q)} = \EE  \bZ^{(p,q)}_1 \bZ^{(p,q)\T}_1/2$, which is estimated by $\hat\bSigma_2^{(p,q), \TT}$. As a result, we expect the final estimator $\hat\bSigma_{q}$ to be close to 
\begin{equation*} 
\bSigma_{q} =  \sum_{j=-1}^{\ceil{(d-1)/q}} \extend{d}{j q +1}(\bSigma^{(j q+1,2q) }) - \sum_{j=0}^{\ceil{(d-1)/q}} \extend{d}{j q+1}(\bSigma^{(j q +1, q)}).
\end{equation*}
By the triangle inequality, we have $\|\hat\bSigma_{q}-\bSigma\|_2  \leq   \|\hat\bSigma_{q}-\bSigma_q\|_2 + \|\bSigma_{q}-\bSigma\|_2 $. We first establish an upper bound for the bias term $\|\bSigma_{q}-\bSigma\|_2$. According to the decomposition illustrated by Figure \ref{Fig_bandable}, $\bSigma_{q}$ is a  banded version of the population covariance with bandwidth between $q$ and $2q$. Therefore, we bound the spectral norm of $\bSigma_{q}-\bSigma$ by the $\|\cdot \|_{1,1}$ norm as follows:
\begin{equation*} 
\|\bSigma_{q}-\bSigma\|_2 \leq  \max_{1\leq \ell  \leq d}\sum_{k:  \vert
	k - \ell  \vert > q }    \vert \sigma _{k \ell } \vert \leq \frac{M }{q^{ \alpha}}.
\end{equation*}

It remains to control the estimation error $\|\hat\bSigma_{q}-\bSigma_q\|_2$. Define $\Db^{(p,q)}=\hat\bSigma_2^{(p,q), \TT}-\bSigma^{(p,q)}$, $$
	\Sb_1=\sum_{j=-1 : j {\rm~is~odd}}^{\ceil{(d-1)/q}} \extend{d}{j q +1} \{  \Db^{(j q+1,2q) } \} , \quad \Sb_2=\sum_{j= 0: j {\rm~is~even}}^{\ceil{(d-1)/q}} \extend{d}{j q +1}\{  \Db^{(j q+1,2q) } \} ,
$$ 
and $\Sb_3=\sum_{j=0}^{\ceil{(d-1)/q}} \extend{d}{j q+1}\{  \Db^{(j q +1, q)}\}$. Note that each $\Sb_i$ above is a sum of disjoint block diagonal matrices. Therefore,
\#
\|\hat\bSigma_{q}-\bSigma_q\|_2 	&\leq   \|\Sb_1\|_2 +\|\Sb_3\|_2 +\|\Sb_3\|_2    \notag\\
	&\leq  3  \max_{j=-1}^{\ceil{(d-1)/q}}\{\|  \Db^{(j q+1,2q)}\|_2 ,\| \Db^{(j q+1,q)}\|_2 \} \label{eq: max band0}.
\#
Applying Theorem~\ref{thm:u-type} to each principal submatrix with the choice of $\delta=( n^{c_0} d )^{-1}$ in  $\tau$, and by the union bound, we obtain that with probability at least $1-2d\delta=1-2n^{-c_0}$,
\#
& \max_{j=-1}^{\ceil{(d-1)/q}}\{\| \Db^{(j q+1,2q)}\|_2 ,\|  \Db^{(j q+1,q)}\|_2 \} \notag \\
&\leq  2\| \bSigma \|_2^{1/2} \{  (M_1+1) q\| \bSigma \|_2  +  \|\bSigma \|_2 \}^{1/2} \sqrt{ \frac{\log(4q)+\log \delta^{-1}}{m}  }  \nn \\
&\leq  2M_0 \sqrt{ 1+ (M_1+1)q  }   \sqrt{ \frac{\log(4d)+c_0\log (nd)}{n} },  \nn
\#
where we used the inequalities ${\rm tr}(\Db^{(j q+1,2q)})\leq 2q \| \bSigma \|_2$ and $\| \bSigma \|_2\leq M_0$. Plugging this into (\ref{eq: max band0}), we obtain that with probability at least $1-2n^{-c_0}$, 
\begin{equation*}
\|\hat\bSigma_{q}-\bSigma_q\|_2 \leq 6 M_0 \sqrt{ 1+ (M_1+1)q } \sqrt{ \frac{\log(4d)+c_0\log (nd)}{n} }.
\end{equation*}

In view of the upper bounds on  $\|\hat\bSigma_{q}-\bSigma_q\|_2$ and $\|\bSigma_{q}-\bSigma\|_2$,  the optimal bandwidth  $q$ is of order $\{ n/ \log (nd)\}^{1/(2\alpha + 1)} \wedge  d$, which leads to the desired result. \qed

\subsection{Proof of Theorem~\ref{thm:spre}}
Define the symmetrized Bregman divergence for the loss function $\cL(\bTheta)= \langle\bTheta^2, \hat{\bSigma}_{1}^{\TT} \rangle-\tr(\bTheta)$ as $D_\cL^{\rm s}(\bTheta_1,\bTheta_2)=\langle \nabla\cL(\bTheta_1)-\nabla\cL(\bTheta_2),\bTheta_1-\bTheta_2\rangle$. 
We first need the following two lemmas. 
\begin{lemma}\label{lemma:cone}
Provided $\lambda\geq 2\| \nabla \cL(\bTheta^*)\|_{\max}$, $\widehat\bTheta$ falls in the $\ell_1$-cone
\$
\|\widehat\bTheta_{\cS^{{\rm c}}}-\bTheta^*_{\cS^{{\rm c}}} \|_{\ell_1}\leq 3 \|\widehat\bTheta_\cS-\bTheta^*_\cS \|_{\ell_1}.
\$
\end{lemma}
\begin{proof}[Proof of Lemma \ref{lemma:cone}]
Set $\hat \bGamma=(\hat{\Gamma}_{k\ell})_{1\leq k,\ell\leq d}\in \mathbb{R}^{d \times d},$ where $\hat{\Gamma}_{k\ell}=\partial |\hat{\Theta}_{k\ell}|\in [−1,1]$ whenever $k\neq \ell$, and $\hat{\Gamma}_{k\ell}=0$ whenever $k=\ell$. Here $\partial f(x_0)$ denotes the subdifferential of $f$ at $x_0$. By the convexity of the loss function and the optimality condition, we have
\$
 0&\leq \langle \nabla\cL (\hbP)-\nabla\cL(\bTheta^*), \hbP-\bTheta^* \rangle\\
 &=\langle - \lambda\widehat\bGamma-\nabla\cL(\bTheta^*), \hbP-\bTheta^*\rangle\\
 &=-\langle \lambda\widehat\bGamma, \hbP-\bTheta^* \rangle-\langle \nabla\cL(\bTheta^*), \hbP-\bTheta^* \rangle\\
 &\leq -\lambda \|\widehat\bTheta_{\cS^{{\rm c}}}-\bTheta^*_{\cS^{{\rm c}}}\|_{\ell_1} +\lambda \|\widehat\bTheta_{\cS}-\bTheta^*_{\cS}\|_{\ell_1} + \frac{\lambda}{2} \|\widehat\bTheta-\bTheta^* \|_{\ell_1}.
\$
Rearranging terms proves the stated result.
\end{proof}

\begin{lemma}\label{lemma:b.4}
Under the restricted eigenvalue condition, it holds
\$
D_{\cL}^{{\rm s}}(\hbP,\bTheta^*) \geq \kappa_- \|\hbP-\bTheta^* \|_{\text F}^2.
\$
\end{lemma}
\begin{proof}
We use $\text{vec}(\Ab)$ to denote the vectorized form of  matrix $\Ab$. Let $\bDelta=\hbP-\bTheta^*$. Then by the mean value theorem, there exists a $\gamma\in [0,1]$ such that
\$
D_\cL^{\rm s}(\hbP, \bTheta^*) &= \langle \nabla\cL(\hbP)-\nabla\cL(\bTheta^*),\hbP-\bTheta^* \rangle\\
&=\text{vec}(\hbP-\bTheta^*)^\T  \nabla^2\cL(\hbP+\gamma\bDelta) \text{vec}(\hbP-\bTheta^*)\\
&\geq \kappa_- \|\bDelta\|_{\text F}^2, 
\$
where the last step is due to the restricted eigenvalue condition and Lemma \ref{lemma:cone}.  This completes the proof. 
\end{proof}
Applying Lemma \ref{lemma:b.4}  gives
 \#\label{tech_2.1}
  \kappa_-\| \hbP-\bTheta^* \|_{\text F}^2 \leq \langle \nabla\cL (\hbP)-\nabla\cL(\bTheta^*), \hbP-\bTheta^* \rangle.
  \#
Next, note that the sub-differential of the norm $\|\cdot\|_{\ell_1}$ evaluated at $\bPsi = (\Psi_{k\ell})_{1\leq k, \ell \leq d}$ consists the set of all symmetric matrices $\bGamma= (\Gamma_{k\ell})_{1\leq k, \ell \leq d}$ such that $\Gamma_{k \ell}=0$ if $k=\ell$, $\Gamma_{k \ell}=\text{sign}(\Psi_{k \ell})$ if $k \ne  \ell$ and $\Psi_{k \ell }\ne 0$, $\Gamma_{k \ell}\in [-1,+1]$ if $k\ne \ell$ and $\Psi_{k \ell}= 0$. Then by the Karush-Kuhn-Tucker conditions, there exists some $\widehat{\bGamma}\in \partial \|\hbP\|_{\ell_1}$ such that 
\$
\nabla\cL(\hbP)+\lambda \widehat{\bGamma}={\bf 0}.
\$  
Plugging the above equality into \eqref{tech_2.1} and  rearranging terms, we obtain
	\#\label{tech_2.3}
 &\kappa_-\|\hbP-\bTheta^* \|_\F^2+\underbrace{\langle\nabla\cL(\bTheta^* ),\hbP-\bTheta^* \rangle}_{\text{\Rom{1}}} +\underbrace{\langle\lambda\widehat{\bGamma},\hbP -\bTheta^* \rangle}_{\text{\Rom{2}}}\leq 0.
	\#
We bound terms \Rom{1} and \Rom{2} separately, starting with \Rom{1}. Our first observation is
\$
\nabla\cL(\bTheta^*)=(\bTheta^* \widehat\bSigma_1^\TT-\Ib)/2+(\widehat\bSigma_1^\TT\bTheta^* - \Ib)/2.
\$
By Theorem \ref{thm:element-wise}, we obtain that with probability at least $1-2\delta$,
\$
\|\nabla\cL (\bTheta^*)\|_{\max}
&\leq \|\bTheta^*\|_{1,1}\|\widehat\bSigma_1^\TT-\bSigma \|_{\max}\leq 2 M \|\Vb \|_{\max}  \sqrt{\frac{2\log d + \log \delta^{-1} }{\lfloor n/2\rfloor } } \leq  \lambda/2.  
\$
Let $\cS$ be the support of nonzero elements of $\bTheta^*$ and  $\cS^{{\rm c}}$ be its complement  with respect to the full index set $\{(k,\ell):1\leq k, \ell \leq d\}$. For term \Rom{1}, separating the support of $\nabla\cL(\bTheta^*)$ and $\hbP-\bTheta^*$ to $\cS$ and $\cS^{{\rm c}}$ and applying the matrix H{\"o}lder inequality, we obtain
	\$
	\langle\nabla\cL(\bTheta^*),\hbP-\bTheta^*\rangle
	&=\langle (\nabla\cL(\bTheta^* ) )_{\cS}, (\hbP-\bTheta^*)_{\cS}\rangle +\langle (\nabla\cL(\bTheta^*))_{\cS^{{\rm c}} }, (\hbP - \bTheta^* )_{\cS^{{\rm c}}}\rangle  \notag\\
	& \geq -\|   (\nabla\cL(\bTheta^*) )_{\cS}\|_{\text{F}}\|  (\hbP - \bTheta^* )_{\cS}\|_{\text F} -\| (\nabla\cL(\bTheta^*) )_{\cS^{{\rm c}}}\|_{\text{F}}\| (\hbP - \bTheta^* )_{\cS^{{\rm c}} }\|_{\text{F}}.
	\$
	For term \Rom{2}, separating the support of $\lambda\widehat{\bGamma}$ and $\hbP-\bTheta^*$  to  $\cS$ and $\cS^{{\rm c}}$, we have	
	\#\label{tech_2.5}
	\langle \lambda \widehat\bGamma, \hbP -\bTheta^* \rangle&=\langle\lambda\widehat\bGamma_{\cS },(\hbP-\bTheta^*)_{\cS}\rangle+\langle\lambda\widehat\bGamma_{\cS^{{\rm c}} },(\hbP-\bTheta^*)_{\cS^{{\rm c}}}\rangle .
	\#
Let $1_\cA\in \RR^{d\times d}$ be a $d$-by-$d$ matrix such that $1_{k \ell}=1$ if $(k,\ell)\in \cA$, $1_{k \ell}=0$ otherwise. For  the last term in the above equality, we have 
\#\label{tech_2.6}
	\langle\lambda\widehat\bGamma_{\cS^{{\rm c}}},(\hbP-\bTheta)_{\cS^{{\rm c}}}\rangle&=\langle\lambda  \cdot 1_{\cS^{{\rm c}}},|\hbP_{\cS^{{\rm c}}}|\rangle=\langle\lambda  \cdot 1_{\cS^{{\rm c}}}, |(\hbP-\bTheta)_{\cS^{{\rm c}}}|\rangle.
	\#
	Plugging \eqref{tech_2.6} into \eqref{tech_2.5} and applying the matrix H{\"o}lder inequality yields
	\$
	\langle \lambda \widehat\bGamma,\hbP-\bTheta^* \rangle 
	&= \langle\lambda\widehat\bGamma_{\cS},(\hbP-\bTheta^*)_{\cS}\rangle+\langle\lambda  \cdot 1_{S^{{\rm c}}}, |(\hbP-\bTheta^*)_{\cS^{{\rm c}}}|\rangle \notag\\
	&= \langle\lambda\widehat\bGamma_{\cS},(\hbP-\bTheta^*)_{\cS}\rangle+\|\lambda  \cdot 1_{\cS^{{\rm c}}}\|_{\text F}\|(\hbP-\bTheta^*)_{\cS^{{\rm c}}}\|_{\text F}\\
	&\geq -\|\lambda \cdot 1_{\cS}\|_{\text F}\|(\hbP - \bTheta^*)_{\cS}\|_{\text F}+\lambda \sqrt{s}\|(\hbP-\bTheta^*)_{\cS^{{\rm c}}}\|_{\text F}.
\$
	Plugging the bounds for  \Rom{1} and \Rom{2}  back into \eqref{tech_2.3}, we find 
	\$ 
	&\kappa_- \|\hbP-\bTheta^* \|^2_{\text F}+ (\|\lambda \cdot 1_{\cS^{{\rm c}}}\|_{\text F}-\|(\nabla\cL(\bTheta^*))_{\cS^{{\rm c}}}\|_{\text F})\|(\hbP-\bTheta^*)_{\cS^{{\rm c}}}\|_{\text F}\\
	& ~~~~\leq (\| (\nabla\cL(\bTheta^*) )_{\cS}\|_{\text F}+\|\lambda \cdot 1_{\cS}\|_{\text F})\|\hbP-\bTheta^*\|_{\text F}.
	\$
Since $\|(\nabla\cL(\bTheta^*))_{\cS^{{\rm c}}}\|_{\text F}  \leq |\cS^{{\rm c}}|^{1/2} \|  \nabla\cL(\bTheta^*) \|_{\max} \leq  |\cS^{{\rm c}}|^{1/2} \lambda  = \|\lambda \cdot 1_{\cS^{{\rm c}}}\|_{\text F}$, it follows that
\#
\kappa_- \|\hbP-\bTheta^* \|^2_{\text F} \leq (\|(\nabla\cL(\bTheta^*))_{\cS}\|_{\text F}+\|\lambda \cdot 1_{\cS}\|_{\text F})\|\hbP-\bTheta^*\|_{\text F}. \nn
\#
Canceling $\|\widehat \bTheta-\bTheta^*\|_\F$ on both sides yields
        \$
	\kappa_-\|\hbP-\bTheta^*\|_{\text F}
	&\leq \|\lambda \cdot 1_S\|_{\text F}+\|\nabla\cL(\bTheta^*)_{\cS}\|_{\text F}
	\leq 3\lambda\sqrt{s}/2 
	\$
under the scaling $\lambda\geq 2\|\nabla\cL(\bTheta^*)\|_{\max}$. Plugging $\lambda$ completes the proof.   \qed

\section{Robust estimation and inference under factor models}
\label{sec:factor.model}

As a complement to the three examples considered in the main text, in this section we discuss  robust covariance estimation (Section \ref{sec:factor}) and inference (Section \ref{robust.FDP}) under factor models, which might be of independent interest. In Section \ref{robust.FDP}, we provide a self-contained analysis to prove the consistency of estimating the false discovery proportion, while there is no such a theoretical guarantee in \cite{FKSZ2019} without using sample splitting.

\subsection{Covariance estimation through factor models}
\label{sec:factor}

Consider the approximate factor model of the form $\bX = (X_1,\ldots, X_d)^\T = \bmu + \Bb \bbf + \bvarepsilon$, from which we observe
\#
	  \bX_i = (X_{i1}, \ldots, X_{id})^\T = \bmu + \Bb \bbf_i + \bvarepsilon_i , \ \ i=1,\ldots, n, \label{afm}
\#
where $\bmu$ is a $d$-dimensional unknown mean vector, $\Bb = (\bb_1, \ldots, \bb_d)^\T \in \RR^{d\times r}$ is the factor loading matrix,  $\bbf_i \in \RR^r$ is a vector of common factors to the $i$th observation and is independent of the idiosyncratic noise $\bvarepsilon_i$. 
For more details about factor analysis, we refer the readers to \cite{AR1956}, \cite{CR1983}, \cite{BL2012} and \cite{FH2017}, among others. 
Factor pricing model has been widely used in  financial economics, where $X_{ik}$ is the excess return of fund/asset $k$ at time $i$, $\bbf_i$'s are the systematic risk factors related to some specific linear pricing model, such as the capital asset pricing model (CAPM) \citep{S1964}, and the Fama-French three-factor model \citep{FF1993}.

Under model \eqref{afm}, the covariance matrix of $\bX $ can be written as
\#
	\bSigma = (\sigma_{k\ell})_{1\leq k,\ell \leq d} = \Bb \cov(\bbf) \Bb^\T + \bSigma_{\varepsilon}, \label{cov.X}
\#
where $\bSigma_\varepsilon = (\sigma_{\varepsilon , k\ell})_{1\leq k,\ell \leq d}$ denotes the covariance matrix of $\bvarepsilon = (\varepsilon_1, \ldots \varepsilon_d)^\T$, which is typically assumed to be sparse. When $\bSigma_\varepsilon = \Ib_d$, model \eqref{afm} is known as the strict factor model.
To make the model identifiable, following \cite{BL2012} we assume that $\cov(\bbf) = \Ib_r$ and that the columns of $\Bb$ are orthogonal.

We consider the robust estimation of $\bSigma$ based on independent observations $\bX_1,\ldots, \bX_n$ from model \eqref{afm}. By \eqref{cov.X} and the identifiability condition, $\bSigma$ is comprised of two components: the low-rank component $\Bb \Bb^\T$ and the sparse component $\bSigma_{\varepsilon}$. Using a pilot robust covariance estimator $\hat \bSigma_{1}^{\TT}$ given in \eqref{element-wise.censored.est} or $\hat\bSigma^{\HH}_1$ given in \eqref{Huber1}, we propose the following robust version of the principal orthogonal complement thresholding (POET) procedure \citep{FLM2013}:

\begin{enumerate}
\item[(i)] Let $\hat \lambda_1 \geq \hat \lambda_2 \geq \cdots \geq \hat \lambda_r$ be the top $r$ eigenvalues of $\hat  \bSigma_1^{\HH}$ (or $\hat{\bSigma}_1^{\TT}$) with corresponding eigenvectors $\hat \bv_1 , \hat  \bv_2 , \ldots, \hat \bv_r$. Compute the principal orthogonal complement 
\#
	\hat \bSigma_\varepsilon = (\hat \sigma_{\varepsilon, k\ell })_{1\leq k,\ell \leq d} = \hat \bSigma_1^{\HH}  -  \hat \Vb \hat \bLambda \hat \Vb^\T,  \label{poc}
\#
where $\hat \Vb = (\hat \bv_1 ,  \ldots, \hat \bv_r)$ and $\hat{\bLambda} =  {\rm diag}\,( \hat \lambda_1, \ldots, \hat \lambda_r)$.

\item[(ii)] To achieve sparsity, apply the adaptive thresholding method \citep{RLZ2009, CL2011} to $\hat \bSigma_\varepsilon$ and obtain $\hat  \bSigma_\varepsilon^{\cT} = (\hat{\sigma}^{\cT}_{\varepsilon, k\ell})_{1\leq k,\ell\leq d}$ such that  
\#  \label{soft.thresholding}
	 \hat{\sigma}^{\cT}_{\varepsilon, k\ell} = \begin{cases}
	 \hat{\sigma}_{\varepsilon, k\ell}  & \mbox{ if } k=\ell,  \\
	 s_{k \ell}( \hat{\sigma}_{\varepsilon, k\ell} )  & \mbox{ if } k \neq \ell,
	 \end{cases} 
\#
where $s_{k\ell} (z) = \sgn(z) (|z| - \lambda_{k\ell})$, $z\in \RR$ is the soft thresholding function with $\lambda_{k\ell} = \lambda ( \hat \sigma_{\varepsilon, kk}  \,\hat \sigma_{\varepsilon, \ell \ell })^{1/2}$ and $\lambda>0$ being a regularization parameter.

\item[(iii)] Obtain the final estimator of $\bSigma$  as $\hat{\bSigma} = \hat \Vb \hat \bLambda \hat \Vb^\T + \hat{\bSigma}_\varepsilon^{\cT}$.
\end{enumerate}

\begin{remark}
 The POET method \citep{FLM2013} employs the sample covariance matrix as an initial estimator and has desirable properties for sub-Gaussian data.  
 For elliptical distributions, 
  \cite{FLW2017} proposed to use the marginal Kendall's tau to estimate $\bSigma$, and to use its top $r$ eigenvalues and the spatial Kendall's tau to estimate the corresponding leading eigenvectors. In the above robust POET procedure, we only need to compute one initial estimator of $\bSigma$ and moreover, optimal convergence rates can be achieved in high dimensions under finite fourth moment conditions; see Theorem \ref{thm:poet}. 
\end{remark}


\begin{cond} \label{cond:fm}
Under model \eqref{afm}, the latent factor $\bbf \in \RR^r$ and the idiosyncratic noise $\bvarepsilon\in \RR^d$ are independent. Moreover,
\begin{itemize}
\item[(i)](Identifiability) $\cov(\bbf) = \Ib_r$ and the columns of $\Bb$ are orthogonal;

\item[(ii)](Pervasiveness) there exist positive constants $c_l, c_u$ and $C_1$ such that 
$$
	c_l \leq  \min_{1\leq \ell \leq r}\{ \lambda_\ell (\Bb^\T \Bb /d ) -  \lambda_{\ell +1}(\Bb^\T \Bb /d ) \} \leq c_u ~\mbox{ with }~  \lambda_{r+1}(\Bb^\T \Bb /d ) =0,
$$
and $\max\{ \| \Bb \|_{\max} , \| \bSigma_\varepsilon \|_2  \}\leq C_1$;

\item[(iii)](Moment condition) $\max_{1\leq \ell \leq d} {\rm kurt}(X_\ell)   \leq C_2$ for some constant $C_2 >0$;

\item[(iv)](Sparsity) $\bSigma_\varepsilon$ is sparse in the sense that $s:= \max_{1\leq k\leq d} \sum_{\ell =1}^d I(\sigma_{\varepsilon, k\ell} \neq 0) $ satisfies
$$
	s^2 \log d = o(n) ~\mbox{ and }~ s^2 = o(d) ~\mbox{ as }~ n, d \to \infty.
$$
\end{itemize}
\end{cond}

\begin{theorem} \label{thm:poet}
Under Condition~\ref{cond:fm}, the robust POET estimator with 
$$
	\tau_{k\ell } \asymp \sqrt{n/(\log d )} , \, 1\leq k, \ell \leq d , ~\mbox{ and }~ \lambda \asymp w_{n,d} := \sqrt{ \log(d) /n} + d^{-1/2} 
$$
satisfies
\#
	 \| \hat{\bSigma}_\varepsilon^{\cT}  - \bSigma_\varepsilon \|_{\max}  =  O_{\PP}(  w_{n,d} ) ,  \quad \| \hat{\bSigma}_\varepsilon^{\cT}  - \bSigma_\varepsilon \|_2 =  O_{\PP}(s w_{n,d} ) ,  \label{sig.noise.bound}  \\
 	\| \hat{\bSigma} - \bSigma \|_{\max} = O_{\PP}(w_{n,d} )   ~\mbox{ and }~  \| \hat{\bSigma} - \bSigma \|_2 = O_{\PP}(d  w_{n,d}  )   \label{sig.total.bound}
\#
as $n, d \to \infty$.
\end{theorem}

\subsection{Factor-adjusted multiple testing}
\label{robust.FDP}

Here we consider the problem of simultaneously testing the hypotheses
\#
	H_{0k} : \mu_k = 0 ~\mbox{ versus }~ H_{1k}: \mu_k \neq 0, ~\mbox{ for } k=1,\ldots, d ,  \label{mt} 
\#
under model \eqref{afm}. Although the key implication from the multi-factor pricing theory is that the intercept $\mu_k$ should be zero, known as the ``mean-variance efficiency'' pricing, for any asset $k$, an important question is whether such a pricing theory can be validated by empirical data. In fact, a very small proportion of $\mu_k$'s might be nonzero according to the Berk and Green equilibrium \citep{BG2004}. Various statistical methods have been proposed to identify those positive $\mu_k$'s \citep{BSW2010, FH2017, LD2017}. These works assume that both the factor and idiosyncratic noise follow multivariate normal distributions. To accommodate the heavy-tailed character of empirical data, we develop a robust multiple testing procedure that controls the overall false discovery rate or false discovery proportion.

For each $1\leq k\leq d$, let $T_k$ be a generic test statistic for testing the individual hypothesis $H_{0k} : \mu_k = 0$. For any threshold level $z>0$, we reject the $j$th hypothesis whenever $|T_j| \geq z$. The numbers of total discoveries $R(z)$ and false discoveries $V(z)$ are defined by
\#
	R(z) = \sum_{k=1}^d  I( | T_k | \geq z )  \ \ \mbox{ and } \ \ V(z) = \sum_{ k \in \mathcal{H}_0} I( |T_k | \geq z ), 
\#
respectively, where $\mathcal{H}_0 = \{  1\leq k\leq d: \mu_k = 0\}$. The main object of interest is the false discovery proportion (FDP), given by
\$
\FDP(z) = V(z)/ R(z). 
\$ 
Throughout we use the convention $0/0 = 0$. Note that $R(z)$ is observable given all the test statistics, while $V(z)$ is an unobservable random variable that needs to be estimated. For testing individual hypotheses $H_{0 k}$, the standardized means $Z_k$, where $Z_k = n^{-1/2} \sn X_{ik}$, are sensitive to the tails of the sampling distributions. In particular, when the number of features $d$ is large, stochastic outliers from the test statistics $Z_k$ can be so large that they are mistakenly regarded as discoveries. Motivated by recent advances on robust estimation and inference \citep{C2012, ZBFL2017}, we consider the following robust $M$-estimator of $\mu_k$:
\#
	\hat{\mu}_k  = \argmin_{\theta \in \RR } \sn \ell_{\tau_k} (X_{i k } - \theta ) ~\mbox{ for some }~ \tau_k >0 .  \label{rest}
\#
The corresponding test statistic is then given by $T_k= \sqrt{n} \,\hat{\mu}_k$ for $k =1,\ldots, d$.

Based on the law of large numbers, we define the approximate FDP by
\# \label{AFDP.def}
	\FDP_{{\rm A}}( z ) = \frac{1}{R(z)}  \sum_{ k=1}^d  \bigg\{ \Phi\bigg(  \frac{ - z + \sqrt{n}\, \bb_k^\T \overline{\bbf} }{\sqrt{\sigma_{kk} - \| \bb_k \|_2^2}}  \bigg) + \Phi\bigg(  \frac{- z - \sqrt{n}\, \bb_k^\T \overline{\bbf}}{\sqrt{\sigma_{kk} - \| \bb_k\|_2^2 } } \bigg) \bigg\} , 
\#
where $\overline{\bbf} = (1/n)\sn  \bbf_i $. It is shown in the appendix that  the approximate FDP in \eqref{AFDP.def} serves as a conservative surrogate for the true FDP.

Note that the approximate FDP defined in \eqref{AFDP.def} depends on a number of unknown parameters, say $\{ \bb_k , \sigma_{kk} \}_{k=1}^d$ and $\overline{\bbf}$. In this section, we describe robust procedures to estimate these quantities using the only observations $\bX_1,\ldots, \bX_n$.
\begin{enumerate}
\item[(a)] Compute the Huber-type covariance estimator $\hat{\bSigma}_1^{\HH} = (\hat{\sigma}^{\HH}_{1, k\ell})_{1\leq k,\ell \leq d}$  (or the truncated estimator $\hat{\bSigma}^{\TT}_1$), and let $\hat{\lambda}_1 \geq \cdots \geq \hat{\lambda}_r$ and $\hat{\bv}_1,\ldots, \hat{\bv}_r$ be its top $r$ eigenvalues and the corresponding eigenvectors, respectively.

\item[(b)] Compute $\hat{\Bb} = ( \hat{\lambda}^{1/2}_1 \hat{\bv}_1, \ldots,  \hat{\lambda}^{1/2}_r \hat{\bv}_r ) \in \RR^{d\times r}$ and $\hat{\bu} =  \sqrt{n} \, (\hat{\Bb}^\T \hat{\Bb})^{-1} \hat{\Bb}^\T \overline{\bX} \in \RR^r$, which serve as estimators of $\Bb$ and $\sqrt{n} \overline{\bbf}$, respectively. Here $\overline{\bX} = (1/n) \sn \bX_i$.

\item[(c)] Denote by $\hat{\bb}_1, \ldots, \hat{\bb}_d$ the $d$ rows of $\hat{\Bb}$. For any $z\geq 0$, we estimate the approximate FDP $\FDP_{{\rm A}}(z)$ by
\#
	\hat{\FDP}_{{\rm A}}(z) = \frac{1}{R(z)}  \sum_{ k=1}^d  \bigg\{ \Phi\bigg(  \frac{ - z +  \hat \bb_k^\T \hat{\bu} }{\sqrt{ \hat  \sigma_{\varepsilon, kk} }}  \bigg) + \Phi\bigg(  \frac{- z -  \hat \bb_k^\T  \hat{\bu} }{\sqrt{ \hat{\sigma}_{\varepsilon , kk} } } \bigg) \bigg\} ,
\#
where $\hat{\sigma}_{\varepsilon , kk } = \hat \sigma_{1,kk}^{\HH} - \| \hat \bb_k\|_2^2$ for $k=1,\ldots, d$.
\end{enumerate}

The construction of $\hat{\Bb}$ is based on the observation that  principal component analysis and factor analysis are approximately equivalent under the pervasive assumption in high dimensions \citep{FLM2013}. To estimate $\overline{\bbf}$, note from model \eqref{afm} that $\overline{\bX} = \bmu +  \Bb \overline{\bbf} + \overline \bvarepsilon$, where $\bmu$ is assumed to be sparse and therefore is ignored for simplicity.

\begin{theorem} \label{thm:fdp}
Under model \eqref{afm}, assume that $\bbf$ and $\bvarepsilon$ are independent zero-mean random vectors and let $s_1 = \| \bmu \|_0$. Assume (i)--(iii) of Condition~\ref{cond:fm} hold, and that $(n, d, s_1)$ satisfies $\log d = o(n)$ and $n s_1 = o(d)$ as $n, d\to \infty$. 
Then for any $z\geq 0$,
\# \label{FDP.consistency}
	\hat{\FDP}_{{\rm A}}(z) /\FDP_{{\rm A}}(z)   \xrightarrow \PP 1 ~\mbox{ as }~ n,  d \to \infty.
\#
\end{theorem}

\subsection{Proof of Theorem~\ref{thm:poet}}

The proof is based on Theorem~2.1 and (A.1) in \cite{FLW2017}, which provides high level results for the generic POET procedure. To that end, it suffices to show that with properly chosen $ \tau_{k \ell}$,
\#
	\| \hat \bSigma_1^{\HH} - \bSigma \|_{\max} = O_{\PP} \{ n^{-1/2} (\log d)^{1/2} \}, \quad \max_{1\leq \ell \leq r} |  \hat{\lambda}_\ell / \lambda_\ell - 1 | = O_{\PP}  \{ n^{-1/2} (\log d)^{1/2} \}   \label{cond.12} \\
   \mbox{and}~~\max_{1\leq \ell \leq r} \| \hat \bv_\ell -  \bv_\ell \|_{\infty} = O_{\PP} \{  (nd)^{-1/2}  (\log d)^{1/2}   \} ,  \label{cond.3}
\#
where $\lambda_1 \geq \cdots \geq \lambda_r$ are the top $r$ eigenvalues of $\bSigma$ and $\bv_1,\ldots, \bv_r$ are the corresponding eigenvectors.

First, applying Theorem~\ref{thm:huber-type} with $\tau_{k\ell}  \asymp  \sqrt{n/(\log d )}$ implies that $	 \| \hat  \bSigma_1^{\HH }  - \bSigma \|_{\max} \lesssim n^{-1/2} (\log d)^{1/2} $ with probability at least $1-d^{-1}$. This verifies the first criterion in \eqref{cond.12}. Next, by Weyl's inequality and the inequality $\| \Ab \|_2 \leq  \| \Ab \|_{1,1}$ for symmetric matrices, we have 
\#
	\max_{1\leq \ell \leq r} |\hat{\lambda}_\ell - \lambda_\ell  | \leq \| \hat  \bSigma_1^{\HH}  - \bSigma \|_2  \leq \|  \hat \bSigma_1^{\HH} - \bSigma \|_{1,1} \leq d \|  \hat \bSigma_1^{\HH}  - \bSigma \|_{\max} . \nn
\#
Let $\overline{\lambda}_1 \geq \cdots \geq \overline{\lambda}_r$ be the top $r$ eigenvalues of $\Bb \Bb^\T$, and therefore of $\Bb^\T \Bb$. Note that, by Weyl's inequality, $\max_{1\leq \ell \leq r } | \lambda_\ell - \overline{\lambda}_\ell | \leq \| \bSigma_\varepsilon \|_2$. It thus follows from Condition~\ref{cond:fm} that
$$
	\min_{1\leq \ell  \leq r-1} | \lambda_\ell - \lambda_{\ell+1} |  \asymp d ~\mbox{ and }~ \lambda_{r} \asymp d ~\mbox{ as } d \to \infty.
$$
Together, the last two displays imply $\max_{1\leq \ell \leq r} |  \hat{\lambda}_\ell / \lambda_\ell - 1 |  \lesssim n^{-1/2} (\log d)^{1/2}$ with probability at least $1-d^{-1}$. Therefore, the second criterion in \eqref{cond.12} is fulfilled.

For \eqref{cond.3}, applying Theorem~3 and Proposition~3 in \cite{FWZ2016} we arrive at 
\#
& \max_{1\leq \ell \leq r} \| \hat \bv_\ell  - \bv_\ell \|_{\infty}  \nn \\ 
& \lesssim d^{-3/2} ( r^4 \| \hat \bSigma_1^{\HH} - \bSigma \|_{\infty, \infty} + r^{3/2} \| \hat  \bSigma_1^{\HH} - \bSigma \|_2 )  \lesssim  r^4  d^{-1/2}   \| \hat \bSigma_1^{\HH} - \bSigma \|_{\max} .\nn
\#
This validates \eqref{cond.3}. 

In summary, \eqref{sig.noise.bound} and the first bound in \eqref{sig.total.bound} follow from Theorem~2.1 and (A.1) in \cite{FLW2017}, and the second bound in \eqref{sig.total.bound} follows directly from the fact that $\| \hat{\bSigma}_1^{\HH} - \bSigma \|_2 \leq \| \hat{\bSigma}_1^{\HH} - \bSigma \|_{1,1} \leq d \| \hat{\bSigma}_1^{\HH} - \bSigma \|_{\max}$. \qed

\subsection{Asymptotic property of {\rm FDP}}
\label{sec:B}

In this section, we show that the approximate FDP in \eqref{AFDP.def} serves as a conservative surrogate for the true FDP.

\begin{cond} \label{cond:fdp}
Under model \eqref{afm}, $\bbf$ are $\bvarepsilon$ are independent zero-mean random vectors. (i) $\cov(\bbf) = \Ib_r$ and $\| \bbf \|_{\psi_2} = \sup_{\bu \in \mathbb{S}^{r-1}} \| \bu^\T \bbf \|_{\psi_2} \leq C_f$ for some constant $C_f >0$; (ii) the correlation matrix $\Rb_\varepsilon = ( \varrho_{\varepsilon, k\ell})_{1\leq k,\ell \leq d}$ of $\bvarepsilon$ satisfies $d^{-2}\sum_{1\leq k \neq \ell \leq d} \varrho_{\varepsilon, k\ell} \leq C_0 d^{- \delta_0}$ for some constants $C_0, \delta_0>0$; (iii) $d=d(n) \to \infty$ and $\log d = o(n^{1/2})$ as $n\to \infty$, and $\liminf_{n\to \infty} \frac{d_0}{d} >0$, where $d_0=\sum_{k=1}^d I(\mu_k=0)$; (iv) $C_l \leq \sigma_{\varepsilon, kk} \leq v_k^{1/2} \vee \sigma_{kk} \leq C_u$ for all $1\leq k\leq d$, where $v_k = \EE(\varepsilon_k^4)$ and $C_u , C_l $ are positive constants.
\end{cond}

\begin{theorem} \label{thm:lln}
Assume that Condition~\ref{cond:fdp} holds. In \eqref{rest}, let $\tau_k=a_k\sqrt{n/\log(nd)}$ with $a_k \geq \sigma_{kk}^{1/2}$ for $k=1,\ldots,d$. Then, as $n,d\to \infty$,
\#
	\frac{V(z)}{d_0} 	&= \frac{1}{d_0} \sum_{ k  \in \mathcal{H}_0} \bigg\{ \Phi\bigg(  \frac{ - z + \sqrt{n}\bb_k^\T \overline{\bbf} }{\sqrt{\sigma_{\varepsilon , kk }}}  \bigg) + \Phi\bigg(  \frac{ z - \sqrt{n}\, \bb_k^\T \overline{\bbf}}{\sqrt{\sigma_{\varepsilon, kk} } } \bigg) \bigg\} \nn \\
	& \quad~   +   O_{\PP}\Bigg[  \frac{1}{d^{(\delta_0 \wedge 1 )/2}} +\frac{\log(nd)}{\sqrt{n}} +  \bigg\{ \frac{\log(nd)}{n}  \bigg\}^{1/4} \Bigg]  \label{FDP.lln} 
\#
and
\# 
 \frac{R(z)}{d} & = \frac{1}{d} \sum_{ k =1}^d \bigg\{ \Phi\bigg(  \frac{ - z +  \sqrt{n}(\mu_k + \bb_k^\T \overline{\bbf} )}{\sqrt{\sigma_{\varepsilon , kk }}}  \bigg) + \Phi\bigg(  \frac{- z - \sqrt{n} (\mu_k + \bb_k^\T \overline{\bbf} ) }{\sqrt{\sigma_{\varepsilon, kk} } } \bigg) \bigg\} \nn \\
& \quad~   + O_{\PP}\Bigg[  \frac{1}{d^{(\delta_0 \wedge 1 )/2}} +\frac{\log(nd)}{\sqrt{n}} +  \bigg\{ \frac{\log(nd)}{n}  \bigg\}^{1/4} \Bigg] \label{Rz.lln} 
\#
uniformly over $z\geq 0$. In addition, for any $z\geq 0$ it holds
\#
 \FDP(z)	 = \FDP_{{\rm orc}}( z ) + o_{\PP}(1) ~\mbox{ as } n ,d \to \infty ,
\#
where
 \#
\FDP_{{\rm orc}}( z ) := \frac{1}{R(z) }  \sum_{ k \in \mathcal{H}_0 }  \bigg\{ \Phi\bigg(  \frac{ - z + \sqrt{n}\, \bb_k^\T \overline{\bbf} }{\sqrt{\sigma_{kk} - \| \bb_k \|_2^2}}  \bigg) + \Phi\bigg(  \frac{- z - \sqrt{n}\, \bb_k^\T \overline{\bbf}}{\sqrt{\sigma_{kk} - \| \bb_k\|_2^2 } } \bigg) \bigg\} \nn .
\#
\end{theorem}

\subsubsection{Preliminaries}

To prove Theorem~\ref{thm:lln}, we need the following results on the robust estimators $\mu_k$'s given in \eqref{rest}. Define $u_k = X_k - \mu_k = \bb_k^\T \bbf + \varepsilon_k$ for $k=1,\ldots , d$. Assume that $\EE(\bbf) = \textbf{0}$, $\EE(\varepsilon_k)=0$ and $\bbf$ are $\varepsilon_k$ are independent. Then we have $\EE(u_k) = 0$ and $\EE(u_k^2) =\sigma_{kk}=\| \bb_k \|_2^2 + \sigma_{\varepsilon,kk}$.

The first lemma is Theorem 5 in \cite{FLW2014} regarding the concentration of the robust mean estimator.

\begin{lemma} \label{lem1}
For every $1\leq k\leq d$ and $t>0$, the estimator $\hat{\mu}_k$ in \eqref{rest} with $\tau_k = a_k (n/t)^{1/2}$ for $a_k \geq \sigma_{kk}^{1/2}$ satisfies $  | \hat{\mu}_k  - \mu_k  |  \leq   4 a_k (t/n)^{1/2}$ with probability at least $1-2e^{-t}$ provided $n\geq 8t$.
\end{lemma}

The next result provides a nonasymptotic Bahadur representation for $\hat{\mu}_k$, which follows directly from Lemma~\ref{lem1} and Theorem~2.1 in \cite{ZBFL2017}. Let $u_{ik} = \bb_k^\T \bbf_i + \varepsilon_{ik}$ for $i=1,\ldots, n$ and $k=1,\ldots ,d$.

\begin{lemma} \label{lem2}
Under the conditions of Lemma~\ref{lem1}, it holds for every $1\leq k\leq d$ that
\#
	\bigg|  \sqrt{n} \,( \hat{\mu}_k - \mu_k )  - \frac{1}{\sqrt{n}} \sn \psi_{\tau_k}(u_{ik}) \bigg| \leq C \frac{a_k  t}{\sqrt{n}} 
\#
with probability at least $1- 3e^{-t}$ as long as $n\geq 8t$, where $C>0$ is an absolute constant and $\psi_\tau(\cdot)$ is given in \eqref{psi.def}.
\end{lemma}

Let $\tau_k$ be as in Lemma~\ref{lem1} and write
\#
	v_k = \EE( \varepsilon_k^4)  ,  \quad \xi_k = \psi_k(u_k) ~\mbox{ for }~  k=1,\ldots, d. \label{vk.xik}
\#
Here $\xi_k$ are truncated versions of $u_k$. The next result shows that the differences between the first two (conditional) moments of $u_k$ and $\xi_k$ given $\bbf$ decay as $\tau_k$ grows.

\begin{lemma} \label{lem3}
Assume that $v_k <\infty$ for $k=1,\ldots, d$.
\begin{enumerate}
\item  On the event $\cG_k:= \{ |\bb_k^\T \bbf | \leq  \tau_k /2 \}$, the following inequalities hold almost surely:
\#
	| \e_{\bbf} ( \xi_k)   -  \bb_k^\T \bbf | \leq \min(   2 \tau_k^{-1}  \sigma_{\varepsilon, kk}  , 8 \tau_k^{-3}  v_k  ) \label{approxi.mean} \\
 \mbox{ and }~ \sigma_{\varepsilon, kk } -  4 \tau_k^{-2} ( v_k  + \sigma_{\varepsilon, kk}^2)	 \leq  \var_{\bbf}(\xi_k) \leq \sigma_{\varepsilon, kk} , \label{approxi.var}
\#
where $\EE_{\bbf}(\cdot )$ and $\var_{\bbf}(\cdot)$ denote the conditional mean and variance, separately.

\item On the event $\cG_k \cap \cG_\ell$, the following holds almost surely:
\#
 |  \cov_{\bbf}(\xi_k, \xi_\ell) - \sigma_{\varepsilon, k\ell} |  \leq C   \frac{v_k \vee v_\ell }{( \tau_k \wedge \tau_\ell )^2} , \label{approxi.cov}
\#
where $C>0$ is an absolute constant.
\end{enumerate}
\end{lemma}

\begin{proof}[Proof of Lemma~\ref{lem3}]
First we prove \eqref{approxi.mean} and \eqref{approxi.var}.
Fix $k$ and let $\tau=\tau_k$ for simplicity.
Since $\varepsilon_k$ and $\bbf$ are independent, we have
\begin{align}
	&	\e_{\bbf}   \xi_k   - \bb_k^\T \bbf \nn \\
	& = - \e_{\bbf}  (\varepsilon_k + \bb_k^\T \bbf - \tau )  I( \varepsilon_k > \tau - \bb_k^\T \bbf ) + \e_{\bbf} (- \varepsilon_k - \bb_k^\T \bbf - \tau) I( \varepsilon_k < -\tau - \bb_k^\T \bbf  ). \nn
\end{align}
Therefore, on the event $\cG_k$,  it holds for any $2\leq q\leq 4$ that
\begin{align}
	  | \e_{\bbf}    \xi_k     - \bb_k^\T \bbf | \leq  \e_{\bbf} \{ | \varepsilon_k | I (| \varepsilon_k | > \tau - |\bb_k^\T \bbf| )   \} \leq   (\tau - |\bb_k^\T \bbf|  )^{1-q} \,\e ( |\varepsilon_k |^q )   \nn
\end{align}
almost surely. This proves \eqref{approxi.mean} by taking $q$ to be 2 and 4. For the conditional variance, by \eqref{approxi.mean} and the decomposition $\e_{\bbf}  ( \xi_k - \bb_k^\T \bbf )^2 = \var_{\bbf} (\xi_k) + ( \e_{\bbf} \xi_k - \bb_k^\T \bbf )^2$, we have
\#
	  \e_{\bbf}  ( \xi_k - \bb_k^\T \bbf )^2 - \frac{\sigma_{\varepsilon, kk}^2 }{(\tau - |\bb_k^\T \bbf | )^2}  \leq  \var_{\bbf} (\xi_k) \leq \e_{\bbf}  ( \xi_k - \bb_k^\T \bbf )^2. \label{var.ineq}
\#
Note that $\xi_k - \bb_k^\T \bbf$ can be written as
\#
   \varepsilon_k I( | \bb_k^\T \bbf + \varepsilon_k | \leq  \tau ) + (\tau - \bb_k^\T \bbf) I(\bb_k^\T \bbf +  \varepsilon_k > \tau  ) - ( \tau + \bb^\T_k \bbf) I(\bb_k^\T \bbf +  \varepsilon_k < -\tau ) . \nn
\#
It follows that
\begin{align}
	&  ( \xi_k  - \bb_k^\T \bbf )^2 \nn \\
	& =  \varepsilon_k^2 I(| \bb_k^\T \bbf +  \varepsilon_k | \leq \tau ) + (\tau - \bb_k^\T \bbf)^2 I(\bb_k^\T \bbf +  \varepsilon_k > \tau ) + ( \tau + \bb^\T_k \bbf)^2I(\bb_k^\T \bbf +  \varepsilon_k<-\tau ) . \nn
\end{align}
Taking conditional expectations on both sides gives
\begin{align}
	   \e_{\bbf} (\xi_k  - \bb_k^\T \bbf )^2 & = \e(  \varepsilon_k^2 ) - \e_{\bbf} \{ \varepsilon_k^2 I(|\bb_k^\T \bbf +  \varepsilon_k| >\tau  )  \}  \nn \\
	 & \quad   + (\tau - \bb_k^\T \bbf)^2 \PP_{\bbf} (   \varepsilon_k > \tau - \bb_k^\T \bbf  ) + (\tau + \bb_k^\T \bbf)^2 \PP_{\bbf} (  \varepsilon_k < - \tau - \bb_k^\T \bbf ).  \nn
\end{align}
Using the identity that $u^2 = 2 \int_0^u t\, dt$ for any $u>0$, we have
\begin{align}
	 &  \e_{\bbf}  \{ \varepsilon_k^2 I( \bb_k^\T \bbf + \varepsilon_k >\tau  ) \} \nn \\
	&  = 2 \e_{\bbf} \int_0^\infty I(  \varepsilon_k > t) I(  \varepsilon_k > \tau - \bb_k^\T \bbf  ) t \, dt \nn \\
	& = 2 \e_{\bbf} \int_0^{\tau - \bb_k^\T \bbf}  I(  \varepsilon_k > \tau - \bb_k^\T \bbf  )t \, dt + 2 \e_{\bbf} \int_{\tau - \bb_k^\T \bbf}^\infty I(  \varepsilon_k > t)  t \, dt \nn \\
	& =  (\tau - \bb_k^\T \bbf)^2  \PP_{\bbf} (  \varepsilon_k > \tau - \bb_k^\T \bbf )    + 2 \int_{\tau - \bb_k^\T \bbf}^\infty \PP(  \varepsilon_k > t)  t \, dt  . \nn
\end{align}
It can be similarly shown that
\begin{align}
	 \e_{\bbf} \{ \varepsilon_k^2 I( \bb_k^\T \bbf + \varepsilon_k <-\tau )  \}  = (\tau + \bb_k^\T \bbf)^2 \PP_{\bbf} ( \varepsilon_k < - \tau - \bb_k^\T \bbf  ) + 2 \int_{\tau + \bb_k^\T \bbf}^\infty \PP(- \varepsilon_k >t )  t \, dt  . \nn
\end{align}
Together, the last three displays imply
\begin{align}
	  0  & \geq  \e_{\bbf} ( \xi_k - \bb_k^\T \bbf )^2 - \e ( \varepsilon_k^2) \nn \\
	& \geq  - 2 \int_{\tau - |\bb_k^\T \bbf|}^{\infty} \PP ( |\varepsilon_k | > t ) t \, dt   \geq  - 2  v_k  \int_{\tau - |\bb_k^\T  \bbf|}^{\infty} \frac{dt}{t^3}   = -  \frac{ v_k }{( \tau - |\bb_k^\T \bbf | )^2 }   . \nn
\end{align}
Combining this with \eqref{var.ineq} and \eqref{approxi.mean} proves \eqref{approxi.var}.

Next, we study the covariance $\cov_{\bbf} ( \xi_k, \xi_\ell ) $ for $k \neq \ell$, which can be written as
\begin{align}
 \cov_{\bbf} ( \xi_k, \xi_\ell )   & = \e_{\bbf}  (  \xi_k  -  \bb_k^\T \bbf  + \bb_k^\T \bbf - \e_{\bbf}  \xi_k  )  (\xi_\ell -  \bb_\ell^\T \bbf  + \bb_\ell^\T \bbf - \e_{\bbf} \xi_\ell )  \nn \\
	& =   \underbrace{ \e_{\bbf} (  \xi_k  -  \bb_k^\T \bbf )( \xi_\ell -  \bb_\ell^\T \bbf ) }_{\Pi_1}    -   \underbrace{ (  \e_{\bbf} \xi_k -  \bb_k^\T  \bbf ) ( \e_{\bbf} \xi_\ell  -  \bb_\ell^\T \bbf  ) }_{\Pi_2} . \nn
\end{align}
For $\Pi_2$, it follows immediately from \eqref{approxi.mean} that $|\Pi_2| \lesssim  (\tau_k  \tau_\ell )^{-1} \sigma_{\varepsilon, kk} \,\sigma_{\varepsilon, \ell \ell}$ almost surely on the event $\cG_{k\ell}:=\{ | \bb_k^\T \bbf | \leq \tau_k/2  \} \cap \{  | \bb_\ell^\T \bbf | \leq \tau_\ell /2 \}$. It remains to consider $\Pi_1$. Recall that $\xi_k - \bb_k^\T \bbf=  \varepsilon_k I(| u_k | \leq \tau_k) + (\tau_k - \bb_k^\T \bbf)I( u_k >\tau_k ) - (\tau_k + \bb_k^\T \bbf)I( u_k < -\tau_k)$, where $u_k =  \bb_k^\T \bbf + \varepsilon_k$. Then, $\Pi_1$ can be written as
\begin{align}
  & \e_{\bbf}  \varepsilon_k \varepsilon_\ell I( | u_k | \leq  \tau_k , |u_\ell  |\leq  \tau_\ell ) 
	 + (\tau_\ell - \bb_\ell^\T \bbf)  \e_{\bbf} \varepsilon_k I( | u_k | \leq  \tau , u_\ell  >\tau ) \nn \\
	& \quad - (\tau_\ell + \bb_\ell^\T \bbf) \e_{\bbf}  \varepsilon_k I( | u_k | \leq \tau_k , u_\ell  < -\tau_\ell ) 
	+ (\tau_k - \bb_k^\T \bbf) \e_{\bbf} \varepsilon_\ell I( u_k > \tau_k , | u_\ell | \leq \tau_\ell ) \nn \\
	& \quad + (\tau_k - \bb_k^\T \bbf)(\tau_\ell - \bb_\ell^\T \bbf) \e_{\bbf}	I(  u_k   > \tau_k, u_\ell > \tau_\ell )
	 - ( \tau_k - \bb_k^\T \bbf )(\tau_\ell  + \bb_\ell^\T \bbf)  \e_{\bbf} I(  u_k   > \tau_k , u_\ell < - \tau_\ell ) \nn \\
	& \quad - (\tau_k  + \bb_k^\T \bbf )\e_{\bbf} \varepsilon_\ell I(  u_k < -\tau_k ,  | u_\ell  | \leq  \tau_\ell ) 
	- (\tau_k + \bb_k^\T \bbf)( \tau_\ell - \bb_\ell^\T \bbf ) \e_{\bbf} I( u_k < -\tau_k , u_\ell >\tau_\ell ) \nn \\
	& \quad + (\tau_k + \bb_k^\T \bbf)(\tau_\ell + \bb_\ell^\T \bbf) \e_{\bbf}I(  u_k <-\tau_k, u_\ell < -\tau_\ell ).  \label{cov.I.dec}
\end{align}
For the first term in \eqref{cov.I.dec}, note that
\begin{align}
	 & \e_{\bbf} \varepsilon_k \varepsilon_\ell I( |u_k| \leq \tau_k , |u_\ell | \leq  \tau_\ell  )    \nn \\
	 & = \cov( \varepsilon_k, \varepsilon_\ell  ) -  \e_{\bbf} \varepsilon_k \varepsilon_\ell I( |u_k| > \tau_k  ) - \e_{\bbf} \varepsilon_k \varepsilon_\ell I(  |u_\ell | > \tau_\ell ) + \e_{\bbf} \varepsilon_k \varepsilon_\ell I( |u_k | > \tau_k , |u_\ell | > \tau_\ell ) , \nn
\end{align}
where $ |  \e_{\bbf} \varepsilon_k \varepsilon_\ell I( |u_k | > \tau_k )  |  \leq   ( \tau_k - | \bb_k^\T \bbf | )^{-2} \e (  | \varepsilon_k |^{3} | \varepsilon_\ell |  ) \leq   4 \tau_k^{-2} v_k^{3/4}  v_\ell^{1/4}$ and
\begin{align}	
	&  | \e_{\bbf} \varepsilon_k \varepsilon_\ell I( |u_k| > \tau_k, |u_\ell | > \tau_\ell )  |   \leq (  \tau_k -  |\bb_k^\T \bbf | )^{-1}(  \tau_\ell - | \bb_\ell^\T \bbf | )^{-1}  \e  (  \varepsilon_k^2  \varepsilon_\ell^2   ) \leq  4 \tau_k^{-1} \tau_\ell^{-1}  v_k^{1/2} v_\ell^{1/2}  \nn
\end{align}
almost surely on $\cG_{k\ell}$. Hence,
\begin{align}
	| \e_{\bbf} \varepsilon_k \varepsilon_\ell  I ( |u_k| \leq \tau_k , |u_\ell  | \leq \tau_\ell  )  - \cov(\varepsilon_k, \varepsilon_\ell )  | \lesssim (\tau_k \wedge \tau_\ell )^{-2} \nn
\end{align}
holds almost surely on the same event. For the remaining terms in \eqref{cov.I.dec}, it can be similarly derived that, almost surely on the same event,
\begin{align}
	|  \e_{\bbf} \varepsilon_k I( |u_k | \leq  \tau_k , u_\ell > \tau_\ell )  | \leq  | \tau_\ell - \bb_\ell^\T \bbf |^{-3} \e ( | \varepsilon_k|  | \varepsilon_\ell |^3 ) , \nn \\
	  |  \e_{\bbf} \varepsilon_k I( |u_k | \leq \tau_k , u_\ell < - \tau_\ell )  |  \leq  | \tau_\ell +   \bb_\ell^\T \bbf |^{-3} \e( | \varepsilon_k| | \varepsilon_\ell |^{3}  ),  \nn \\
	\mbox{ and }~\e_{\bbf} I( u_k > \tau_k , \xi_\ell < -\tau_\ell ) \leq  | \tau_k - \bb_k^\T \bbf  |^{-2}   | \tau_\ell  + \bb_\ell^\T \bbf  |^{-2}   \e (  \varepsilon_k^2 \varepsilon_\ell^{2}  ) .\nn
\end{align}
Putting the pieces together, we arrive at $| \Pi_1 -\cov(\varepsilon_k, \varepsilon_\ell )  | \lesssim  (\tau_k \wedge \tau_\ell)^{-2} (v_k \vee v_\ell)$ almost surely on $\cG_{k\ell}$. This proves the stated result \eqref{approxi.cov}. 
\end{proof}

The next lemma provides several concentration results regarding the factors $\bbf_i$'s and their functionals.

\begin{lemma} \label{lem4}
Assume that (i) of Condition~\ref{cond:fdp} holds. Then, for any $t>0$,
\begin{align}
  \PP \bigg\{ \max_{1\leq i\leq n} \| \bbf_i \|_2 > C_1 C_f(r+\log n + t)^{1/2 } \bigg\} \leq e^{-t},  \label{max.norm.bound} \\
		\PP \{  \| \sqrt{n}\, \ol{\bbf} \|_2 > C_2  C_f (r+t)^{1/2} \} \leq e^{-t} , \label{factor.concentration.ineq} \\
		\mbox{ and }~\PP [  \|  \hat{\bSigma}_f - \Ib_r \|_2 >  \max \{   C_3C_f^2 n^{-1/2} ( r +t)^{1/2} , C_3^2 C_f^4 n^{-1} (r+t)  \} ] \leq  2e^{-t} , \label{factor.cov.concentration} 
\end{align}	
where $\ol{\bbf} = (1/n) \sn \bbf_i$, $\hat{\bSigma}_f = (1/n) \sn \bbf_i \bbf_i^\T$ and $C_1$--$C_3$ are absolute constants.
\end{lemma}

\begin{proof}[Proof of Lemma~\ref{lem4}]
Under (i) of Condition~\ref{cond:fdp}, it holds for any $\bu \in \RR^r$ that
\begin{align}
    \e e^{  \bu^\T \bbf_i   } \leq e^{C C_f^2 \| \bu\|_2^2 } ~\mbox{ for } i=1,\ldots, n , \nn \\
	\mbox{ and } \e e^{ \sqrt{n}\,\bu^\T  \ol{\bbf} } = \prod_{i=1}^n \e e^{ n^{-1/2} \bu^\T \bbf_i   } \leq \prod_{i=1}^n e^{ C C_f^2 n^{-1} \| \bu \|_2^2  } \leq   e^{ C C_f^2 \| \bu \|_2^2 } ,  \nn
\end{align}
where $C>0$ is an absolute constant. Using Theorem~2.1 in \cite{HKZ2012}, we derive that for any $t>0$,
\#
	\PP \{  \| \bbf_i \|_2^2 >  2C C_f^2 (r + 2\sqrt{rt} + 2t ) \} \leq e^{-t}   \mbox{ and } \PP \{ \|  \sqrt{n} \,\ol{\bbf} \|_2^2  >  2C C_f^2\big( r + 2\sqrt{ r  t} + 2t ) \}  \leq e^{-t}.  \nn 
\#
Then, \eqref{max.norm.bound} follows from the first inequality and the union bound, and the second inequality leads to \eqref{factor.concentration.ineq} immediately. Finally, applying Theorem~5.39 in \cite{V2012} gives \eqref{factor.cov.concentration}.
\end{proof}

\subsubsection{Proof of Theorem~\ref{thm:lln}}

First we introduce the following notations:
$$
 v_k = \EE(\varepsilon_k^4) , \quad \kappa_{\varepsilon, k} = v_k / \sigma_{\varepsilon, kk}^2 ,  \quad   u_{ik} = \bb_k^\T \bbf_i + \varepsilon_{ik} , \ \  k=1,\ldots, d, \, i=1,\ldots, n.
$$
Let $t \geq 1$ and set $\tau_k = a_k (n/t)^{1/2}$ with $a_k\geq \sigma_{kk}^{1/2}$ for $k=1,\ldots, d$. In view of Lemma~\ref{lem2}, define the event 
\#
	\cE_1(t) =  \bigcap_{k=1}^d   \bigg\{ \bigg|  \sqrt{n} \, (\hat{\mu}_k - \mu_k ) - \frac{1}{\sqrt{n}} \sn \psi_{\tau_k}(u_{ik}) \bigg|  \leq C \frac{a_k t }{\sqrt{n}} \bigg\} , \label{event1}
\#	
such that $\PP\{\cE_1(t)^{{\rm c}}\} \leq 3d e^{-t}$. Moreover, by Lemma~\ref{lem4}, let $\cE_2(t) $ be the event that the following hold:
\#
	 \max_{1\leq i\leq n } \| \bbf_i \|_2 \leq C_1 C_f (r+\log n + t)^{1/2} , \quad \| \sqrt{n} \, \ol \bbf \|_2\leq C_2 C_f (r+t)^{1/2} , \nn \\
\mbox{ and }~ \|  \hat{\bSigma}_f - \Ib_r \|_2 \leq \max\{    C_3 C_f^2 n^{-1/2} ( r +t)^{1/2} ,   C_3^2 C_f^4 n^{-1} (r+t)  \} . \label{cov.concentration}
\#
By the union bound, $\PP\{\cE_2(t)^{{\rm c}} \} \leq 4e^{-t}$.

Now we are ready to prove \eqref{FDP.lln}. The proof of \eqref{Rz.lln} follows the same argument, and thus is omitted.
For $k=1,\ldots, d$, define
\#
	B_k = \sqrt{n} \, \bb_k^\T \ol \bbf, \ \  V_k = \frac{1}{\sqrt{n}} \sn V_{ik} := \frac{1}{\sqrt{n}} \sn \{  \psi_{\tau_k}(u_{ik}) - \EE_{\bbf_i}   \psi_{\tau_k}(u_{ik}) \} , \label{Vk.def}
\#
and $R_k = n^{-1/2}\sn \{\EE_{\bbf_i}   \psi_{\tau_k}(u_{ik}) - \bb_k^\T \bbf_i \}$, where $\EE_{\bbf_i}(\cdot):= \EE(\cdot  | \bbf_i)$. On the event $\cE_1(t)$,
\#
	 | T_k - (V_k + B_k +R_k ) | \leq C a_k n^{-1/2} t ~\mbox{ for all } 1\leq k\leq d. \label{Tk.approxi}
\#
On $\cE_2(t)$, it holds $\max_{1\leq i\leq n} | \bb_k^\T \bbf_i | \leq C_1 C_f  \| \bb_k \|_2 (r+\log n + t)^{1/2} \leq C_1 C_f  \sigma_{kk}^{1/2} (r+\log n + t)^{1/2}$, which further implies 
\#
\max_{1\leq i\leq n} | \bb_k^\T \bbf_i | \leq \tau_k /2  ~\mbox{ for all } 1\leq k\leq d \label{factor.bound}
\#
as long as $n\geq 4(C_1 C_f)^2 (r+\log n + t) t$. Then, it follows from Lemma~\ref{lem3} that
\#
	| R_k  | \leq \sqrt{n}\max_{1\leq i\leq n} |\EE_{\bbf_i}   \psi_{\tau_k}(u_{ik}) - \bb_k^\T \bbf_i  |  \leq 8 \sqrt{n}\, \tau_k^{-3} v_k \leq  8 \sigma_{kk}^{-3/2} v_k   n^{-1} t^{3/2} \label{Rk.bound}
\#
holds  almost surely on $\cE_2(t)$ for all $1\leq k\leq d$. Together, \eqref{Tk.approxi} and \eqref{Rk.bound} imply that for any $z \geq 0$,
\#
	& \sum_{k \in \cH_0} I\Big( |V_k + B_k | \geq z + C a_k n^{-1/2} t  + 8\kappa_{\varepsilon, k} \sigma_{\varepsilon, kk}^{1/2}  n^{-1} t^{3/2}\Big) \nn \\
	& \leq V(z) \leq \sum_{k \in \cH_0} I\Big( |V_k + B_k | \geq z - C a_k n^{-1/2} t - 8 \kappa_{\varepsilon, k} \sigma_{\varepsilon, kk}^{1/2} n^{-1} t^{3/2}\Big) \label{Vz.bound}
\#
holds almost surely on $\cE_1(t) \cap \cE_2(t)$. In view of \eqref{Vz.bound}, we will instead deal with $\wt V_+(z)$ and $\wt V_-(z)$, where 
\#
	\wt V_+(z) :=  \sum_{k \in \cH_0} I(V_k \geq z - B_k)     ~\mbox{ and }~ \wt V_-(z) :=  \sum_{k \in \cH_0} I(V_k \leq  -z - B_k)  \nn 
\#
are such that $\wt V_+(z) + \wt V_-(z)= \sum_{k \in \cH_0} I(|V_k+B_k| \geq z)$.

In the following, we will focus on $\wt V_+(z)$ ($\wt V_-(z)$ can be dealt with in the same way). Observe that, conditional on $\cF_n:= \{ \bbf_i \}_{i=1}^n$, $I(V_1 \geq z- B_1),\ldots, I(V_d \geq z- B_d)$ are weakly correlated random variables. Define $Y_k=  I(V_k \geq z- B_k)$ and $P_k = \EE(Y_k | \cF_n)$ for $k=1,\ldots, d$. To prove the consistency of $\wt V_+(z)$, we calculate its variance:
\#	
 & \var\bigg( \frac{1}{d_0}  \sum_{k \in \cH_0} Y_k \bigg| \cF_n \bigg)  =   \EE \bigg[ \bigg\{ \frac{1}{d_0} \wt V_+(z) - \frac{1}{d_0} \sum_{k\in \cH_0} P_k \bigg\}^2 \bigg| \cF_n \bigg] \nn\\
 & =  \frac{1}{d_0^2}  \sum_{k \in \cH_0} \var(Y_k | \cF_n ) + \frac{1}{d_0^2} \sum_{ k, \ell \in \cH_0 : k\neq \ell } \cov(Y_k, Y_\ell | \cF_n )  \nn \\
 & \leq  \frac{1}{4 d_0 } + \frac{1 }{ d_0^2 } \sum_{ k, \ell \in \cH_0 : k\neq \ell }  \{ \EE(Y_k Y_\ell | \cF_n ) - P_k P_\ell  \} \label{var.bound}
\#
almost surely. In what follows, we study $P_k$ and $\EE(Y_k Y_\ell | \cF_n)$ separately, starting with the former. For each $k$, $V_{k}$ is a sum of  conditionally independent zero-mean random variables given $\cF_n$. Define 
$$
		\nu_k^2  = \var(V_k|\cF_n) = \frac{1}{n} \sn \nu_{ik}^2 ~\mbox{ with }~ \nu_{ik}^2 = \var(V_{ik}|\cF_n ).
$$
Then, it follows from the Berry-Esseen theorem that
\#
	& \sup_{z\in \RR}  |  \PP(V_k \leq \nu_k  x |\cF_n ) - \Phi(x) |  \nn \\
	&  \lesssim  \frac{1}{\nu_k^3 n^{3/2}} \sn \EE\{ | \psi_{\tau_k} (u_{ik}) |^3  | \cF_n \} \lesssim \frac{1}{\nu_k^3 n^{3/2}} \sn   (   | \bb_k^\T \bbf_i |^3 + \EE |\varepsilon_{ik} |^3 ) \label{BE.ineq}
\#
almost surely. By \eqref{factor.bound} and \eqref{approxi.var}, it holds
\#
  1  -   4  (  1 + \kappa_{\varepsilon, k} )   n^{-1 } t   \leq  \sigma_{\varepsilon, kk }^{-1}	\nu_k^2   \leq 1 \label{cond.var.bound}
\#
almost surely on $\cE_2(t)$ for all $1\leq k\leq d$. Combining this with \eqref{cov.concentration} and \eqref{BE.ineq} gives
\#
	\bigg|  P_k - \Phi\bigg(\frac{-z+B_k}{\nu_k }  \bigg)  \bigg| \lesssim   \kappa_{\varepsilon, k}^{3/4} \frac{1}{\sqrt{n}}  +   \sigma_{\varepsilon, kk}^{-3/2} \| \bb_k \|_2^3 \sqrt{\frac{r+\log n + t}{n}} \nn
\# 
almost surely on $\cE_2(t)$ for all $1\leq k\leq d$ as long as $n\geq 2C_3^2 C_f^4 (r+t) \vee 8(1+\kappa_{\varepsilon, \max} )t$, where $\kappa_{\varepsilon, \max} = \max_{1\leq \ell \leq d} \kappa_{\varepsilon, \ell}$. By the mean value theorem, there exists some $\eta_k \in [\sigma_{\varepsilon, kk}^{-1/2} , \nu_k^{-1}] $ such that
\#
	& \bigg| \Phi\bigg(\frac{-z+B_k}{\nu_k }  \bigg) - \Phi\bigg(\frac{-z+B_k}{ \sqrt{ \sigma_{\varepsilon , kk} } }  \bigg) \bigg|  = \phi( \eta_k | z-B_k| )   \frac{\eta_k | z-B_k | }{\eta_k}\bigg| \frac{1}{\nu_k} -  \frac{1}{\sqrt{\sigma_{\varepsilon,kk}}}\bigg|  \lesssim  \frac{ \kappa_{\varepsilon, k } t}{n} . \nn
\#
Together, the last two displays imply that almost surely on $\cE_2(t)$,
\#
	\bigg| P_k - \Phi\bigg(\frac{-z+B_k}{ \sqrt{ \sigma_{\varepsilon , kk} } }  \bigg)  \bigg|  \lesssim  \kappa_{\varepsilon, k} \bigg( \frac{1}{\sqrt{n}} + \frac{t}{n} \bigg) +  \sigma_{\varepsilon, kk}^{-3/2} \| \bb_k \|_2^3 \sqrt{\frac{r+\log n + t}{n}} \label{Pk.bound}
\#
uniformly for all $1\leq k\leq d$ and $z\geq 0$.

Next we consider $\EE(Y_k Y_\ell | \cF_n)= \PP(V_k\geq z-B_k, V_\ell \geq z-B_\ell | \cF_n)$. Define bivariate random vectors $\bV_i = (\nu_k^{-1} V_{ik} , \nu_\ell^{-1} V_{i \ell})^\T$ for $i=1,\ldots, n$, where $V_{ik}, V_{i \ell}$ are as in \eqref{Vk.def}. Observe that $\bV_1,\ldots, \bV_n$ are conditionally independent random vectors given $\cF_n$. Denote by $\bTheta = ( \theta_{uv})_{1\leq u,v \leq 2}$ the conditional covariance matrix of $n^{-1/2} \sn \bV_i = (\nu_k^{-1}V_k ,\nu_\ell^{-1} V_\ell)^\T$ given $\cF_n$, such that $\theta_{11} = \theta_{22} = 1$ and $\theta_{12} = \theta_{21} = (\nu_k \nu_\ell n)^{-1} \sn \cov_{\bbf_i}(V_{ik}, V_{i\ell})$. By \eqref{approxi.cov}, \eqref{factor.bound} and \eqref{cond.var.bound},
\#
	|  \theta_{12} -  r_{\varepsilon, k\ell} |  \lesssim ( \kappa_{\varepsilon, k} \vee \kappa_{\varepsilon, \ell} ) n^{-1} t \label{corr.bound}
\# 
holds almost surely on $\cE_2(t)$ for all $1\leq k \neq \ell \leq d$ and sufficient large $n$, say $n\gtrsim \kappa_{\varepsilon, \max} t$. Let $\bG= (G_1, G_2)^\T$ be a Gaussian random vector with $\EE(\bG) = \textbf{0}$ and $\cov(\bG) = \bTheta$. Applying Theorem~1.1 in \cite{B2005} and \eqref{cov.concentration}, we have
\#
	& 	\sup_{x, y \in \RR }   |  \PP( V_k \geq \nu_k x,  V_\ell \geq \nu_\ell y | \cF_n )  - \PP(G_1 \geq x , G_2 \geq y ) | \nn \\
& \lesssim  \frac{1}{n^{3/2}} \sn \EE \| \bTheta^{-1/2} \bV_i \|_2^3 \nn \\
&  \lesssim \frac{1}{ (\sigma_{\varepsilon, kk} n)^{3/2}}   \sn(   \EE |\varepsilon_{ik}|^3 + |\bb_k^\T \bbf_i|^3)   +  \frac{1}{ (\sigma_{\varepsilon, \ell \ell} n)^{3/2}}   \sn(   \EE |\varepsilon_{i\ell}|^3 + |\bb_\ell^\T \bbf_i|^3) \nn \\
& \lesssim   \frac{\kappa_{\varepsilon, k} + \kappa_{\varepsilon, \ell}}{\sqrt{n}} +  ( \sigma_{\varepsilon, kk}^{-3/2} \| \bb_k \|_2^3 + \sigma_{\varepsilon, \ell \ell}^{-3/2} \| \bb_\ell \|_2^3 ) \sqrt{\frac{r+\log n + t}{n}} \nn
\#
almost surely on $\cE_2(t)$ for all $1\leq k\neq \ell \leq d$. In particular, taking $x=\nu_k^{-1}(z-B_k)$ and $y=\nu_\ell^{-1}( z-B_\ell)$ gives
\#
 &   |  \EE(Y_k Y_\ell | \cF_n ) - \PP\{ G_1 \geq\nu_k^{-1}(z-B_k) , G_2 \geq \nu_\ell^{-1}(z-B_\ell) \} | \nn \\
& \lesssim \frac{\kappa_{\varepsilon, k} + \kappa_{\varepsilon, \ell}}{\sqrt{n}} +  ( \sigma_{\varepsilon, kk}^{-3/2} \| \bb_k \|_2^3 + \sigma_{\varepsilon, \ell \ell}^{-3/2} \| \bb_\ell \|_2^3 ) \sqrt{\frac{r+\log n + t}{n}}  \label{joint.prob.bound}
\#
almost surely on $\cE_2(t)$. In addition, it follows from Corollary~2.1 in \cite{LS2002} that
\#
	|   \PP(G_1 \geq x , G_2 \geq y ) - \{ 1-\Phi(x)\} \{1-\Phi(y)\} | \leq \frac{|\theta_{12}|}{4} e^{- (x^2+ y^2)/(2+2|\theta_{12}|)} \leq \frac{|\theta_{12}|}{4}   \label{gaussian.compare}
\#
for all $x, y \in \RR$.

Substituting the bounds \eqref{Pk.bound}, \eqref{corr.bound}, \eqref{joint.prob.bound} and \eqref{gaussian.compare} into \eqref{var.bound}, we obtain
\#
& \EE \bigg[ \bigg\{ \frac{1}{d_0} \wt V_+(z) - \frac{1}{d_0} \sum_{k\in \cH_0} P_k \bigg\}^2 \bigg| \cF_n \bigg]  \nn \\
& \lesssim \frac{1}{d_0^2} \sum_{k,\ell \in \cH_0: k \neq \ell} \varrho_{\varepsilon, k\ell} +   \frac{1}{d_0}  +    \frac{\kappa_{\varepsilon, \max}}{\sqrt{n}} +    \frac{\kappa_{\varepsilon, \max} t} {n} +   \max_{1\leq k\leq d} \sigma_{\varepsilon, kk}^{-3/2} \| \bb_k \|_2^3 \sqrt{\frac{r+\log n + t}{n}} \nn
\#
and
\#
 & \bigg| \frac{1}{d_0} \sum_{k\in \cH_0} P_k - \frac{1}{d_0} \sum_{k\in\cH_0} \Phi\bigg( \frac{-z+B_k}{\sqrt{\sigma_{\varepsilon, kk}}} \bigg) \bigg| \nn \\
 & \lesssim   \frac{\kappa_{\varepsilon, \max}}{\sqrt{n}} +    \frac{\kappa_{\varepsilon, \max} t} {n} +   \max_{1\leq k\leq d} \sigma_{\varepsilon, kk}^{-3/2} \| \bb_k \|_2^3 \sqrt{\frac{r+\log n + t}{n}}  \nn
\#
almost surely on $\cE_2(t)$. Similar bounds can be derived for 
$$
	\var\bigg( \frac{1}{d_0} \wt V_-(z) \bigg| \cF_n \bigg) ~\mbox{ and }~ \frac{1}{d_0} \EE\{ \wt V_-(z)| \cF_n \}.
$$
Taking $t=\log(nd)$ so that $\PP\{\cE_1(t)^{{\rm c}}\} \leq 3n^{-1}$ and $\PP\{ \cE_2(t)^{{\rm c}}\} \leq 4(nd)^{-1}$. Under Condition~\ref{cond:fdp}, it follows that
\#
	\frac{1}{d_0} \wt V_+(z) &= \frac{1}{d_0} \sum_{k\in\cH_0} \Phi\bigg( \frac{-z+B_k}{\sqrt{\sigma_{\varepsilon, kk}}} \bigg) + O_{\PP} [ d^{- (1\wedge \delta_0)/2} + n^{-1/4}\{\log(nd) \}^{1/4} ] ,\nn \\
	\frac{1}{d_0} \wt V_-(z) &= \frac{1}{d_0} \sum_{k\in\cH_0} \Phi\bigg( \frac{-z-B_k}{\sqrt{\sigma_{\varepsilon, kk}}} \bigg) + O_{\PP} [ d^{- (1\wedge \delta_0)/2} + n^{-1/4}\{\log(nd) \}^{1/4} ]  \nn
\#
uniformly over all $z\geq 0$. This, together with \eqref{Vz.bound} and the fact that $|\Phi(z_1)-\Phi(z_2)|\leq (2\pi)^{-1/2} |z_1-z_2|$, proves \eqref{FDP.lln}. \qed

\subsection{Proof of Theorem~\ref{thm:fdp}}

Let $\bb^{(1)}, \ldots,  \bb^{(r)} \in \RR^d$ be the columns of $\Bb$. Without loss of generality, assume that $\| \bb^{(1)} \|_2 \geq \cdots \geq \| \bb^{(r)} \|_2$. Under (i) of Condition~\ref{cond:fm}, $\Bb \Bb^\T$ has non-vanishing eigenvalues $\{ \overline{\lambda}_\ell := \| \bb^{(\ell)} \|_2^2  \}_{\ell=1}^r$ with eigenvectors $\{ \overline{\bv}_\ell :=\bb^{(\ell)} / \| \bb^{(\ell)} \|_2  \}_{\ell=1}^r$, and $\Bb = (\bb_1,\ldots,\bb_d)^\T=  ( \ol \lambda_1^{1/2} \ol\bv_1 , \ldots, \ol \lambda_r^{1/2} \ol\bv_r)$. Moreover, write $\hat{\Bb} = (\hat{\bb}_1, \ldots, \hat{\bb}_d )^\T = ( \hat{\lambda}^{1/2}_1 \hat{\bv}_1, \ldots,  \hat{\lambda}^{1/2}_r \hat{\bv}_r )$ with $\hat{\bv}_\ell = (\hat{v}_{\ell 1},\ldots, \hat{v}_{\ell d})^\T $ for $\ell =1 ,\ldots, r$ and $\hat{\bb}_k = ( \hat{\lambda}_{1}^{1/2} \hat{v}_{1k} , \ldots, \hat{\lambda}_{r}^{1/2} \hat{v}_{rk})^\T$ for $k=1,\ldots, d$. 

A key step in proving \eqref{FDP.consistency} is the derivation of an upper bound on the estimation error $\Delta_d := \max_{1\leq k\leq d} | \| \hat{\bb}_k \|_2 - \|  \bb_k \|_2 |$. By Weyl's inequality and the decomposition that $\hat{\bSigma}_1^{\HH}= \Bb \Bb^\T + (\hat{\bSigma}_1^{\HH} - \bSigma) + \bSigma_\varepsilon $,
\#
	\max_{1\leq \ell \leq r} | \hat{\lambda}_\ell - \ol \lambda_\ell | \leq  \| \hat{\bSigma}_1^{\HH} - \bSigma \|_2 + \| \bSigma_\varepsilon \|_2 . \label{eigenvalue.perturbation}
\#
Applying Corollary~1 in \cite{YWS2015} yields that, for every $1\leq \ell \leq r$,
$$
	\|  \hat{\bv}_\ell - \ol \bv_\ell \|_2 \leq \frac{2^{3/2} ( \| \hat{\bSigma}_1^{\HH} -\bSigma \|_2 + \|\bSigma_\varepsilon \|_2 ) }{\min( \ol \lambda_{\ell-1} - \ol \lambda_\ell , \ol \lambda_\ell - \ol \lambda_{\ell+1} )} ,
$$
where we set $\ol \lambda_0 = \infty$ and $\ol \lambda_{r+1} = 0$. Under (ii) of Condition~\ref{cond:fm}, it follows that
\#
	\max_{1\leq \ell \leq r} \|  \hat{\bv}_\ell  - \ol \bv_\ell  \|_2 \lesssim d^{-1} ( \| \hat{\bSigma}_1^{\HH} -\bSigma \|_2 + \|\bSigma_\varepsilon \|_2 ) . \label{l2.perturbation}
\#
Moreover, apply Theorem~3 and Proposition~3 in \cite{FWZ2016} to reach
\#
\max_{1\leq \ell \leq r} \|  \hat{\bv}_\ell  - \ol \bv_\ell \|_{\infty} \lesssim r^4  ( d^{-1/2} \| \hat{\bSigma}_1^{\HH} - \bSigma \|_{\max} +  d^{-1} \|\bSigma_\varepsilon \|_2 ) . \label{linfty.perturbation}
\#
Note that, under (ii) of Condition~\ref{cond:fm},
\#
	\|  \ol \bv_\ell \|_\infty = \|  \bb^{(\ell)} \|_\infty / \|  \bb^{(\ell)}\|_2 \leq \| \Bb \|_{\max}/\|  \bb^{(\ell)} \|_2 \lesssim d^{-1/2}  ~\mbox{ for all } \ell =1,\ldots, r. \label{barv.bound}
\#
Define $\wt  \bb_k = ( \ol \lambda_1^{1/2} \hat{v}_{1k} , \ldots ,  \ol \lambda_r^{1/2} \hat{v}_{rk} )^\T$. By the triangular inequality,
\#
	& \| \hat{\bb}_k - \bb_k \|_2 \leq  	\| \hat{\bb}_k - \wt  \bb_k \|_2 + 	\| \wt{\bb}_k - \bb_k \|_2 \nn \\
& = \bigg\{ \sum_{\ell =1}^r  ( \hat{\lambda}_\ell^{1/2} - \ol \lambda_\ell^{1/2} )^2 \,\hat{v}_{\ell k}^2   \bigg\}^{1/2} + \bigg\{  \sum_{\ell =1}^r  \ol \lambda_\ell (\hat{v}_{\ell k} - \ol v_{\ell k })^2   \bigg\}^{1/2} \nn \\
& \leq  r^{1/2} \bigg( \max_{1\leq \ell \leq r} |\hat{\lambda}_\ell^{1/2} - \ol \lambda_\ell^{1/2} | \| \hat{\bv}_\ell \|_\infty +  \max_{1\leq \ell \leq r}   \ol \lambda_\ell^{1/2} \| \hat{\bv}_\ell - \ol \bv_\ell \|_\infty \bigg). \nn
\#
This, together with \eqref{eigenvalue.perturbation}--\eqref{barv.bound} and Theorem~\ref{thm:huber-type}, implies
\#
	\Delta_d  \leq  \max_{1\leq k\leq d}   \| \hat{\bb}_k - \bb_k \|_2 = O_{\PP}(w_{n,d}).  \label{hatb.consistency}
\#

With the above preparations, now we are ready to prove \eqref{FDP.consistency}. To that end, define  $\wt \bu = \sqrt{n} \,(\Bb^\T \Bb)^{-1} \Bb^\T \ol{\bX}$, so that for every $1\leq k\leq d$,
$$
	\bb_k^\T \wt \bu =  \sqrt{n} \, \bb_k^\T \ol{\bbf}  + \sqrt{n} \, \bb_k^\T (\Bb^\T \Bb)^{-1} \Bb^\T \bmu + \sqrt{n} \, \bb_k^\T (\Bb^\T \Bb)^{-1} \Bb^\T \ol{\bvarepsilon} .
$$
Consider the decomposition 
\#
		& \bigg| \Phi\bigg(  \frac{ - z +  \hat \bb_k^\T \hat{\bu} }{\sqrt{ \hat  \sigma_{\varepsilon, kk} }}  \bigg) - \Phi\bigg(  \frac{ - z +  \sqrt{n} \, \bb_k^\T \ol{\bbf} }{\sqrt{  \sigma_{\varepsilon, kk} }}  \bigg) \bigg| \nn \\
	& \leq \bigg| \Phi\bigg(  \frac{ - z +  \hat \bb_k^\T \hat{\bu} }{\sqrt{ \hat  \sigma_{\varepsilon, kk} }}  \bigg) - \Phi\bigg(  \frac{ - z +   \bb_k^\T \wt{\bu} }{\sqrt{  \sigma_{\varepsilon, kk} }}  \bigg) \bigg|  +\bigg| \Phi\bigg(  \frac{ - z +  \bb_k^\T \wt{\bu} }{\sqrt{   \sigma_{\varepsilon, kk} }}  \bigg) - \Phi\bigg(  \frac{ - z +  \sqrt{n} \, \bb_k^\T \ol{\bbf} }{\sqrt{  \sigma_{\varepsilon, kk} }}  \bigg) \bigg|  \nn \\
& \leq   \bigg| \Phi\bigg(  \frac{ - z +    \bb_k^\T \wt{\bu} }{\sqrt{ \hat  \sigma_{\varepsilon, kk} }}  \bigg) - \Phi\bigg(  \frac{ - z +   \bb_k^\T \wt{\bu} }{\sqrt{   \sigma_{\varepsilon, kk} }}  \bigg) \bigg|    \nn  \\ 
& \quad + \bigg| \Phi\bigg(  \frac{ - z +  \hat \bb_k^\T \hat{\bu} }{\sqrt{ \hat  \sigma_{\varepsilon, kk} }}  \bigg) - \Phi\bigg(  \frac{ - z +   \bb_k^\T \wt{\bu} }{\sqrt{  \hat  \sigma_{\varepsilon, kk} }}  \bigg) \bigg|   +\bigg| \Phi\bigg(  \frac{ - z +  \bb_k^\T \wt{\bu} }{\sqrt{   \sigma_{\varepsilon, kk} }}  \bigg) - \Phi\bigg(  \frac{ - z +  \sqrt{n} \, \bb_k^\T \ol{\bbf} }{\sqrt{  \sigma_{\varepsilon, kk} }}  \bigg) \bigg|  \nn \\
& := \Delta_{k1} + \Delta_{k2} + \Delta_{k3}.  \label{decomposition}
\#
In the following, we deal with $\Delta_{k1}, \Delta_{k2}$ and $\Delta_{k3}$ separately.

By the mean value theorem, there exists some $\xi_k$ between $\hat{\sigma}_{\varepsilon, kk}^{-1/2}$ and $\sigma_{\varepsilon,kk}^{-1/2}$ such that
\#
& \Delta_{k1}  = \phi(  \xi_k | z-  \bb_k^\T \wt{\bu} | ) | z-  \bb_k^\T \wt{\bu} |    \bigg| \frac{1}{\sqrt{\hat{\sigma}_{\varepsilon, kk}}} - \frac{1}{\sqrt{ {\sigma}_{\varepsilon, kk}}}  \bigg| , \nn
\#
where $\phi(\cdot) = \Phi'(\cdot)$. With $\tau_{kk} \asymp \sqrt{n/\log(nd)}$ for $k=1,\ldots, d$, it follows that the event 
$$
 \mathcal{E}_0 = \bigg\{	\max_{1\leq k\leq d} | \hat{\sigma}_{kk} - \sigma_{kk} | \lesssim \sqrt{ \log(d)/n} \bigg\}
$$
satisfies $\PP(\mathcal{E}_0^{{\rm c}}) \lesssim  n^{-1}$. On $\mathcal{E}_0$, it holds $\hat{\sigma}_{\varepsilon, kk}^{-1} \geq (2\sigma_{kk} )^{-1}$, ${\sigma}_{\varepsilon, kk}^{-1} \geq \sigma_{kk}^{-1}$ and therefore $\xi_k \geq  (2\sigma_{kk} )^{-1/2}$ uniformly for all $1\leq k\leq d$ as long as $n \gtrsim \log d$. This further implies $\max_{1\leq k\leq d}\max_{z\geq 0} \phi(  \xi_k |z-   \bb_k^\T \wt{\bu} | ) | z- \bb_k^\T \wt{\bu} | = O_{\PP}(1)$. By \eqref{hatb.consistency},
\#
	  \bigg| \frac{1}{\sqrt{\hat{\sigma}_{\varepsilon, kk}}} - \frac{1}{\sqrt{ {\sigma}_{\varepsilon, kk}}}  \bigg|  = O_{\PP} (  | \hat{\sigma}_{kk} - \sigma_{kk} | + \| \hat{\bb}_k - \bb_k \|_2 ) = O_{\PP}(w_{n,d}) \label{var.consistency}
\#
uniformly over $k=1,\ldots, d$, where $w_{n,d} = \sqrt{ \log(d)/n} + d^{-1/2}$. Putting the above calculations together, we arrive at
\#
	\frac{1}{d} \sum_{k=1}^d \Delta_{k1} = O_{\PP}(w_{n,d}). \label{Deltak1.bound}
\#

Turning to $\Delta_{k2}$, again by the mean value theorem, there exists some $\eta_k$ between $\hat{\bb}_k^\T \hat{\bu}$ and $\bb_k^\T \wt \bu$ such that 
$$
	\Delta_{k2} = \phi\bigg( \frac{-z+ \eta_k }{\sqrt{\hat{\sigma}_{\varepsilon,kk}}} \bigg) \frac{\hat{\bb}_k^\T \hat{\bu } -  {\bb}_k^\T \wt {\bu} } {\sqrt{\hat{\sigma}_{\varepsilon, kk}}}.
$$
In view of \eqref{var.consistency}, 
$$
\max_{1\leq k\leq d} \phi\bigg( \frac{-z+ \eta_k }{\sqrt{\hat{\sigma}_{\varepsilon,kk}}} \bigg) \frac{1}{\sqrt{\hat{\sigma}_{\varepsilon, kk}}} = O_{\PP}(1) .
$$
Observe that $\hat \Bb \hat \bu  =  ( \sum_{\ell =1}^r \hat \bv_\ell \hat \bv_\ell^\T ) \bZ$ and $\Bb \wt \bu  =  ( \sum_{\ell=1}^r \ol \bv_\ell \ol\bv_\ell^\T ) \bZ$, where $\bZ= \sqrt{n}   \ol{\bX}$. By the Cauchy-Schwarz inequality,
\#
	& \sum_{k=1}^d |  \hat{\bb}_k^\T \hat{\bu } -  {\bb}_k^\T \wt {\bu} | \leq d^{1/2}   \| \hat \Bb \hat \bu - \Bb \wt \bu \|_2  \nn \\
& \leq d^{1/2} \bigg\| \sum_{\ell=1}^r  ( \hat \bv_\ell \hat \bv_\ell^\T -   \ol \bv_\ell \ol \bv_\ell^\T ) \bigg\|_2 \| \bZ \|_2 \leq 2  r d^{1/2}   \max_{1\leq \ell \leq r}\| \hat{\bv}_\ell - \ol \bv_\ell \|_2 \,\| \bZ \|_2 . \label{l1.norm.bound}
\#
For $\| \bZ \|_2$, we calculate $\EE \| \bZ \|_2^2  = n \| \bmu \|_2^2 + \sum_{k=1}^d \sigma_{kk}$, indicating that $\| \bZ \|_2 = O_{\PP}( \sqrt{n} \, \| \bmu \|_2 + d^{1/2})$. Combining this with \eqref{l2.perturbation}, \eqref{l1.norm.bound} and Theorem~\ref{thm:huber-type}, we conclude that
\#
	\frac{1}{d} \sum_{k=1}^d \Delta_{k2} =  O_{\PP} \{  ( d^{-1/2} \sqrt{n} \, \| \bmu \|_2 + 1 )  w_{n,d}  \}. \label{Deltak2.bound}
\#

For $\Delta_{k3}$, following the same arguments as above, it suffices to consider 
\#
 \frac{\sqrt{n}}{d}\sum_{k=1}^d | \bb_k^\T (\Bb^\T \Bb)^{-1} \Bb^\T  \bmu +  \bb_k^\T (\Bb^\T \Bb)^{-1} \Bb^\T \ol{\bvarepsilon}  | , \nn
\#
which, by the Cauchy-Schwarz inequality, is bounded by 
$$
	\sqrt{\frac{n}{d}} \, \bigg\| \sum_{\ell=1}^r \ol{\bv}_\ell \ol{\bv}_\ell^\T  \bigg\|_2    \| \bmu \|_2 +   \max_{1\leq k\leq d} \| \bb_k \|_2 \| \bu  \|_2 \leq \sqrt{\frac{n}{d}} \, \| \bmu \|_2 + \max_{1\leq k\leq d} \sqrt{ \sigma_{kk} } \,  \| \bu  \|_2 , $$	
where $\bu = \sqrt{n}\, (\Bb^\T \Bb)^{-1} \Bb^\T   \ol \bvarepsilon \in \RR^r$ is a zero-mean random vector with covariance matrix $\bSigma_u=(\Bb^\T \Bb)^{-1} \Bb^\T \bSigma_{\varepsilon} \Bb (\Bb^\T \Bb)^{-1}$. Recall that $\Bb^\T \Bb \in \RR^{r \times r}$ has non-increasing eigenvalues $\ol \lambda_1 \geq \cdots \geq \ol \lambda_r$. Under (ii) of Condition~\ref{cond:fm}, it holds $\EE \| \bu \|_2^2 = {\rm Tr}(\bSigma_u) \leq \|\bSigma_\varepsilon \| \sum_{\ell=1}^r \ol \lambda_\ell^{-1} \lesssim r d^{-1}$ and thus $\| \bu \|_2 = O_{\PP}(d^{-1/2})$. Putting the pieces together, we get
\#
	\frac{1}{d}\sum_{k=1}^d \Delta_{k3} = O_{\PP} ( d^{-1/2} \sqrt{n} \,\| \bmu \|_2 + d^{-1/2} ) .  \label{Deltak3.bound}
\#

Combining \eqref{decomposition}, \eqref{Deltak1.bound}, \eqref{Deltak2.bound} and \eqref{Deltak3.bound}, we reach
\#
 \frac{1}{d} \sum_{k=1}^d  \Phi\bigg(  \frac{ - z +  \hat \bb_k^\T \hat{\bu} }{\sqrt{ \hat  \sigma_{\varepsilon, kk} }}  \bigg)  =  \frac{1}{d} \sum_{k=1}^d \Phi\bigg(  \frac{ - z +  \sqrt{n} \, \bb_k^\T \ol{\bbf} }{\sqrt{  \sigma_{\varepsilon, kk} }}  \bigg)   + O_{\PP}  ( w_{n,d} +  d^{-1/2} \sqrt{n} \,\| \bmu \|_2 ). \nn
\#
Using the same argument, it can similarly derived that
\#
 \frac{1}{d} \sum_{k=1}^d  \Phi\bigg(  \frac{ - z -  \hat \bb_k^\T \hat{\bu} }{\sqrt{ \hat  \sigma_{\varepsilon, kk} }}  \bigg)  =  \frac{1}{d} \sum_{k=1}^d \Phi\bigg(  \frac{ - z-  \sqrt{n} \, \bb_k^\T \ol{\bbf} }{\sqrt{  \sigma_{\varepsilon, kk} }}  \bigg)   + O_{\PP}  ( w_{n,d} +  d^{-1/2} \sqrt{n} \,\| \bmu \|_2 ). \nn
\#
Together, the last two displays lead to the stated result \eqref{FDP.consistency}. \qed

\end{document}